\documentclass[11pt]{article}

\RequirePackage[margin=1in]{geometry}

\usepackage{hyperref}
\usepackage{amsfonts}
\usepackage{graphicx}
\usepackage{amsmath}
\usepackage{amsthm}
\usepackage{hhline}
\usepackage{multirow}
\usepackage{color}

\usepackage{amssymb}
\usepackage{float}
\usepackage{bbm}
\usepackage{algorithm}
\usepackage{algorithmic}
\usepackage{comment}
\usepackage{caption}
\usepackage{pgf}
\usepackage{tikz}
\usepackage{enumitem}
\usepackage{thmtools}
\usepackage{thm-restate}
\usepackage{cleveref}
\usepackage{subfigure}

\newcommand{\R}{\mathbb{R}}

\def\sample{{\mathcal{S}}}
\newcommand \expect {\mathbb{E}}

\newcommand \dmin {d_{\rm min}}

\def\V{{\mathcal{V}}}
\def\E{{\mathbb{E}}}
\def\seed{{\texttt{seed}}}
\def\clus{{\texttt{clus}}}
\def\lloyd{{\texttt{lloyds}}}
\def\lloyds{{(\alpha,\beta)\text{-Lloyds++}}}
\def\kmeans{{k\text{-means++}}}

\def\alo{{\alpha_{\ell}}}
\def\ahi{{\alpha_{h}}}

\let\citep\cite
\let\citet\cite

\newcommand \uniform {\operatorname{Uniform}}

\newtheorem{theorem}{Theorem}
\newtheorem{lemma}[theorem]{Lemma}

\newtheorem{definition}[theorem]{Definition}

\newenvironment{sproof}%
{%
 \par\noindent{\it Proof sketch. \ }%
}%
{\qed}

\newcommand{\cost}{{\textnormal{\texttt{cost}}}}
\newcommand \dratio {R}

\allowdisplaybreaks

\begin{document}

\title{Data-Driven Clustering via Parameterized Lloyd's Families\footnote{Authors' addresses: \texttt{\{ninamf,tdick,crwhite\}@cs.cmu.edu}.}}
 \author{Maria-Florina Balcan \and
 Travis Dick \and
Colin White}
\date{}

\maketitle

%\keywords{clustering, k-means, algorithm configuration}

\begin{abstract}

Algorithms for clustering points in metric spaces is a long-studied area of research.
Clustering has seen a multitude of work both theoretically, in understanding
the approximation guarantees possible for many objective
functions such as $k$-median and $k$-means clustering, and experimentally,
in finding the fastest algorithms and seeding procedures for Lloyd's algorithm.
The performance of a given clustering algorithm depends on the specific application at hand, and this may not be known up front.
For example, a ``typical instance'' may vary depending on the application,
and different clustering heuristics perform differently depending on the instance.

In this paper, we define an infinite family of algorithms generalizing Lloyd's algorithm, with one parameter controlling the initialization procedure,
and another parameter controlling the local search procedure.
This family of algorithms includes the celebrated $k$-means++ algorithm,
as well as the classic farthest-first traversal algorithm.
We design efficient learning algorithms which receive samples from an
application-specific distribution over clustering instances and learn a
near-optimal clustering algorithm from the class with respect to that
distribution with provable sample complexity guarantees.
%
% We design efficient learning algorithms which receive samples from an application-specific distribution over clustering instances
% and learn a near-optimal clustering algorithm from the class.
We show the best parameters vary significantly across application domains such as MNIST, CIFAR, and mixtures of Gaussians.
Our learned algorithms never perform worse than $k$-means++, and in some application domains we see significant improvements.

%Algorithms for clustering points in metric spaces is a long-studied area of research. Clustering has seen a multitude of work both theoretically, in understanding the approximation guarantees possible for many objective functions such as $k$-median and $k$-means clustering, and experimentally, in finding the fastest algorithms and seeding procedures for Lloyd's algorithm. The performance of a given clustering algorithm depends on the specific application at hand, and this may not be known up front. For example, a ``typical instance'' may vary depending on the application, and different clustering heuristics perform differently depending on the instance.

%In this paper, we define an infinite family of algorithms generalizing Lloyd's algorithm, with one parameter controlling the initialization procedure, and another parameter controlling the local search procedure. This family of algorithms includes the celebrated $k$-means++ algorithm, as well as the classic farthest-first traversal algorithm. We design efficient learning algorithms which receive samples from an application-specific distribution over clustering instances and learn a near-optimal clustering algorithm from the class. We show the best parameters vary significantly across datasets such as MNIST, CIFAR, and mixtures of Gaussians. Our learned algorithms never perform worse than $k$-means++, and on some datasets we see significant improvements.

\end{abstract}

\section{Introduction} \label{sec:intro}

Clustering is a fundamental problem in machine learning with applications in many areas including
text analysis, transportation networks, social networks, and so on.
The high-level goal of clustering is to divide a dataset into natural subgroups.
For example, in text analysis we may want to divide documents based on topic, and in social networks we might want to find communities.
A common approach to clustering is to set up an objective function
and then approximately find the optimal solution according to the objective.
There has been a wealth of both theoretical and empirical research in clustering using this approach
\cite{gonzalez1985clustering,charikar1999constant,arya2004local,arthur2007k,kaufman2009finding,
ostrovsky2012effectiveness,byrka2015improved,ahmadian2016better}.

%\cite{charikar1999constant,arya2004local,ostrovsky2012effectiveness,ahmadian2016better}. %paper

The most popular method in practice for clustering is local search, where
we start with $k$ centers and iteratively make incremental improvements until a local optimum is reached.
For example, Lloyd's method (sometimes called $k$-means) \cite{lloyd1982least}
and $k$-medoids \cite{friedman2001elements,cohen2016geometric} are two popular local search algorithms.
There are multiple decisions an algorithm designer must make when using a local search algorithm.
First, the algorithm designer must decide how to seed local search, e.g., how the algorithm chooses the $k$ initial centers.
There is a large body of work on seeding algorithms, since the initial choice of centers
can have a large effect on both the quality of the outputted clustering
and the time it takes for the algorithm to converge \cite{higgs1997experimental,pena1999empirical,arai2007hierarchical}.
The best seeding method often depends on the specific application at hand.
For example, a ``typical problem instance''
in one setting may have significantly different properties from that in another, causing some seeding methods to perform better than others.
Second, the algorithm designer must decide on an objective function for the local search phase ($k$-means, $k$-median, etc.)
For some applications, there is an obvious choice. For instance, if the application is Wi-Fi hotspot location,
then the explicit goal is to minimize the $k$-center objective function. %, and the cluster partitioning is irrelevant.
For many other applications such as clustering communities in a social network,
the goal is to find clusters which are close to an unknown target clustering,
and we may use an objective function for local search in the hopes that approximately minimizing the chosen objective
will produce clusterings which are close to matching the target clustering (in terms of the number of misclassified points).
As before, the best objective function for local search may depend on the specific application.

In this paper, we show positive theoretical and empirical results for learning
the best initialization and local search procedures over a large family of
algorithms. We take a transfer learning approach where we assume there is an
unknown distribution $\mathcal{D}$ over problem instances corresponding to our application,
and the goal is to use experience from the early instances to perform well on
the later instances. For example, if our application is clustering facilities in
a city, we would look at a sample of cities with existing optimally-placed
facilities, and use this information to find the empirically best seeding/local
search pair from an infinite family, and we use this pair to cluster facilities
in new cities.
% \textcolor{blue}{ Our assumption that clustering instances are
% drawn i.i.d. from some unknown distribution over clustering tasks is similar in
% spirit to the \emph{learning to learn} setting of~\cite{baxter1997bayesian},
% where they suppose that a learning system is faced with a sequence of learning
% tasks, each sampled from an unknown distribution over tasks (though their goal
% is to find a good Bayesian prior, rather than good algorithm parameters).}

\vspace{1em}\noindent\textbf{$\lloyds$.}
We define an infinite family of algorithms generalizing Lloyd's method, with two parameters $\alpha$ and $\beta$.
Our algorithms have two phases, a seeding phase to find $k$ initial centers (parameterized by $\alpha$),
and a local search phase which uses Lloyd's method to converge to a local optimum (parameterized by $\beta$).
In the seeding phase, each point $v$ is sampled with probability proportional to
$d_{\text{min}}(v,C)^\alpha$, where $C$ is the set of centers chosen so far and $d_{\text{min}}(v,C) = \min_{c \in C}d(v,c)$.
Then Lloyd's method is used to converge to a local minima for the $\ell_\beta$ objective.
By ranging $\alpha\in [0,\infty)\cup\{\infty\}$ and $\beta\in [1,\infty)\cup\{\infty\}$, we define our infinite family of algorithms which we call $\lloyds$.
Setting $\alpha = \beta = 2$ corresponds to the $\kmeans$ algorithm~\cite{arthur2007k}.
The seeding phase is a spectrum between random seeding ($\alpha=0$),
and farthest-first traversal \cite{gonzalez1985clustering,dasgupta2005performance} ($\alpha=\infty$),
and the Lloyd's step is able to optimize over common objectives including
$k$-median ($\beta=1$), $k$-means ($\beta=2$), and $k$-center ($\beta=\infty$).
We design efficient learning algorithms which receive samples from an application-specific distribution over clustering instances
and learn a near-optimal clustering algorithm from our family.

\vspace{1em}\noindent\textbf{Theoretical analysis.} In
Section~\ref{sec:alpha} we study both the sample and computational complexity of
learning the parameters for $\lloyds$ that have the lowest expected cost on the
application-specific distribution $\mathcal{D}$. The expected cost is over two
sources of randomness: the distribution $\mathcal{D}$ and the algorithmic
randomness during the seeding phase of $\lloyds$. To aid in our analysis, we
define an associated deterministic version of $\lloyds$ that takes as input a clustering
instance $\V$ and a vector $\vec{Z} = (z_1, \dots, z_k) \in [0,1]^k$, where the
value $z_i$ is used to deterministically choose the $i^{\rm th}$ center
during the seeding phase. When $\vec{Z}$ is sampled uniformly from $[0,1]^k$,
the distribution over outputs of the deterministic algorithm run with $\vec{Z}$
is identical to the randomized version of $\lloyds$. Our learning procedure
receives a sample $(\V_1, \vec{Z}_1), \dots, (\V_m, \vec{Z}_m)$ drawn i.i.d.
from $\mathcal{D} \times \uniform([0,1]^k)$ and returns the
parameters $\hat \alpha$ and $\hat \beta$ so that the deterministic algorithm
has the lowest average cost on the sample. First, we show that when the sample
size $m$ is sufficiently large, these parameters have approximately optimal cost
in expectation over both $\mathcal{D}$ and the internal randomness of $\lloyds$.
We also give efficient algorithms for finding the empirically optimal
parameters.

We prove that when the sample size is $m = \tilde O(k / \epsilon^2)$, where $k$
is the number of clusters and $\tilde O(\cdot)$ suppresses logarithmic terms,
the empirically optimal parameters $(\hat\alpha,\hat\beta)$ have expected cost
at most $\epsilon$ higher than the optimal parameters $(\alpha^*,\beta^*)$ over
the distribution, with high probability over the random sample. The key
challenge is that for any clustering instance $\V$ and vector
$\vec{Z} \in [0,1]^k$, the cost of the outputted clustering is not even a
continuous function of $\alpha$ or $\beta$ since a slight tweak in the
parameters may lead to a completely different run of the algorithm.
In fact, we show that for any clustering instance $\V$ and
vector $\vec{Z}$, the cost is a piecewise constant function of the parameters
$\alpha$ and $\beta$. The key step in our sample complexity guarantees is to
bound the number of discontinuities of the cost function. This requires a delicate
reasoning about the structure of ``decision points'', which are parameter values
where the algorithm output changes, each introducing a discontinuity in the cost
function. Our key technical contribution is to leverage the randomness over
$\vec{Z} \sim \uniform([0,1]^k)$ to prove polynomial bounds on the expected
number of decision points for the $\alpha$ parameter. By contrast, if we ignored
the distribution of $\vec{Z}$ and applied the techniques exploited by prior
work, we would only get exponential bounds on the number of decision points.

% We overcome this obstacle by showing a strong bound on the number of
% discontinuities of the cost as a function of the parameters in expectation over
% the random vector $Z \in [0,1]^k$.
%
% This analysis requires delicate reasoning
% about the structure of the ``decision points'' in the execution of the
% algorithm; in other words, for a given clustering instance, we must reason about
% the total number of outcomes the algorithm can produce over the full range of
% parameters. This allows us to use Rademacher complexity, a distribution-specific
% technique for achieving uniform convergence.

Next, we complement our sample complexity result with a computational efficiency result. Specifically,
we give a novel meta-algorithm which efficiently finds a near-optimal value $\hat\alpha$ with high probability.
The high-level idea of our algorithm is to run depth-first-search over the ``execution tree'' of the algorithm,
where a node in the tree represents a state of the
algorithm, and edges represent decision points.
A key step in our meta-algorithm is to iteratively solve for the decision points of the algorithm, which itself is nontrivial since the equations governing the decision points
do not have closed-form solutions. We show the equations have a certain structure which allows us to binary search through the range of parameters to find the decision points.

\vspace{1em}\noindent\textbf{Experiments.}
We give a thorough experimental analysis of our family of algorithms by
evaluating their performance on a number of different real-world and synthetic
application domains including MNIST, Cifar10, CNAE-9, and mixtures of Gaussians.
In each case, we create clustering instances by choosing subsets of the labels.
For example, we look at an instance of MNIST with digits $\{0,1,2,3,4\}$, and also an instance with digits $\{5,6,7,8,9\}$.
We show the optimal parameters transfer from one instance to the other.
Among domains, there is no single parameter setting that is nearly optimal,
and for some domains, the best algorithm from the $\lloyds$ family performs significantly better than
known algorithms such as $\kmeans$ and farthest-first traversal.

\section{Related Work} \label{sec:related_work}

\paragraph{Clustering.}
The iterative local search method for clustering, known as Lloyd's algorithm or sometimes called $k$-means, is one of the most popular
algorithms for $k$-means clustering \cite{lloyd1982least},
and improvements are still being found \cite{max1960quantizing,macqueen1967some,dempster1977maximum,pelleg1999accelerating,kanungo2002efficient,kaufman2009finding}.
The worst-case runtime of Lloyd's method is exponential \cite{arthur2006slow} even in $\mathbb{R}^2$ \cite{vattani2011k},
however, it converges very quickly in practice \cite{har2005fast}, and the smoothed complexity is polynomial \cite{arthur2011smoothed}.
Many different initialization approaches have been proposed \cite{higgs1997experimental,pena1999empirical,arai2007hierarchical}.
When using $d^2$-sampling to find the initial $k$ centers, the algorithm is known as $k$-means++,
and the approximation guarantee is provably $O(\log k)$ \cite{arthur2007k}.
If the data satisfies a natural stability condition, $k$-means++ returns a near-optimal clustering \cite{ostrovsky2012effectiveness}.
The farthest-first traversal algorithm is an iterative method to find $k$ centers,
and it was shown to give a 2-approximation algorithm for $k$-center \cite{gonzalez1985clustering},
and an 8-approximation for hierarchical $k$-center \cite{dasgupta2005performance}.

\paragraph{Transfer learning for unsupervised settings.}
Balcan et al.\ shows provable guarantees for learning over a different family of
algorithms, linkage-based clustering with dynamic pruning, in the same distribution
as the current work, however, they provide no experimental guarantees \cite{balcan2017learning}.
%A recent paper shows positive results for learning
%linkage-based algorithms with pruning over a distribution over clustering
%instances \cite{balcan2017learning}, although there is no empirical study done.
There are several related models for learning the best representation and
transfer learning for clustering.
For example, Ashtiani and Ben-David analyze the problem of learning a
near-optimal data embedding function from a given family of embeddings for the
$k$-means objective \cite{ashtiani2015representation}.

%
% Ashtiani and Ben-David show how to use a small
% labeled fraction of the dataset to learn a metric embedding of the entire
% dataset such that $k$-means performs well \cite{ashtiani2015representation}.
%
% \textcolor{blue}{Baxter studied situations where a supervised learner is faced
% with a sequence of learning tasks, each sampled i.i.d. from some unknown
% distribution over tasks \cite{baxter1997bayesian}. Their goal is to
% automatically identify a Bayesian prior that usefully biases the leaner for the
% task distribution. In contrast, we have a distribution over clustering instances
% and our goal is to find the algorithm parameters leading to the highest quality
% clusterings.}

There are a few models for the question of finding the best clustering algorithm to use on a single instance,
given a small amount of expert advice.
Ackerman et al.\ (building off of the celebrated clustering impossibility result of \cite{kleinberg2003impossibility}) study the problem of taxonomizing clustering algorithmic paradigms, by using a list of abstract properties of clustering functions \cite{ackerman2010towards}.
In their work, the goal is for a user to choose a clustering algorithm based
on the specific properties which are important for her application.

Another related area is the problem of unsupervised domain adaption. In this problem, the machine learning algorithm has access to a labeled
training dataset, and an unlabeled target dataset over a different distribution. The goal is to find an accurate classifier over the
target dataset, while only training on the training distribution~\cite{sener2016learning,ganin2014unsupervised,tzeng2014deep}.

There has been more research on related questions for transfer learning on unlabeled data and unsupervised tasks.
Raina et al.\ study transfer learning using unlabeled data, to a supervised learning task \cite{raina2007self}.
Jiang and Chung, and Yang et al.\ study transfer learning for clustering, in which a clustering algorithm has access to unlabeled data, and uses it to better cluster a related problem instance \cite{yang2009heterogeneous,jiang2012transfer}.
This setting is a bit different from ours, since we assume we have access to the target clustering for each training instance, but we tackle the harder question of finding the best clustering objective.

\section{Preliminaries} \label{sec:prelim}

\medskip\noindent \textbf{Clustering.}
A clustering instance $\V$ consists of a point set $V$ of size $n$, a distance metric $d$
(such as Euclidean distance in $\mathbb{R}^d$), and a desired number of clusters $1\leq k\leq n$.
A clustering  $\mathcal{C}=\{C_1,\dots,C_k\}$ is a $k$-partitioning of $V$.
Often in practice, clustering is carried out by approximately minimizing an objective function
(which maps each clustering to a nonzero value).
Common objective functions such as $k$-median and $k$-means come from the $\ell_p$ family,
where each cluster $C_i$ is assigned a center $c_i$ and
$\text{cost}(\mathcal{C})=\left(\sum_i\sum_{v\in C_i}d(v,c_i)^p\right)^\frac{1}{p}$
($k$-median and $k$-means correspond to $p=1$ and $p=2$, respectively).
There are two distinct goals for clustering depending on the application.
For some applications such as computing facility locations, the algorithm designer's
only goal is to find the best centers, and the actual partition $\{C_1,\dots,C_k\}$ is not needed.
For many other applications such as clustering documents by subject, clustering proteins by function,
or discovering underlying communities in a social network, there exists an unknown ``target'' clustering
$\mathcal{C}^*=\{C_1^*,\dots,C_k^*\}$, and the goal is to output a clustering $\mathcal{C}$
which is close to $\mathcal{C}^*$.
Formally, we define $\mathcal{C}$ and $\mathcal{C}'$ to be $\epsilon$-close if there exists a permutation $\sigma$ such that
$\sum_{i=1}^k |C_i\setminus C_{\sigma(i)}' |\leq \epsilon n$.
For these applications, the algorithm designer chooses an objective function while hoping that minimizing the objective
function will lead to a clustering that is close to the target clustering.
In this paper, we will focus on the cost function set to the distance to the target clustering, however, our analysis
holds for an abstract cost function $\cost$ which can be set to an objective function
or any other well-defined measure of cost.

\medskip\noindent \textbf{Algorithm Configuration.} In this work, we assume that
there exists an unknown, application-specific distribution $\mathcal{D}$ over a
set of clustering instances such that for each instance $\V$, $|V|\leq n$. We
suppose there is a cost function that measures the quality of a clustering of
each instance. As discussed in the previous paragraph, we can set the cost
function to be the expected Hamming distance of the returned clustering to the
target clustering, the cost of an $\ell_p$ objective, or any other function. The
learner's goal is to find the parameters $\alpha$ and $\beta$ that approximately
minimize the expected cost with respect to the distribution $\mathcal{D}$. Our
main technical results bound the intrinsic complexity of the class of $\lloyds$
clustering algorithms, which leads to generalization guarantees through standard
Rademacher complexity~\cite{bartlett2002rademacher,koltchinskii2001rademacher}.
This implies that the empirically optimal parameters are also nearly optimal in
expectation.

\section{\texorpdfstring{$\lloyds$}{(alpha,beta)-Lloyds++}} \label{sec:alpha}

In this section, we define an infinite family of algorithms generalizing Lloyd's algorithm, with one parameter controlling the initialization procedure,
and another parameter controlling the local search procedure.
Our main results bound the intrinsic complexity of this family of algorithms (Theorems \ref{thm:expect} and \ref{thm:pdim_upper})
and lead to sample complexity results guaranteeing the empirically optimal parameters over a sample are close to the optimal parameters over the unknown distribution.
We measure optimality in terms of agreement with the target clustering.
We also show theoretically that no parameters are optimal over all clustering applications (Theorem \ref{thm:alpha-just}).
Finally, we give an efficient algorithm for learning the best initialization parameter (Theorem \ref{thm:runtime}).

Our family of algorithms is parameterized by choices of $\alpha \in [0,\infty)\cup\{\infty\}$ and $\beta\in [1,\infty)\cup\{\infty\}$.
Each choice of $(\alpha,\beta)$ corresponds to one local search algorithm.
A summary of the algorithm is as follows. % (see Algorithm \ref{alg:clus}).
The algorithm has two phases. The goal of the first phase is to output $k$ initial centers. Each center
is iteratively chosen by picking a point with probability proportional to the minimum distance to all centers picked so far,
raised to the power of $\alpha$.
The second phase is an iterative two step procedure similar to Lloyd's method,
where the first step is to create a Voronoi partitioning of the points induced by the initial set of centers,
and then a new set of centers is chosen by computing the $\ell_\beta$ mean of the points in each Voronoi tile.

Our goal is to find parameters that return clusterings close to the
ground-truth in expectation. Setting $\alpha=\beta=2$ corresponds to the
$k$-means++ algorithm. The seeding phase is a spectrum between random seeding
($\alpha=0$), and farthest-first traversal ($\alpha=\infty$), and the Lloyd's
algorithm can optimize for common clustering objectives including $k$-median
($\beta=1$), $k$-means ($\beta=2$), and $k$-center ($\beta=\infty$).

On the way to proving our main results, we analyze a
deterministic version of $\lloyds$ that takes as input both a clustering
instance $\V$ and a vector $\vec{Z} = (z_1, \dots, z_k) \in [0,1]^k$. The
deterministic algorithm uses the value $z_t$ when choosing the $t^\text{th}$
center in the first phase of the algorithm. More specifically, the algorithm
chooses the $t^\text{th}$ center as follows: for each point $v_i \in V$ we
determine the $d^\alpha$-sampling probability of choosing $v_i$ as the next
center. Then, we construct a partition of $[0,1]$ into $|V|$ intervals, where
each point $v_i$ is associated with exactly one interval, and the width of the
interval is equal to the $d^\alpha$-sampling probability of choosing $v_i$.
Finally, the algorithm chooses the next center to be the point $v_i$ whose
interval contains the value $z_t$. When $z_t$ is drawn uniformly at random from
$[0,1]$, the probability of choosing the point $v_i$ to be the next center is
the width of its corresponding interval, which is the $d^\alpha$-sampling
probability. Therefore, for any fixed clustering instance $\V$, sampling the
vector $\vec{Z}$ uniformly at random from the cube $[0,1]^k$ and running this
algorithm on $\V$ and $\vec{Z}$ has the same output distribution as $\lloyds$.
Pseudocode for the deterministic version of $\lloyds$ is given in
Algorithm~\ref{alg:clus}.

\medskip\noindent\textbf{Notation.} Before presenting our
results, we introduce some convenient notation. We define cost functions for
both the randomized and deterministic versions of $\lloyds$. Given a clustering
instance $\V$ and a vector $\vec{Z} \in [0,1]^k$, we let
$\clus_{\alpha,\beta}(\V,\vec{Z})$ be the cost of the clustering output by
Algorithm~\ref{alg:clus} (i.e., the distance to the ground-truth clustering for
$\V$) when run on $(\V, \vec{Z})$ with parameters $\alpha$ and $\beta$. Since
the algorithm is deterministic, this is a well-defined function. Next, we also
define $\clus_{\alpha, \beta}(\V) = \expect_{\vec{Z}}[\clus_{\alpha, \beta}(\V,
\vec{Z})]$ to denote the expected cost of (randomized) $\lloyds$ run on instance
$\V$, where the expectation is taken over the randomness of the algorithm (i.e.,
over the draw of the vector $\vec{Z} \sim \uniform([0,1]^k)$). To
facilitate our analysis of phase 2 of Algorithm~\ref{alg:clus}, we let
$\lloyd_\beta(\V, C, T)$ denote the cost of the clustering obtained by running
Lloyd's algorithm with parameter $\beta$ starting from initial centers $C$ for
at most $T$ iterations on the instance $\V$.

\begin{algorithm}
\caption{Deterministic $\lloyds$ Clustering}\label{alg:clus}
\begin{algorithmic} %[1]
\STATE {\bfseries Input:} Instance $\mathcal{V} = (V,d,k)$, vector $\vec{Z} = (z_1, \dots, z_k) \in [0,1]^k$, parameters $\alpha$ and $\beta$.
\STATE  {\bfseries Phase 1: Choosing initial centers with $\mathbf{d^\alpha}$-sampling} \label{step:sample}
\begin{enumerate}[leftmargin=*,itemsep=0pt]
\item Initialize $C = \emptyset$.% and draw a vector $\vec{Z}=\{z_1,\dots,z_k\}$ from $[0,1]^k$ uniformly at random.
\item For each $t = 1, \dots, k$:
\begin{enumerate}[leftmargin=10pt,itemsep=0pt]
\item Partition $[0,1]$ into $n$ intervals, where there is an interval $I_{v_i}$ for each $v_i$
with size equal to the probability of choosing $v_i$ during $d^\alpha$-sampling in round $t$ (see Figure \ref{fig:alpha-interval}).
\item Denote $c_t$ as the point such that $z_t\in I_{c_t}$, and add $c_t$ to $C$.
\end{enumerate}
\end{enumerate}
\STATE {\bfseries Phase 2: Lloyd's algorithm}  \label{step:lloyd}
\begin{enumerate}[leftmargin=*,itemsep=0pt]
\setcounter{enumi}{4}
\item Set $C'=\emptyset$. Let $\{C_1,\dots,C_k\}$ denote the Voronoi tiling of $V$ induced by centers $C$.
\item Compute $\text{argmin}_{x\in V}\sum_{v\in C_i}d(x,v)^\beta$ for all $1\leq i\leq k$, and add it to $C'$.
\item If $C' \neq C$, set $C = C'$ and goto 5.
\end{enumerate}
\STATE {\bfseries Output:} Centers $C$ and clustering induced by $C$.
\end{algorithmic}
\end{algorithm}

When analyzing phase 1 of Algorithm~\ref{alg:clus}, we let
$\seed_\alpha(\V, \vec{Z})$ denote the vector of centers output by phase 1 when
run on a clustering instance $\V$ with vector $\vec{Z}$. For a given set of
centers $C$, we let $d_i$ denote the distance from point $v_i$ to the set $C$;
that is, $d_i = \min_{c \in C} d(v_i, c)$. For each point index $i$, we define
$D_i(\alpha) = \sum_{j=1}^i d_j^\alpha$, so that the probability of choosing
point $i$ as the next center under $d^\alpha$-sampling is equal to $d_i^\alpha /
D_n(\alpha)$, and the probability that the chosen index belongs to $\{1, \dots,
i\}$ is $D_i(\alpha)/D_n(\alpha)$.  When we use this notation, the set of
centers $C$ will always be clear from context. Finally, for a point set $V$, we
let $\dratio = \max \{ d(x,x') / d(y,y') \mid x,x',y,y' \in V, d(y,y') \neq 0\}$
denote the maximum ratio between any pair of non-zero distances in the point
set. The notation used throughout the paper is summarized in
Appendix~\ref{app:notation}.

We start with two structural results about the family of $\lloyds$ clustering
algorithms. The first shows that for sufficiently large $\alpha$, phase 1 of
Algorithm~\ref{alg:clus} is equivalent to farthest-first traversal. This means
that it is sufficient to consider $\alpha$ parameters in a bounded range.
%Throughout the paper, we give full details of the proof in the supplementary material.

Farthest-first traversal \cite{gonzalez1985clustering} starts by choosing a
random center, and then iteratively choosing the point farthest to all centers
chosen so far, until there are $k$ centers. We assume that ties are broken
uniformly at random. Farthest-first traversal is equivalent to the first phase
of Algorithm~\ref{alg:clus} when run with $\alpha = \infty$. The following
result guarantees that when $\alpha$ is sufficiently large,
Algorithm~\ref{alg:clus} chooses the same initial centers as farthest-first
traversal with high probability.

\begin{restatable}{lemma}{lemFarthestFirst} \label{lem:farthest-first}
  For any clustering instance $\V = (V,d,k)$ and $\delta > 0$, if $\alpha >
  \log\left(\frac{nk}{\delta}\right) / \log s $, where $s$ denotes the minimum
  ratio $d_1/d_2$ between two distances $d_1>d_2$ in the point set, then
  $P_{\vec{Z}}(\seed_\alpha(\V, \vec{Z}) = \seed_\infty(\V, \vec{Z})) \geq
  1-\delta$.
  %
% Given a clustering instance $\V$ and $\delta>0$,
% if $\alpha>\frac{\log\left(\frac{nk}{\delta}\right)}{\log s}$,
% then $d^\alpha$-sampling will choose the same centers as farthest-first traversal with probability $>1-\delta$.
% Here, $s$ denotes the minimum ratio $d_1/d_2$ between two distances $d_1>d_2$ in the point set.
\end{restatable}

For some datasets, $1/\log s$ might be very large. In Section
\ref{sec:experiments}, we empirically observe that for all datasets we tried,
$\lloyds$ behaves the same as farthest-first traversal for $\alpha>20$.
\footnote{ In Appendix \ref{app:theory}, we show that if the dataset satisfies a
stability assumption called separability \cite{kobren2017online,pruitt2011ncbi},
then $\lloyds$ outputs the same clustering as farthest-first traversal with high
probability when $\alpha>\log n$. % (Lemma \ref{lem:stability}).
}

Next, to motivate learning the best parameters, we show that for \emph{any} pair of parameters $(\alpha^*,\beta^*)$,
there exists a clustering instance such that $(\alpha^*,\beta^*)$-Lloyds++ outperforms
all other values of $\alpha,\beta$.
This implies that $d^\beta$-sampling is not always the best choice of seeding for
the $\ell_\beta$ objective.
% Let $\clus_{\alpha,\beta}(\V)$ denote the expected cost of the clustering outputted by $\lloyds$, with respect
% to the target clustering.
% Formally, $\clus_{\alpha,\beta}(\V)=\E_{\vec{Z}\sim [0,1]^k}\left[\clus_{\alpha,\beta}\left(\V,\vec{Z}\right)\right],$
% %\[\clus^*_{\alpha,\beta}(\V)=\E_{Z\sim [0,1]^k}\left[\clus_{\alpha,\beta}\left(\V,\vec{Z}\right)\right],\]
% where $\clus_{\alpha,\beta}\left(\V,\vec{Z}\right)$ is the cost of the clustering outputted by
% $\lloyds$ with randomness $\vec{Z}\in [0,1]^k$ (see line 1 of Algorithm \ref{alg:clus}).

\begin{theorem} \label{thm:alpha-just}
For $\alpha^*\in [.01,\infty)\cup\{\infty\}$ and $\beta^*\in [1,\infty)\cup\{\infty\}$, there exists a clustering instance $\V$
whose target clustering is the optimal $\ell_{\beta^*}$ clustering,
such that $\clus_{\alpha^*,\beta^*}(\V)<\clus_{\alpha,\beta}(\V)$ for all $(\alpha,\beta)\neq (\alpha^*,\beta^*)$.
\end{theorem}

\begin{proof}[Proof sketch]         %%%%%%%%%%%%%%%%%%%%%%%%%%%%%%%%%%%%%%%% sketch of Theorem \ref{thm:alpha-just}
Consider $\alpha^*,\beta^*\in [0,\infty)\cup\{\infty\}$.
The clustering instance consists of 6 clusters, $C_1,\dots,C_6$. These 6 clusters are both the target clustering for the instance, and they are optimal for the $\ell_{\beta^*}$ objective.
The proof consists of three sections.
First, we construct $C_1,\dots,C_4$ so that $d^{\alpha^*}$ sampling has the best chance of putting
exactly one point into each optimal cluster.
Then we add ``local minima traps'' to each cluster, so that if any cluster received two centers in the
sampling phase, Lloyd's method will not be able to move the centers to a different cluster.
Finally, we construct $C_5$ and $C_6$ so that if seeding put one point in each cluster,
then ${\beta^*}$-Lloyd's method will outperform any other $\beta\neq\beta^*$.

We construct the instance in an abstract metric space where we can define
pairwise distances to take any values, provided that they still satisfy the
triangle inequality. We refer to a collection of points as a clique if all
pairwise distances within the collection are equal.

 \begin{figure}[h]
 	\centering
  \includegraphics[width=0.9\linewidth]{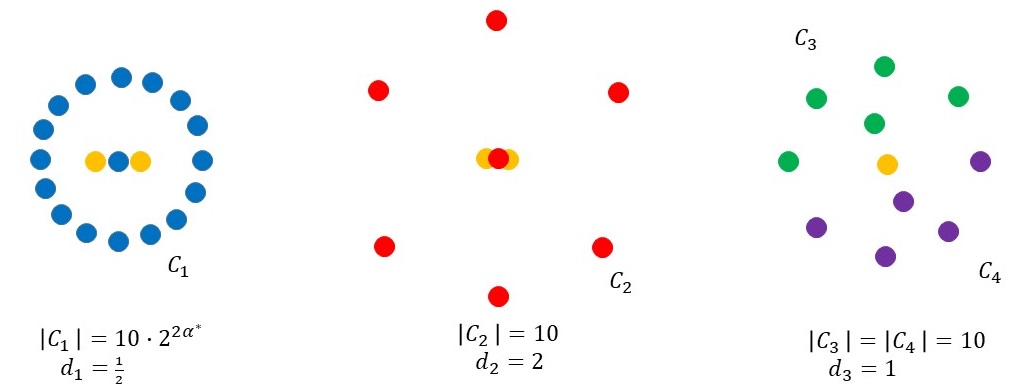}
 	\caption{Optimal instance for $d^{\alpha^*}$-sampling}
	\label{fig:alpha-just}
 \end{figure}

Step 1: we define $C_1$ and $C_2$ to be two different cliques, and we define a third clique whose points are equal to $C_3\cup C_4$.
%(see the supplementary material for a figure).
See Figure \ref{fig:alpha-just}.
We space these cliques arbitrarily far apart so that with high probability, the first three sampled centers will each be in a different clique.
Now the idea is to define the distances and sizes of the cliques so that $\alpha^*$ is the value of $\alpha$ with the
greatest chance of putting the last center into the third clique.
If we set the distances in cliques 1,2,3 to $d_1=2$, $d_2=1/2$, and $d_3=1$, and set $|C_1|=2^{2\alpha^*}|C_2|$,
then the probability of sampling a 4th center in the third clique for $\alpha=\alpha^*+\delta$ is equal to
$$\frac{|C_3\cup C_4|}{|C_3\cup C_4|+(2^{\alpha^*+\delta}+2^{\alpha^*-\delta})|C_2|}.$$
This is maximized when $\delta=0$.

Now we add local minima traps for Lloyd's method as follows.
In the first two cliques, we add three centers so that the 2-clustering cost is only slightly better than the 1-clustering cost.
In the third clique, which consists of $C_3\cup C_4$, add centers so that the 2-clustering cost is much lower than the 1-clustering cost.
We also show that since all cliques are far apart, it is not possible for a center to move between clusters during Lloyd's method.

Finally, we add three centers $c_5,b_5,b_5'$ to the last cluster $C_5$.
We set the rest of the points so that $c_5$ minimizes the $\ell_{\beta^*}$ objective, while $b_5$ and $b_5'$ favor
$\beta=\beta^*\pm \epsilon$. Therefore, $(\alpha^*,\beta^*)$ performs the best out of all pairs $(\alpha,\beta)$.
\end{proof}

%%%%%%%%%%%%%%%%%%%%%%%%%%%%%%%%%%%%%%%%%%%%%%%%%%%%%%%%%%%%%%%%%%%%%%%%%

\paragraph{Sample efficiency.} Now we give sample complexity bounds for learning the best algorithm from the class of $\lloyds$ algorithms.
We analyze the phases of Algorithm \ref{alg:clus} separately.
For the first phase, our main structural result is to show that for a given clustering instance,
with high probability over
the draw of $\vec{Z} \sim \uniform([0,1]^k)$,
%the randomness in Algorithm \ref{alg:clus},
the number of discontinuities of the function $\alpha \mapsto \seed_{\alpha}(\V,\vec{Z})$ as we vary $\alpha\in [\alo,\ahi]$
is $O\bigl(nk\log(n)\log(\ahi/\alo)\bigr)$.
Our analysis crucially harnesses the randomness of $\vec{Z}$
%in the algorithm
to achieve this bound.
For instance, if we ignore the distribution of $\vec{Z}$ and use a purely combinatorial approach as in prior algorithm configuration work, we would only achieve a bound of $n^{O(k)}$,
which is the total number of sets of $k$ centers.
For completeness, we give a combinatorial proof of $O(n^{k+3})$ discontinuities in Appendix \ref{app:theory} (Theorem \ref{thm:combinatorial}). Similarly, for the second phase of the algorithm, we show that for any clustering instance $\V$, initial set of centers $C$, and any maximum number of iterations $T$, the function $\beta \mapsto \lloyd_\beta(\V, C, T)$ has at most $O(\min(n^{3T}, n^{k+3}))$ discontinuities.

We begin by analyzing the number of discontinuities of the function $\alpha
\mapsto \seed_\alpha(\V, \vec{Z})$. Before proving the
$O\bigl(nk\log(n)\log(\ahi/\alo)\bigr)$ upper bound, we define a few concepts
used in the proof. Assume we start to run Algorithm \ref{alg:clus} without a
specific setting of $\alpha$, but rather a range $[\alo,\ahi]$, for some
instance $\V$ and vector $\vec{Z}$. In some round $t$, if Algorithm
\ref{alg:clus} would choose the same center $c_t$ for every setting of
$\alpha\in[\alo,\ahi]$, then we continue normally. However, if the algorithm
would choose a different center depending on the specific value of $\alpha$ used
from the interval $[\alo, \ahi]$, then we fork the algorithm, making one copy
for each possible next center. In particular, we partition $[\alo, \ahi]$ into a
finite number of sub-intervals such that the next center is constant on each
interval. The boundaries between these intervals are ``breakpoints'', since as
$\alpha$ crosses those values, the next center chosen by the algorithm changes.
Our goal is to bound the total number of breakpoints over all $k$ rounds in
phase 1 of Algorithm~\ref{alg:clus}, which bounds the number of discontinuities
of the function $\alpha \mapsto \seed_\alpha(\V, \vec{Z})$.

A crucial step in the above approach is determining where the breakpoints are
located. Recall that in round $t$ of Algorithm \ref{alg:clus}, each datapoint
$v_i$ is assigned an interval in $[0,1]$ of size $d_i^\alpha / D_n(\alpha)$,
where $d_i$ is the minimum distance from $v_i$ to the current set of centers,
and $D_j(\alpha)=d_1^\alpha+\cdots+d_j^\alpha$. The interval for point $v_i$ is
$\bigl[\frac{D_{i-1}(\alpha)}{D_n(\alpha)},
\frac{D_{i}(\alpha)}{D_n(\alpha)}\bigr)$ (see Figure \ref{fig:alpha-interval}).
WLOG, we assume that the algorithm sorts the points on each round so that
$d_1\geq\cdots\geq d_n$. We prove the following nice structure about these
intervals.

\begin{figure}[ht]
  \subfigure{\includegraphics[width=.5\linewidth]{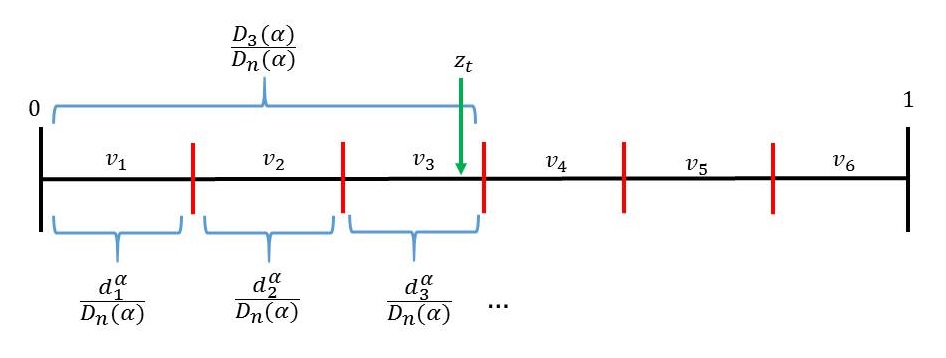}}
  \subfigure{\includegraphics[width=.5\linewidth]{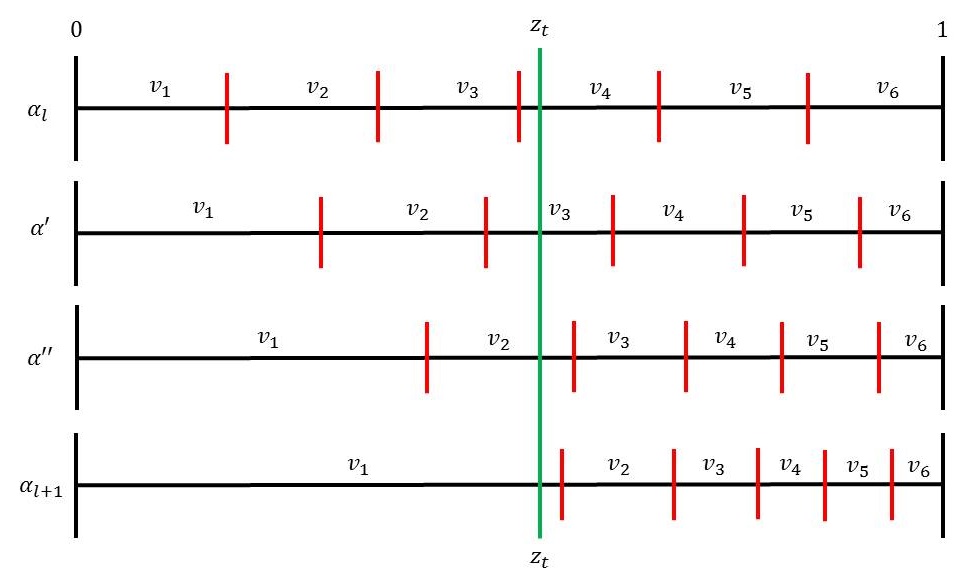}}
  \vspace{-1em}
  \caption{The algorithm chooses $v_3$ as a center (left). In the interval $[\alpha_\ell,\alpha_{\ell+1}]$, the algorithm may choose $v_4,v_3,v_2,$ or $v_1$ as a center,
	based on the value of $\alpha$ (right).}
  	\label{fig:alpha-interval}
\end{figure}

\begin{lemma} \label{lem:monotoneBins}
	Assume that $v_1$, \dots, $v_n$ are sorted in decreasing distance from a set $C$ of centers.
	Then for each $i=1,\dots,n$, the function
	$\alpha \mapsto\frac{D_i(\alpha)}{D_n(\alpha)}$ is monotone increasing and continuous
	along $[0,\infty)$.
	Furthermore, for all $1\leq i < j\leq n$ and $\alpha\in[0,\infty)$, we have $\frac{D_i(\alpha)}{D_n(\alpha)}\leq\frac{D_j(\alpha)}{D_n(\alpha)}$.
\end{lemma}

This lemma guarantees two crucial properties. First, we know that for every
(ordered) set $C$ of $t \leq k$ centers chosen by phase 1 of
Algorithm~\ref{alg:clus} up to round $t$, there is a single interval (as opposed
to a more complicated set) of $\alpha$-parameters that would give rise to $C$.
Second, for an interval $[\alo,\ahi]$, the set of possible next centers is
exactly $v_{i_{\ell}}, v_{i_{\ell} + 1}, \dots, v_{i_{h}}$, where $i_{\ell}$ and
$i_{h}$ are the centers sampled when $\alpha$ is $\alo$ and $\ahi$, respectively
(see Figure \ref{fig:alpha-interval}).

Now we are ready to prove our main structural result, which bounds the number of
discontinuities of the function $\alpha \mapsto \seed_\alpha(\V, \vec{Z})$ in
expectation over $\vec{Z} \sim \uniform([0,1]^k)$. We prove two versions of the
result: one that holds over parameter intervals $[\alo, \ahi]$ with $\alo > 0$,
and a version that holds when $\alo = 0$, but that depends on the largest ratio
of any pairwise distances in the dataset.
% Formally, we define $\seed_\alpha(\V,\vec{Z})$ as the outputted centers from phase 1 of Algorithm \ref{alg:clus} on instance $\V$ with randomness $\vec{Z}$.
%formally proven in the thm

\begin{restatable}{theorem}{thmExpect} \label{thm:expect}
  Fix any clustering instance $\V$ and let $\vec{Z} \sim \uniform([0,1]^k)$. Then:
  \begin{enumerate}
    \item For any parameter interval $[\alo, \ahi]$ with $\alo > 0$, the expected
    number of discontinuities of the function $\alpha \mapsto \seed_\alpha(\V, \vec{Z})$
    on $[\alo, \ahi]$ is at most $O(n k \log(n) \log(\ahi / \alo))$.
    \item For any parameter interval $[0, \ahi]$, the expected number of
    discontinuities of the function $\alpha \mapsto \seed_\alpha(\V, \vec{Z})$
    on $[0, \ahi]$ is at most $O\bigl(nk\log(n)\log(\ahi \log(\dratio))\bigr)$,
    where $R$ is the largest ratio between any pair of non-zero distances in
    $\V$.
  \end{enumerate}

  % For clustering instance $\V$ and interval $[\alo, \ahi]$ with $\alo > 0$, the
  % expected number of discontinuities of the function $\alpha \mapsto
  % \seed_\alpha(\V, \vec{Z})$ for $\alpha \in [\alo, \ahi]$ over the draw of
  % $\vec{Z} \sim \uniform([0,1]^k)$ is bounded by $O\left(nk \log
  % (n)\log(\ahi/\alo)\right)$. Additionally, the expected number of
  % discontinuities of $\alpha \mapsto \seed_\alpha(\V, \vec{Z})$ for $\alpha \in
  % [0, \ahi]$ is $O(nk\log(n)(\log \ahi + \log\log \dratio))$, where $\dratio$ is
  % the largest ratio of any non-zero pairwise distances in $\V$.
  %
% Given a clustering instance $\V$, the expected number of discontinuities
% of $\seed_\alpha(\V,\vec{Z})$ as a function of $\alpha$ over $[\alo,\ahi]$ is $O\left(nk \log (n)\log(\frac{\ahi}{\alo})\right)$.
% Over $[0,\ahi]$, the expected number of discontinuities is $O(nk\log(n)(\log \ahi + \log\log \dratio))$.
% Here, the expectation is over the uniformly random draw of $\vec{Z}\in [0,1]^k$ and $\dratio$ is the largest ratio of any non-zero pairwise distances in $\V$.
  %
\end{restatable}

\begin{proof}[Proof sketch] % of Theorem \ref{thm:expect}]
  Consider round $t$ in the run of phase 1 in Algorithm \ref{alg:clus} on
  instance $\V$ with vector $\vec{Z}$. Suppose at the beginning of round $t$,
  there are $L$ possible states of the algorithm; that is, $L$
  $\alpha$-intervals such that within each interval, the choice of the first
  $t-1$ centers is fixed. By Lemma \ref{lem:monotoneBins}, we can write these
  sets as $[\alpha_0,\alpha_1],\dots,[\alpha_{L-1},\alpha_L]$, where
  $0=\alpha_0<\cdots<\alpha_L=\ahi$. Given one interval,
  $[\alpha_{\ell},\alpha_{\ell+1}]$, we claim the expected number of new
  breakpoints, denoted by $\#I_{t,\ell}$, introduced by choosing a center in
  round $t$ starting from the state for interval $\ell$ is bounded by
  $$\min\left\{2n\log(\dratio) (\alpha_{\ell+1}-\alpha_\ell),\,\, n-t-1,\,\,
  4n\log (n)(\log\alpha_{\ell+1}-\log\alpha_\ell)\right\}.$$ Note that
  $\#I_{t,\ell}+1$ is the number of possible choices for the next center in
  round $t$ using $\alpha$ in $[\alpha_{\ell},\alpha_{\ell+1}]$.

  The claim gives three different upper bounds on the expected number of new
  breakpoints, where the expectation is only over $z_t \sim
  \uniform([0,1])$, and the bounds hold for any given
  configuration of $d_1\geq\cdots\geq d_n$ (i.e., it does not depend on the
  centers that have been chosen on prior rounds). To prove the first statement
  in Theorem~\ref{thm:expect}, we only need the last of the three bounds, and to
  prove the second statement, we need all three bounds.

  First we show how to prove the first part of the theorem assuming the claim,
  and later we will prove the claim. We prove the first statement as follows.
  Let $\#I$ denote the total number of discontinuities of $\alpha \mapsto
  \seed_\alpha(\V,\vec{Z})$ for $\alpha \in [\alo, \ahi]$. Then we have

  \begin{align*}
  \E_{\vec{Z}}[\#I]
  &\leq \E_{\vec{Z}} \left[\sum_{t=1}^k\sum_{\ell=1}^{L-1} (\#I_{t,\ell})\right]\\
  &= \sum_{t=1}^k \sum_{\ell=0}^{L-1} \E_{\vec{Z}}[\#I_{t,\ell}]\\
  &\leq \sum_{t=1}^k \sum_{\ell=0}^{L-1} 4n\log (n)\left(\log \alpha_{\ell+1}-\log\alpha_\ell\right) \\
  &= \sum_{t=1}^k 4n\log(n) \left(\log \ahi-\log\alo\right) \\
  % &\leq k \left(4n\log n \log \frac{\ahi}{\alo} \right)\\
  &= O\left(nk\log n\log \frac{\ahi}{\alo}\right)
  \end{align*}

  Now we prove the second part of the theorem. Let $\ell^*$ denote the largest
  value such that $\alpha_{\ell^*}<\frac{1}{\log \dratio}$. Such an $\ell^*$
  must exist because $\alpha_0=0$. Then we have $\alpha_{\ell^*}<\frac{1}{\log
  \dratio}\leq\alpha_{\ell^*+1}$. We use three upper bounds for three different
  cases of alpha intervals: the first $\ell^*$ intervals, interval
  $[\alpha_{\ell^*},\alpha_{\ell^*+1}]$, and intervals $\ell^*+2$ to $L$. Let
  $\#I$ denote the total number of discontinuities of $\alpha \mapsto
  \seed_\alpha(\V,\vec{Z})$ for $\alpha \in [0,\ahi]$.

  \begin{align*}
  \E_{\vec{Z}}[\#I]
  &\leq \E_{\vec{Z}} \left[\sum_{t=1}^k\sum_{\ell=1}^{L-1} (\#I_{t,\ell})\right]\\
  &=\sum_{t=1}^k \sum_{\ell=0}^{L-1} \E_{\vec{Z}}[\#I_{t,\ell}]\\
  &=\sum_{t=1}^k\left(
    \sum_{\ell=0}^{\ell^*-1} \E_{\vec{Z}}[\#I_{t,\ell}] +
    \E_{\vec{Z}}[\#I_{t,\ell^*}]
    + \sum_{\ell=\ell^*+1}^{L-1} \E_{\vec{Z}}[\#I_{t,\ell}]
  \right)\\
  &\leq \sum_{t=1}^k\left( \sum_{\ell=0}^{\ell^*-1} \left(2n\log \dratio (\alpha_{\ell+1}-\alpha_\ell)\right) + (n-t-1)
  + \sum_{\ell=\ell^*+1}^{L-1} \left(4n\log n(\log \alpha_{\ell+1}-\log\alpha_\ell)\right) \right)\\
  &\leq \sum_{t=1}^k\left( 2n\log \dratio \cdot \alpha_{\ell^*} + n + 4n\log n\left(\log \alpha_h-\log\alpha_{\ell^*}\right) \right)\\
  &\leq \sum_{t=1}^k\left( 2n\log(\dratio) \cdot \frac{1}{\log \dratio} + n + 4n\log(n)\left(\log \alpha_h-\log\left(\frac{1}{\log \dratio}\right)\right) \right)\\
  & = O\left(n k \log(n)(\log(\ahi \log (\dratio)) \right)
  % &\leq k \left( 2n+n + 4n\log n ( \log \alpha_h + \log\log \dratio)\right) \\
  % &\leq 4nk\log n(1 + \log \alpha_h + \log\log \dratio).
  \end{align*}

  Now we will prove the claim. Given $z_t\in[0,1]$, let $x$ and $y$ denote the
  minimum indices s.t.\ $\frac{D_x(\alpha_{\ell})}{D_n(\alpha_{\ell})}>z_t$ and
  $\frac{D_y(\alpha_{\ell+1})}{D_n(\alpha_{\ell+1})}>z_t$, respectively. Then
  from Lemma \ref{lem:monotoneBins}, the number of breakpoints for $\alpha \in
  [\alpha_\ell, \alpha_{\ell + 1}]$ is exactly $\#I_{t,\ell} = x-y$ (see Figure
  \ref{fig:case2}). Therefore, our goal is to compute $\E_{z_t}[x-y]$. One
  method is to sum up the expected number of breakpoints for each interval $I_v$
  by bounding the maximum possible number of breakpoints given that $z_t$ lands
  in $I_v$. However, this will sometimes lead to a bound that is too coarse. For
  example, if $\alpha_{\ell+1}-\alpha_\ell=\epsilon\approx 0$, then for each
  bucket $I_{v_j}$, the maximum number of breakpoints is 1, but we want to show
  the expected number of breakpoints is proportional to $\epsilon$. To tighten
  up this analysis, we will show that for each bucket, the probability (over
  $z_t$) of achieving the maximum number of breakpoints is low.

  Assuming that $z_t$ lands in a bucket $I_{v_j}$, we further break into cases
  as follows. Let $i$ denote the minimum index such that
  $\frac{D_i(\alpha_{\ell+1})}{D_n(\alpha_{\ell+1})}>\frac{D_j(\alpha_{\ell})}{D_n(\alpha_{\ell})}$.
  Note that $i$ is a function of $j,\alpha_\ell$, and $\alpha_{\ell+1}$, but it
  does not depend on $z_t$. If $z_t$ is less than
  $\frac{D_i(\alpha_{\ell+1})}{D_n(\alpha_{\ell+1})}$, then we have the maximum
  number of breakpoints possible, since the algorithm chooses center $v_{i-1}$
  when $\alpha=\alpha_{\ell+1}$ and it chooses center $v_j$ when
  $\alpha=\alpha_\ell$. The number of breakpoints is therefore $j-i+1$, by Lemma
  \ref{lem:monotoneBins}. We denote this event by $E_{t,j}$, i.e., $E_{t,j}$ is
  the event that in round $t$, $z_t$ lands in $I_{v_j}$ and is less than
  $\frac{D_i(\alpha_{\ell+1})}{D_n(\alpha_{\ell+1})}$. If $z_t$ is instead
  greater than $\frac{D_i(\alpha_{\ell+1})}{D_n(\alpha_{\ell+1})}$, then the
  algorithm chooses center $v_i$ when $\alpha=\alpha_{\ell+1}$, so the number of
  breakpoints is $\leq j-i$. We denote this event by $E_{t,j}'$ (see Figure
  \ref{fig:case2}). Note that $E_{t,j}$ and $E_{t,j}'$ are disjoint and
  $E_{t,j}\cup E_{t,j}'$ is the event that $z_t\in I_{v_j}$.

  Within an interval $I_{v_j}$, the expected number of breakpoints is
  $$P(E_{t,j})\cdot(j-i+1)+P(E_{t,j}')\cdot(j-i)=P(E_{t,j}\cup E'_{t,j})\cdot
  (j-i)+P(E_{t,j}).$$ We will bound $j-i$ and $P(E_{t,j})$ separately.
%\begin{figure}[ht]
%\includegraphics[width=.5\linewidth]{case2.jpg}
%  \caption{Cases 1 and 2 of Theorem \ref{thm:expect} (top). Details for case 1 of Theorem \ref{thm:expect} (bottom).}
%  	\label{fig:case2}
%\end{figure}

\begin{figure}[ht]
  \subfigure{\includegraphics[width=.5\linewidth]{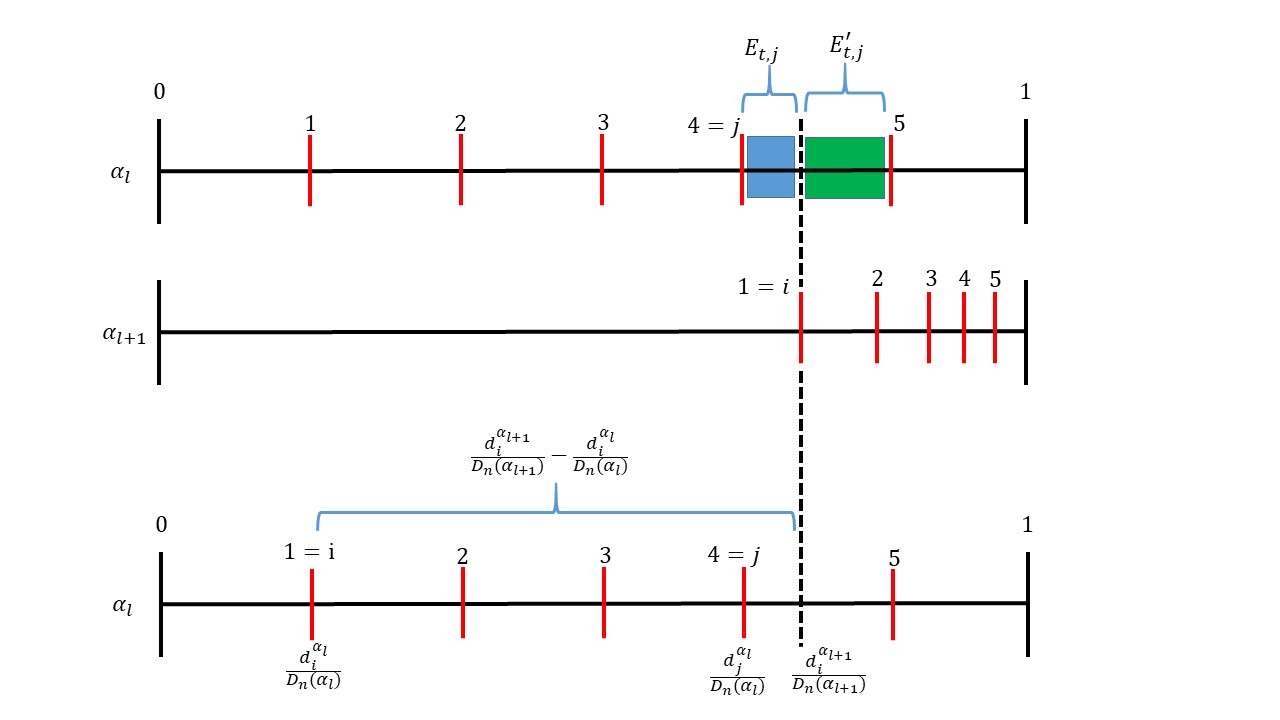}}
  \subfigure{\includegraphics[width=.5\linewidth]{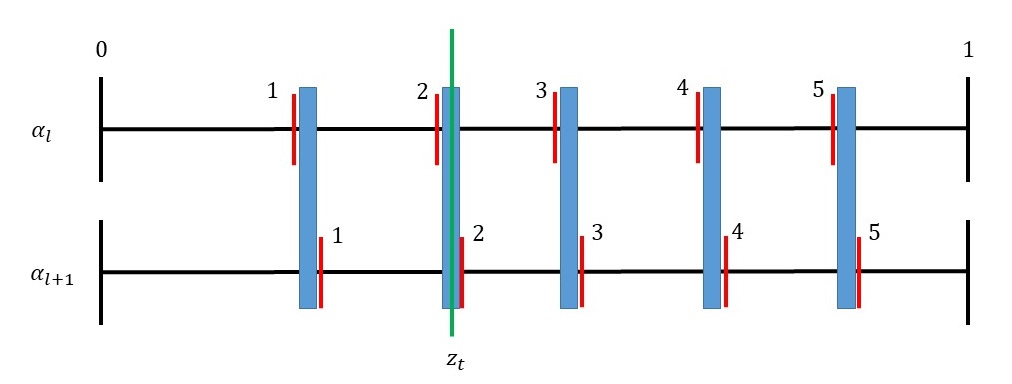}}
 % \vspace{-1em}
  \caption{Definition of $E_{t,j}$ and $E_{t,j}'$, and details for bounding $j-i$ (left).
	Intuition for bounding $P(E_{t,j})$, where the blue regions represent $E_{t,j}$ (right).
	}
  	\label{fig:case2}
\end{figure}

  First we upper bound $P(E_{t,j})$. Recall this is the probability that $z_t$
  is in between $\frac{D_j(\alpha_{\ell})}{D_n(\alpha_\ell)}$ and
  $\frac{D_i(\alpha_{\ell+1})}{D_n(\alpha_{\ell+1})}$, which is
  $$\frac{D_i(\alpha_{\ell+1})}{D_n(\alpha_{\ell+1})}-\frac{D_j(\alpha_{\ell})}{D_n(\alpha_\ell)}\leq\frac{D_j(\alpha_{\ell+1})}{D_n(\alpha_{\ell+1})}-\frac{D_j(\alpha_{\ell})}{D_n(\alpha_\ell)}.$$
  Therefore, we can bound this quantity by bounding the derivative
  $\left|\frac{\partial}{\partial\alpha}
  \left(\frac{D_j(\alpha)}{D_n(\alpha)}\right)\right|$, which we show is at most
  $\min\left\{\frac{2}{\alpha}\log n , \log \left(\dratio\right)\right\}  $ in
  Appendix \ref{app:theory}.

%\begin{figure}[ht]
%\includegraphics[width=.5\linewidth]{epsilon-interval.jpg}
%  \caption{Intuition for case 1 of Theorem \ref{thm:expect}.}
%  	\label{fig:case2}
%\end{figure}

  Now we upper bound $j-i$. Recall that $j-i$ represents the number of intervals
  between $\frac{D_i(\alpha_\ell)}{D_n(\alpha_{\ell})}$ and
  $\frac{D_j(\alpha_\ell)}{D_n(\alpha_\ell)}$ (see Figure \ref{fig:case2}). Note
  that the smallest interval in this range has width
  $\frac{d_j^{\alpha_\ell}}{D_n(\alpha_\ell)}$, and
  $$\frac{D_j(\alpha_\ell)}{D_n(\alpha_\ell)}-\frac{D_i(\alpha_\ell)}{D_n(\alpha_{\ell})}\leq
  \frac{D_i(\alpha_{\ell+1})}{D_n(\alpha_{\ell+1})}-\frac{D_i(\alpha_\ell)}{D_n(\alpha_{\ell})}.$$
  Again, we can use the derivative of $\frac{D_i(\alpha)}{D_n(\alpha)}$ to bound
  $\frac{D_i(\alpha_{\ell+1})}{D_n(\alpha_{\ell+1})}-\frac{D_i(\alpha_\ell)}{D_n(\alpha_{\ell})}.$
  To finish off the proof of the claim, we have

  \begin{align*}
  \E[\#I_{t,\ell}]
  &\leq\sum_j \left(P(E_{t,j}')\cdot (j-i)+P(E_{t,j})\cdot (j-i+1)\right)\\
  %&\leq \sum_j \left( P(E_{t,j}')\cdot (j-i)+ P(E_{t,j})\cdot(j-i)+P(E_{t,j})\right)\\
  &= \sum_j \left( P(E_{t,j}'\cup E_{t,j})\cdot (j-i)+ P(E_{t,j})\right)\\
  &= \sum_j P(z_t\in I_{v_j})\cdot (j-i) + \sum_j P(E_{t,j})\\
  &\leq \sum_j\left(\frac{d_j^{\alpha_\ell}}{D_n(\alpha_\ell)}\right)
  \left( \frac{D_n(\alpha_\ell)}{d_j^{\alpha_\ell}}\cdot \frac{D_j(\alpha)}{D_n(\alpha)}\bigg\rvert_{\alpha_\ell}^{\alpha_{\ell+1}}\right)
  +\sum_j \left( \frac{D_j(\alpha)}{D_n(\alpha)}\bigg\rvert_{\alpha_\ell}^{\alpha_{\ell+1}}\right)\\
  &\leq 2n \left( \frac{D_j(\alpha)}{D_n(\alpha)}\bigg\rvert_{\alpha_\ell}^{\alpha_{\ell+1}} \right)\\
  &\leq 2n \min\left( 2\log n(\log \alpha_{\ell+1}-\log \alpha_\ell) , \log D(\alpha_{\ell+1}-\alpha_\ell) \right)
  \end{align*}

  This accounts for two of the three upper bounds in our claim. To complete the
  proof, we note that $\E[\#I_{t,\ell}]\leq n-t-1$ simply because there are only
  $n-t$ centers available to be chosen in round $t$ of the algorithm (and
  therefore, $n-t-1$ breakpoints).
\end{proof}

%%%%%%%%%%%%%%%%%%%%%%%%%%%%%%%%%%%%%%%%%%%%%%%%%%%%%%%%%%%%%%%%%%%%%%%%%%%%%%%%%%%%%%%%%%%% pdim

Now we analyze phase 2 of Algorithm \ref{alg:clus}.
Since phase 2 does not have randomness, we use combinatorial techniques.
Recall that $\lloyd_\beta(\V,C,T)$ denotes the cost of the outputted clustering from phase 2 of Algorithm \ref{alg:clus} on instance $\V$ with initial centers $C$,
and a maximum of $T$ iterations.

\begin{theorem} \label{thm:pdim_upper}
Given $T\in\mathbb{N}$, a clustering instance $\V$, and a fixed set $C$ of initial centers,
the number of discontinuities of $\lloyd_\beta(\V,C,T)$ as a function of $\beta$ on instance $\V$ is $O(\min(n^{3T},n^{k+3}))$.
\end{theorem}

\begin{sproof}
Given $\V$ and a set of initial centers $C$, we bound the number of discontinuities introduced in the Lloyd's step of Algorithm \ref{alg:clus}.
First, we give a bound of $n^{k+3}$ which holds for any value of $T$.
Recall that Lloyd's algorithm is a two-step procedure, and note that the Voronoi partitioning step is independent of $\beta$.
Let $\{C_1,\dots,C_k\}$ denote the Voronoi partition of $V$ induced by $C$.
Given one of these clusters $C_i$, the next center is computed by $\min_{c\in C_i}\sum_{v\in C_i}d(c,v)^\beta$.
Given any $c_1,c_2\in C_i$, the decision for whether $c_1$ is a better center than $c_2$ is governed by
$\sum_{v\in C_i}d(c_1,v)^\beta<\sum_{v\in C_i}d(c_2,v)^\beta$.
By a consequence of Rolle's theorem, this equation has at most $2n+1$ roots.
This equation depends on the set $C$ of centers, and the two points $c_1$ and $c_2$,
therefore, there are ${n\choose k}\cdot {n\choose 2}$ equations each with $2n+1$ roots.
We conclude that there are $n^{k+3}$ total intervals of $\beta$ such that the outcome of Lloyd's method is fixed.

Next we give a different analysis which bounds the number of discontinuities by $n^{3T}$, where $T$ is the maximum number of Lloyd's iterations.
By the same analysis as the previous paragraph, if we only consider one round, then the total number of equations which govern the output
of a Lloyd's iteration is ${n\choose 2}$, since the set of centers $C$ is fixed. These equations have $2n+1$ roots, so the total number
of intervals in one round is $O(n^3)$. Therefore, over $T$ rounds, the number of intervals is $O(n^{3T})$.
\end{sproof}

By combining Theorem \ref{thm:expect} with Theorem \ref{thm:pdim_upper}, and
using standard learning theory results, we can bound the sample complexity
needed to learn near-optimal parameters $\alpha,\beta$ for an unknown
distribution $\mathcal{D}$ over clustering instances. Recall that
$\clus_{\alpha,\beta}(\V)$ denotes the expected cost of the clustering outputted
by $\lloyds$, with respect to the target clustering, taken over both the random
draw of a new instance and the algorithm randomness. Let $H$ denote an upper
bound on $\clus_{\alpha,\beta}(\V)$.

\begin{theorem} \label{thm:rademacher}
  Let $\mathcal{D}$ be a distribution over clustering instances with $k$
  clusters and at most $n$ points, and let $\sample = \bigl\{ (\V^{(i)},
  \vec{Z}^{(i)} \bigr\}_{i=1}^m$ be an i.i.d. sample from $\mathcal{D} \times
  \uniform([0,1]^k)$. For any parameters $\alpha, \beta$, let $L_S(\alpha,
  \beta) = \frac{1}{m} \sum_{i=1}^m \clus_{\alpha,\beta}(\V^{(i)},
  \vec{Z}^{(i)})$ denote the sample loss, and $L_{\mathcal{D}}(\alpha, \beta)
  = \E[\clus_{\alpha,\beta}(\V,\vec{Z})]$ denote the expected loss. The
  following statements hold:
  \begin{enumerate}[leftmargin=*]
    \item For any $\epsilon > 0$, $\delta > 0$, and parameter interval $[\alo,
    \ahi]$ with $\alo > 0$, if the sample is of size $m =
    O\left(\left(\frac{H}{\epsilon}\right)^2 \left(\min(T,k)\log n + \log k +
    \log\frac{1}{\delta} + \log(\log( \frac{\ahi}{\alo})) \right)\right)$
    then with probability at least $1-\delta$, for all $(\alpha, \beta) \in
    [\alo, \ahi] \times [1,\infty]$, we have $|L_S(\alpha, \beta) -
    L_{\mathcal{D}}(\alpha, \beta)| < \epsilon$.
    \item
    % Suppose that for any $\delta > 0$, with probability at least
    % $1-\delta$ the largest ratio of non-zero distances in a sample $\V \sim
    % \mathcal{D}$ is bounded by $\dratio(\delta)$.
    Suppose that the maximum ratio of any non-zero distances is bounded by
    $\dratio$ for all clustering instances in the support of $\mathcal{D}$. Then
    for any $\epsilon > 0$, $\delta > 0$, and parameter interval $[0, \ahi)$, if
    the sample size is $m = O\left(\left(\frac{H}{\epsilon}\right)^2
    \left(\min(T,k)\log n + \log k + \log\frac{1}{\delta} + \log(\log(\ahi
    \log(\dratio))) \right)\right)$, then with probability at least
    $1-\delta$, for all $(\alpha, \beta) \in [0,\ahi] \times [1,\infty]$, we
    have that $|L_S(\alpha, \beta) - L_{\mathcal{D}}(\alpha, \beta)| <
    \epsilon$.
  \end{enumerate}

  % Given parameters $0 < \alo < \ahi$, $\epsilon > 0$, $\delta > 0$, and a sample
  % $\sample = \bigl\{(\V^{(i)}, \vec{Z}^{(i)})\bigr\}_{i=1}^m$ of size
  % %
  % \[
  % m = O\left(
  % \left(\frac{H}{\epsilon}\right)^2 \left(\min(T,k)\log n + \log k + \log\frac{1}{\delta} + \log \log \frac{\ahi}{\alo} \right)\right)
  % \]
  % %
  % from $\left(\mathcal{D} \times \uniform([0,1]^k)\right)^m$, with
  % probability at least $1-\delta$ over the choice of the sample, for all
  % $\alpha\in [\alo,\ahi]$ and $\beta\in [1,\infty)\cup\{\infty\}$, we have
  % $\bigl|\frac{1}{m} \sum_{i = 1}^m \clus_{\alpha,\beta}(\V^{(i)},\vec{Z}^{(i)})
  % -\underset{\V \sim \mathcal{D}}{\E}\left[\clus_{\alpha,\beta}(\V)\right]\bigr|
  % < \epsilon.$

% Given $\ahi$
% and a sample $(\V^{(1)}, \vec{Z}^{(1)}), \dots, (\V^{(m)}, \vec{Z}^{(m)})$ of size
% $$m = O\left(\left(\frac{H}{\epsilon}\right)^2 \left(\min(T,k)\log n\log(\log \ahi + \log\log \dratio) + \log\frac{1}{\delta}\right)\right)$$ from $\left(\mathcal{D} \times \uniform([0,1]^k)\right)^m$,
% with probability at least $1-\delta$ over the choice of the sample, for all $\alpha\in [0,\ahi]$ and $\beta\in [1,\infty)\cup\{\infty\}$,
% $\bigl|\frac{1}{m} \sum_{i = 1}^m \clus_{\alpha,\beta}(\V^{(i)},\vec{Z}^{(i)}) -\underset{\V \sim \mathcal{D}}{\E}\left[\clus_{\alpha,\beta}(\V)\right]\bigr| < \epsilon.$
%\[\left|\frac{1}{m} \sum_{i = 1}^m \clus_{\alpha,\beta}\left(V^{(i)},\vec{Z}^{(i)}\right) -\underset{V \sim \mathcal{D}}{\E}\left[\clus_{\alpha,\beta}\left(V\right)\right]\right| < \epsilon.\]
\end{theorem}

\begin{proof}  %[Proof of Theorem \ref{thm:rademacher}]
  We begin by proving the first statement, which guarantees uniform convergence
  for parameters $(\alpha, \beta) \in [\alo, \ahi] \times [1,\infty]$.

  First, we argue that with probability at least $1 - \delta/2$, for all sample
  indices $i \in [m]$, the number of discontinuities of the function $\alpha
  \mapsto \seed_\alpha(\V^{(i)}, \vec{Z}^{(i)})$ for $\alpha \in [\alo, \ahi]$
  is $O(mnk \log(n) \log(\ahi/\alo) / \delta)$. By Theorem~\ref{thm:expect}, we
  know that for any sample index $i$, the expected number of discontinuities of
  $\alpha \mapsto \seed_\alpha(\V^{(i)}, \vec{Z}^{(i)})$ over the draw of
  $Z^{(i)} \sim \uniform([0,1]^d)$ is $O(nk\log(n)\log(\ahi/\alo))$. Applying
  Markov's inequality with failure probability $\delta / (2m)$, we have that
  with probability at least $1-\delta/(2m)$ over $\vec{Z}^{(i)}$, the function
  $\alpha \mapsto \seed_\alpha(\V^{(i)}, \vec{Z}^{(i)})$ has at most $O(mnk
  \log(n) \log(\ahi / \alo) / \delta)$ discontinuities. The claim follows by
  taking the union bound over all $m$ functions. We assume this high probability
  event holds for the rest of the proof.

  Next, we construct a set of $N = O(m^3nk \log(n) \log(\ahi / \alo)
  \min(n^{3T}, n^{k+3}) / \delta)$ parameter values $(\alpha_1, \beta_1), \dots,
  (\alpha_N, \beta_N)$ that exhibit all possible behaviors of the $\lloyds$
  algorithm family on the entire sample of clustering instances. Taking the
  union of all $O(m^2nk \log(n) \log(\ahi / \alo) / \delta)$ discontinuities of
  the functions $\alpha \mapsto \seed_\alpha(V^{(i)}, \vec{Z}^{(i)})$ for $i \in
  [m]$, we can partition $[\alo, \ahi]$ into $O(m^2nk \log(n) \log(\ahi / \alo)
  / \delta)$ intervals such that for each interval $I$ and any $\alpha, \alpha'
  \in I$, we have $\seed_\alpha(\V^{(i)}, \vec{Z}^{(i)}) =
  \seed_{\alpha'}(\V^{(i)}, \vec{Z}^{(i)})$ for all $i \in [m]$. In other words,
  on each interval and each sample instance, the initial centers chosen by phase
  1 of the algorithm is fixed. Now consider any interval $I$ in this partition.
  For each instance $(\V^{(i)}, \vec{Z}^{(i)})$ and any $\alpha \in I$, the set
  of initial centers chosen by phase 1 of the algorithm is fixed. Therefore,
  Theorem~\ref{thm:pdim_upper} guarantees that the number of discontinuities of
  the function $\beta \mapsto \lloyd_{\alpha, \beta}(\V^{(i)}, \vec{Z}^{(i)})$
  is at most $O(\min(n^{3T}, n^{k+3}))$. By a similar argument, it follows that
  we can partition the beta parameter space $[1,\infty]$ into $O(m\min(n^{3T},
  n^{k+3}))$ intervals such that for each interval the output clustering is
  constant for all instances (when run with any parameter $\alpha \in I$).
  Combined, it follows that we can partition the joint parameter space $[\alo,
  \ahi] \times [1,\infty]$ into $N = O(m^3nk \log(n) \log(\ahi / \alo)
  \min(n^{3T}, n^{k+3}) / \delta)$ rectangles such that for each rectangle and
  every instance $(\V^{(i)}, \vec{Z}^{(i)})$, the clustering output by $\lloyds$
  (and therefore the loss) is constant for all $(\alpha,\beta)$ values in the
  rectangle. Let $(\alpha_1, \beta_1), \dots, (\alpha_N, \beta_N)$ be a
  collection of parameter values obtained by taking one pair from each
  rectangle in the partition.

  Finally, to prove the uniform convergence guarantee, we bound the empirical
  Rademacher complexity of the family of loss functions $\mathcal{F} = \{ f_{\alpha,\beta} :
  (\V, \vec{Z}) \mapsto \clus_{\alpha,\beta}(\V, \vec{Z}) \mid \alpha \in
  [\alo, \ahi], \beta \in [1,\infty]\}$ on the given sample of instances. The
  empirical Rademacher complexity is defined by
  \[
    \hat R(\mathcal{F}, \sample)
    = \frac{1}{m} \E_\sigma \left[
      \sup_{f_{\alpha,\beta} \in \mathcal{F}} \sum_{i=1}^m \sigma_i f_{\alpha,\beta}\bigl(V^{(i)}, \vec{Z}^{(i)}\bigr)
    \right],
  \]
  where $\sigma$ is a vector of $m$ i.i.d. Rademacher random variables. The
  above arguments imply that we can replace the supremum over all of
  $\mathcal{F}$ by a supremum only over the loss functions with parameters in
  the finite set $\{(\alpha_j, \beta_j)\}_{j=1}^N$:
  \[
    \hat R(\mathcal{F}, \sample)
    = \frac{1}{m} \E_\sigma \left[
      \sup_{j \in [N]} \sum_{i=1}^m \sigma_i f_{\alpha_j,\beta_j}\bigl(V^{(i)}, \vec{Z}^{(i)}\bigr)
    \right].
  \]
  Define the vector $a^{(j)} = \bigl(f_{\alpha_j, \beta_j}(\V^{(1)}, \vec{Z}^{(1)}),
  \dots, f_{\alpha_j, \beta_j}(\V^{(m)}, \vec{Z}^{(m)})\bigr) \in [0,H]^m$ for
  all $j \in [N]$ and let $A = \{a^{(j)} \mid j \in [N]\}$. We have that
  $\Vert a_j\Vert_2 \leq H \sqrt{m}$ for all $j \in [N]$. Applying Massart's Lemma
  \citep{massart2000some} gives
  \[
    \hat R(\mathcal{F}, \sample)
    = \frac{1}{m} \E_\sigma\left[\sup_{a \in A} \sum_{i=1}^m \sigma_i a_i \right]
    \leq H \sqrt{2 \log(N) / m}.
  \]
  From this, the final sample complexity guarantee follows from standard
  Rademacher complexity bounds \cite{bartlett2002rademacher} using the remaining
  $\delta/2$ failure probability.

  The proof of the second statement follows a nearly identical argument. The
  only step that needs to be modified is the partitioning of the $\alpha$
  parameter space. In this case, we use the second statement from
  Theorem~\ref{thm:expect}, which guarantees that for each index $i \in [m]$,
  the expected number of discontinuities of the function $\alpha \mapsto
  \seed_\alpha(\V^{(i)}, \vec{Z}^{(i)})$ for $\alpha \in [0,\ahi]$ (over the
  draw of $\vec{Z} \sim \uniform([0,1]^k)$) is $O(n k \log(n) \log(\ahi
  \log(\dratio)))$. Applying Markov's inequality and the union bound, we have
  that with probability at least $1-\delta/2$, for all $i \in [m]$, the number
  of discontinuities of $\alpha \mapsto \seed_\alpha(\V^{(i)}, \vec{Z}^{(i)})$
  is $O(m n k \log(n) \log(\ahi \log(\dratio))/\delta)$. Now the rest of the
  argument is identical to the proof for the first statement, except the number
  of rectangles in the partition is $N = O(m^3nk \log(n) \log(\ahi
  \log(\dratio)) \min(n^{3T}, n^{k+3}) / \delta)$, which replaces $1/\alo$ by
  $\log(\dratio)$.

% Fix $\delta>0$.
% From Theorems \ref{thm:expect} and \ref{thm:pdim_upper}, for each $i$,
% the expected number of discontinuities of $\clus_{\alpha,\beta}(\V_i,Z_i)$ over $[0,\ahi]$ is
% at most $(8nk(\log n)(\log \ahi + \log\log \dratio)\min(n^k, n^{3T})$.
% By a Markov inequality,
% the number of discontinuities is $\leq\frac{1}{\delta}\cdot(16mnk\log n)(\log \ahi + \log\log \dratio)\min(n^k,n^{3T})$
% with probability $\geq 1-\frac{\delta}{2m}$.
% Therefore, with probability $\geq 1-\frac{\delta}{2}$, the total number of discontinuities of
% $\frac{1}{m}\sum_{i=1}^m\clus_\alpha(\V_i,Z_i)$
% is $\leq \frac{1}{\delta}\cdot (16m^2 nk\log n)(\log \ahi + \log\log \dratio)\min(n^k,n^{3T})$.
%
% Now we apply Massart's Lemma \cite{massart2000some}. %Lemma \ref{lem:massart}).
% Let $N\leq \frac{1}{\delta}\cdot(16m^2 nk\log n)(\log \ahi + \log\log \dratio)\min(n^k,n^{3T})$ denote the number of $\alpha$-intervals such that $\frac{1}{m}\sum_{i=1}^m\clus_\alpha(\V_i,Z_i)$ is constant along each interval.
% For each interval $1\leq i\leq N$, choose an arbitrary $\alpha_i$ in the interval,
% and define $a_i=\left[\clus_{\alpha_i}(\V_1,Z_i),\dots,\clus_{\alpha_N}(\V_m,Z_m)\right]$.
% Then with probability $\geq 1-\frac{\delta}{2}$, we can bound the empirical Rademacher complexity as follows
% \begin{equation*}
% \hat R(A)\leq\max_i ||a_i-\bar a||\cdot\frac{\sqrt{2\log N}}{m}\leq \sqrt{\frac{2\log \frac{16m^2 nk\log n(\log \ahi + \log\log \dratio)\min(n^k,n^{3T})}{\delta}}{m}},
% \end{equation*}
% and the proof follows from standard Rademacher complexity bounds \cite{bartlett2002rademacher}.
\end{proof}

% Note that a corollary of Theorem \ref{thm:rademacher} and Lemma \ref{lem:farthest-first} is a uniform convergence bound for \emph{all}
% $\alpha\in [0,\infty)\cup\{\infty\}$, however, the algorithm designer may decide to set $\ahi<\infty$.

\paragraph{Computational efficiency}
In this section, we present an algorithm for tuning $\alpha$ whose running time
scales with the true number of discontinuities over the sample. Combined with
Theorem~\ref{thm:expect}, this gives a bound on the expected running time of
tuning $\alpha$.

\begin{algorithm}[t]
\caption{Dynamic algorithm configuration} \label{alg:fast}
\begin{algorithmic} %[1]
\STATE {\bfseries Input:} Instance $\V = (V,d,k)$, vector $\vec{Z} \in [0,1]^k$, $\alo$, $\ahi$, $\epsilon>0$
\begin{enumerate}[leftmargin=*,itemsep=0pt]
  \item Initialize $Q$ to be an empty queue, then push the root node $(\langle\rangle, [\alo,\ahi])$ onto $Q$.
  \item While $Q$ is non-empty
  \begin{enumerate}[leftmargin=2ex,itemsep=0pt]
    \item Pop node $(C,A)$ from $Q$ with centers $C$ and alpha interval $A$.
		\item For each point $u_i$ that can be chosen as the next center,
		compute $A_i = \{\alpha\in A \,:\, \hbox{$u_i$ is the sampled center}\}$ up to error $\epsilon$
		and set $C_i=C\cup\{u_i\}$.
    \item For each $i$, if $|C_i| < k$, push $(C_i,A_i)$ onto $Q$. Otherwise, output $(C_i,A_i)$.
  \end{enumerate}
\end{enumerate}
\end{algorithmic}
\end{algorithm}

The high-level idea of our algorithm is to directly enumerate the set
of centers that can possibly be output by $d^\alpha$-sampling for a given
clustering instance $\V$ and pre-sampled randomness $\vec{Z}$.
We know from the previous section how to count the number of new breakpoints at
any given state in the algorithm, however, efficiently solving for the breakpoints poses a new challenge.
From the previous section, we know the breakpoints in $\alpha$ occur when
$\frac{D_i(\alpha)}{D_n(\alpha)}=z_t$. This is an exponential equation with $n$ terms,
and there is no closed-form solution for $\alpha$.
Although an arbitrary equation of this form may have up to $n$ solutions, our
key observation is that if $d_1\geq\cdots\geq d_n$, then $\frac{D_i(\alpha)}{D_n(\alpha)}$ must be monotone decreasing
(from Lemma \ref{lem:monotoneBins}), therefore, it suffices to binary search over
$\alpha$ to find the unique solution to this equation.
We cannot find the exact value of the breakpoint from binary search (and even if there was a closed-form solution for the breakpoint,
it might not be rational), however we can find the value to within additive error $\epsilon$ for all $\epsilon>0$.
%We know that the function $\frac{D_i(\alpha)}{D_n(\alpha)}$ is Lipschitz in $\alpha$ from Lemma~\ref{lem:deriv},
%therefore, it suffices to run binary search to find a solution whose expected cost is close to optimal.
Now we show that the cost function $\clus_{\alpha,\beta}(\V)$ is $(Hnk\log \dratio)$-Lipschitz in $\alpha$
for a constant-size interval,
therefore, it suffices to run $O\left(\log \frac{Hn\log \dratio}{\epsilon}\right)$ rounds of binary search to find a solution whose expected cost is within
$\epsilon$ of the optimal cost.
This motivates Algorithm \ref{alg:fast}.

\begin{lemma} \label{lem:net}
Given any clustering instance $\V$ with maximum non-zero distance ratio $\dratio$, $\epsilon>0$, and $\alpha\in (0,\infty)\cup\{\infty\}$,
$P_{\vec{Z}}(\seed_\alpha(\V,\vec{Z})\neq\seed_{\alpha+\epsilon}(\V, \vec{Z}))
\leq\min\left(2nk\log n\log\left(\frac{\alpha+\epsilon}{\alpha}\right) , \epsilon nk\log \dratio\right)$.
\end{lemma}

\begin{proof}
Given a clustering instance $\V$, $\epsilon>0$, we will show that, over the draw
of $\vec{Z} \sim \uniform([0,1]^k)$, there is low probability that
$\seed_\alpha(\V,\vec{Z})$ outputs a different set of centers than
$\seed_{\alpha+\epsilon}(\V,\vec{Z})$.
Assume in round $t$ of
$d^\alpha$-sampling and $d^{\alpha+\epsilon}$-sampling, both algorithms have $C$
as the current list of centers.
Given we draw $z_t\sim[0,1]$, we will show there
is only a small chance that the algorithms choose different centers in this
round.
Let $I_1, \dots, I_n$ be the intervals such that the next center chosen by the
algorithm with parameter $\alpha$ is $v_i$ whenever $z_t \in I_i$ and $I'_1, \dots, I'_n$
be the intervals for the algorithm with parameter $\alpha + \epsilon$.
First we will show that $P_{\vec{Z}}(\seed_\alpha(\V,\vec{Z})\neq\seed_{\alpha+\epsilon}(\V, \vec{Z}))
\leq\epsilon nk\log \dratio$.
Since the algorithms have an identical set of current centers, the
distances $d(v,C)$ are the same, but the breakpoints of the intervals,
$\frac{\sum_{j=1}^i d(v_j,C)^\alpha}{\sum_{j=1}^n d(v_j,C)^\alpha}$ differ
slightly. If $z_t\sim [0,1]$ lands in $I_i \cap I_i'$, the $d^\alpha$ and $d^{\alpha + \epsilon}$
will both choose $v_i$ as the next center. Thus, we need to bound the size of
$\sum_{i=1}^n (I_i\setminus I'_i)\cup (I'_i\setminus I_i)$. Recall the endpoint
of interval $i$ is $\frac{D_i(\alpha)}{D_n(\alpha)}$, where $D_i=\sum_{j=1}^i
d(v_j,C)^\alpha$. Thus, we want to bound
$\left|\frac{D_i(\alpha)}{D_n(\alpha)}-\frac{D_i(\alpha+\epsilon)}{D_n(\alpha+\epsilon)}\right|$,
and we can use Lemma \ref{lem:deriv}, which bounds the derivative of
$\frac{D_i(\alpha)}{D_n(\alpha)}$ by $\log \dratio$, to show
$\left|\frac{D_i(\alpha)}{D_n(\alpha)}-\frac{D_i(\alpha+\epsilon)}{D_n(\alpha+\epsilon)}\right|\leq
\epsilon \log \dratio$.

Therefore, we have
\begin{align*}
\sum_{i=1}^n (I_i\setminus I'_i)\cup (I'_i\setminus I_i)&\leq\sum_{i=1}^n \left|\frac{D_i(\alpha)}{D_n(\alpha)}-\frac{D_i(\alpha+\epsilon)}{D_n(\alpha+\epsilon)}\right| \\
&\leq\sum_{i=1}^n \epsilon\cdot \log \dratio\\
&\leq \epsilon n\log \dratio
\end{align*}

Therefore, assuming $d^\alpha$-sampling and $d^{\alpha+\epsilon}$-sampling have chosen the same centers so far,
the probability that they choose different centers in round $t$ is $\leq \epsilon n\log \dratio$.
Over all rounds, the probability the outputted set of centers is not identical, is $\leq \epsilon nk\log \dratio$.

Now we will show that $P_{\vec{Z}}(\seed_\alpha(\V,\vec{Z})\neq\seed_{\alpha+\epsilon}(\V, \vec{Z}))
\leq 2nk\log n\log\left(\frac{\alpha+\epsilon}{\alpha}\right)$.
We will bound
$\left|\frac{D_i(\alpha)}{D_n(\alpha)}-\frac{D_i(\alpha+\epsilon)}{D_n(\alpha+\epsilon)}\right|$
a different way, again using Lemma~\ref{lem:deriv}.
We have that
$$\frac{D_j(\alpha)}{D_n(\alpha)}\bigg\rvert_{\alpha_\ell}^{\alpha_{\ell+1}}\leq 2\log n\int_{\alpha_\ell}^{\alpha_{\ell+1}} \frac{1}{\alpha}\, d\alpha
\leq 2\log n\left(\log\alpha\right)\mid_{\alpha_\ell}^{\alpha_{\ell+1}}=2\log n(\log \alpha_{\ell+1}-\log \alpha_\ell).$$

Using this inequality and following the same steps as above, we have
\begin{align*}
\sum_{i=1}^n (I_i\setminus I'_i)\cup (I'_i\setminus I_i)&\leq\sum_{i=1}^n \left|\frac{D_i(\alpha)}{D_n(\alpha)}-\frac{D_i(\alpha+\epsilon)}{D_n(\alpha+\epsilon)}\right| \\
&\leq\sum_{i=1}^n 2\log n(\log \alpha_{\ell+1}-\log \alpha_\ell)\\
&\leq 2n\log n\log\left(\frac{\alpha+\epsilon}{\alpha}\right)
\end{align*}
Over all rounds, the probability the outputted set of centers is not identical,
is $\leq nk\log n \log\left(\frac{\alpha+\epsilon}{\alpha}\right)$.
This completes the proof.
\end{proof}

In order to analyze the runtime of Algorithm \ref{alg:fast},
we consider the  \emph{execution tree} of $d^\alpha$-sampling run on a clustering instance $\V$
with randomness $\vec{Z}$. This is a tree where each node is labeled by a state
(i.e., a sequence $C$ of up to $k$ centers chosen so far by the algorithm) and
the interval $A$ of $\alpha$ values that would result in the
algorithm choosing this sequence of centers. The children of a node correspond
to the states that are reachable in a single step (i.e., choosing the next
center) for some value of $\alpha \in A$. The tree has depth $k$, and there is
one leaf for each possible sequence of $k$ centers that $d^\alpha$-sampling will
output when run on $\V$ with randomness $\vec{Z}$.
Our algorithm enumerates these leaves up to an error $\epsilon$ in the $\alpha$ values,
by performing a depth-first traversal of the tree.

\begin{theorem} \label{thm:runtime}
  Let $\mathcal{D}$ be a distribution over clustering instances with $k$
  clusters and at most $n$ points and let $\sample = \{(\V^{(i)},
  \vec{Z}^{(i)})\}_{i=1}^m$ be an i.i.d. sample from $\mathcal{D} \times
  \uniform([0,1]^k)$. Fix any $\epsilon > 0$, $\delta > 0$, a parameter $\beta
  \in [1,\infty]$, and a parameter interval $[\alo, \ahi]$, and run
  Algorithm~\ref{alg:fast} on each sample instance and collect the breakpoints
  (boundaries between the intervals $A_i$). Let $\overline{\alpha}$ be the
  lowest cost breakpoint across all instances. Then the following statements hold:
  \begin{enumerate}[leftmargin=*]
    \item If $\alo > 0$ and the sample is of size $m =
    O\left(\left(\frac{H}{\epsilon}\right)^2 \left(\log (\frac{nk}{\delta})  +
    \log(\log( \frac{\ahi}{\alo})) \right)\right)$, then with probability
    at least $1-\delta$ we have
    $|\clus_{\bar\alpha,\beta}(\mathcal{S})-\min_{\alo\leq\alpha\leq\ahi}\clus_{\alpha,\beta}(\mathcal{S})|<\epsilon$
    and the total running time of finding the best breakpoint is $O\left(m n^2
    k^2 \log\left(\frac{\alo}{\ahi}\right)\log\left(
    nH\log\left(\frac{\alo}{\ahi}\right)/\epsilon\right) \right)$.
    \item If $\alo = 0$ but the maximum ratio of any non-zero distances for all
    instances in the support of $\mathcal{D}$ is bounded by $\dratio$ and the
    sample is of size $m = O\left(\left(\frac{H}{\epsilon}\right)^2 \left(\log(\frac{nk}{\delta}) + \log(\log(\ahi \log(\dratio)))
    \right)\right)$, then with probability at least $1-\delta$ we have
    $|\clus_{\bar\alpha,\beta}(\mathcal{S})-\min_{0
    \leq\alpha\leq\ahi}\clus_{\alpha,\beta}(\mathcal{S})|<\epsilon$ and the
    total running time of finding the best breakpoint is $O\left(m n^2 k^2 (\log(\ahi \log(\dratio)) \log\left(
    nH\log\left(\log(\ahi\log\dratio)\right)/\epsilon\right) \right)$.
  \end{enumerate}
\end{theorem}

% \begin{theorem} \label{thm:runtime}
% Given parameters $0<\alo<\ahi$, $\epsilon>0$, $\beta \geq 1$, and a sample
% $\mathcal{S}$ of size
% $$m = O\left(\left(\frac{H}{\epsilon}\right)^2 \left(\log (n)\log\log\left(\frac{\ahi}{\alo}\right)+
% \log\frac{1}{\delta}\right)\right)$$
% from
% $\left(\mathcal{D} \times \uniform([0,1]^k)\right)^m$, run Algorithm~\ref{alg:fast} on
% each sample and collect all breakpoints (i.e., boundaries of the intervals $A_i$) along $[\alo,\ahi]$.
% With probability at least $1-\delta$, the breakpoint $\bar\alpha$ with
% lowest empirical cost satisfies
% $|\clus_{\bar\alpha,\beta}(\mathcal{S})-\min_{0\leq\alpha\leq\ahi}\clus_{\alpha,\beta}(\mathcal{S})|<\epsilon$.
% The total running time to find the best breakpoint is
% $O\left(m n^2 k^2 \log (n)\log\left(\alo/\ahi\right)\log\left( nH\log\left(\frac{\alo}{\ahi}\right)/\epsilon\right) \right).$
% \end{theorem}

%If we collect all breakpoints along $[0,\ahi]$, we achieve the same result for
%$$m = O\left(\left(\frac{H}{\epsilon}\right)^2 \left(\log n\log(\log \ahi + \log\log \dratio) + \log\frac{1}{\delta}\right)\right)$$ with runtime
%$O\left(m n^2 k^2 \log (n)\log\log \dratio\log\left( \frac{nH\log\log\dratio}{\epsilon}\right) \right),$
%where $\dratio$ denotes the largest ratio of distances for any clustering instance in the sample.
%\end{theorem}

\begin{proof}
We argue that one of the breakpoints outputted by Algorithm~\ref{alg:fast}
on the sample is approximately optimal over all $\alpha\in[\alo,\ahi]$.
Formally, denote $\hat\alpha$ as the value with the lowest empirical cost over the sample,
and $\bar\alpha$ as the value with the lowest empirical cost over the sample,
among the set of breakpoints returned by the algorithm.
We define $\alpha^*$ as the value with the minimum true cost over the distribution.
We also claim that for all breakpoints $\alpha$, there exists a breakpoint $\hat\alpha$ outputted by Algorithm~\ref{alg:fast}
such that $|\alpha-\hat\alpha|<\frac{\epsilon}{5n^2k\log n\log\left(\frac{\ahi}{\alo}\right)}$.
We will prove this claim at the end of the proof.
Assuming the claim is correct, we denote $\alpha'$ as a breakpoint outputted by the algorithm such that
$|\hat\alpha-\alpha'|<\frac{\epsilon}{5n^2k\log n\log\left(\frac{\ahi}{\alo}\right)}$.

For the rest of the proof, denote
$\underset{\V \sim \mathcal{D}}{\E}\left[\clus_{\alpha,\beta}\left(\V\right)\right]=\text{true}(\alpha)$
and $\frac{1}{m} \sum_{i = 1}^m \clus_{\alpha,\beta}\left(\V^{(i)},\vec{Z}^{(i)}\right)=\text{sample}(\alpha)$
since beta, the distribution, and the sample are all fixed.

By construction, we have $\text{sample}(\hat\alpha)\leq \text{sample}(\alpha^*)$
and $\text{sample}(\bar\alpha)\leq \text{sample}(\alpha')$.
By Theorem~\ref{thm:rademacher}, with probability $>1-\delta$, for all $\alpha$ (in particular, for $\bar\alpha,~\hat\alpha,~\alpha^*,$ and $\alpha'$), we have
$\left|\text{sample}(\alpha)-\text{true}(\alpha)\right|<\epsilon/5$.
Finally, by Lemma~\ref{lem:net}, we have
$$|\hat\alpha-\alpha'|<\frac{\epsilon}{5n^2k\log n\log\left(\frac{\ahi}{\alo}\right)}~\implies
~\left|\text{true}(\hat\alpha)-\text{true}(\alpha')\right|<\epsilon/5.$$

Using these five inequalities for $\alpha',~\hat\alpha,~\bar\alpha,$ and $\alpha^*$, we can show the desired outcome as follows.
\begin{align*}
\text{true}(\bar\alpha)-\text{true}(\alpha^*)
&\leq\left(\text{true}(\bar\alpha)-\text{sample}(\bar\alpha)\right)+\text{sample}(\bar\alpha)
-\left(\text{true}(\alpha^*)-\text{sample}(\alpha^*)\right)-\text{sample}(\alpha^*)\\
&\leq\epsilon/5+\text{sample}(\alpha')+\epsilon/5-\text{sample}(\hat\alpha)\\
&\leq \left(\text{sample}(\alpha')-\text{true}(\alpha')\right)+\left(\text{true}(\alpha')-\text{true}(\hat\alpha)\right)
+\left(\text{true}(\hat\alpha)-\text{sample}(\hat\alpha)\right)+\frac{2\epsilon}{5}\\
&\leq\epsilon.
\end{align*}

\begin{comment}
\begin{align*}
\underset{V \sim \mathcal{D}}{\E}\left[\clus_{\bar\alpha,\beta}\left(V\right)\right]-\underset{V \sim \mathcal{D}}{\E}\left[\clus_{\alpha^*,\beta}\left(V\right)\right]
&\leq\left(\underset{V \sim \mathcal{D}}{\E}\left[\clus_{\bar\alpha,\beta}\left(V\right)\right]-\frac{1}{m} \sum_{i = 1}^m \clus_{\alpha,\beta}\left(V^{(i)},\vec{Z}^{(i)}\right)\right)
+\frac{1}{m} \sum_{i = 1}^m \clus_{\alpha',\beta}\left(V^{(i)},\vec{Z}^{(i)}\right)
-\left(\underset{V \sim \mathcal{D}}{\E}\left[\clus_{\alpha^*,\beta}\left(V\right)\right]-\frac{1}{m} \sum_{i = 1}^m \clus_{\alpha^*,\beta}\left(V^{(i)},\vec{Z}^{(i)}\right)\right)
-\frac{1}{m} \sum_{i = 1}^m \clus_{\alpha^*,\beta}\left(V^{(i)},\vec{Z}^{(i)}\right)\\
%
&\leq\frac{2\epsilon}{5}
+\left(\frac{1}{m} \sum_{i = 1}^m \clus_{\alpha',\beta}\left(V^{(i)},\vec{Z}^{(i)}\right)-\underset{V \sim \mathcal{D}}{\E}\left[\clus_{\alpha',\beta}\left(V\right)\right]\right)
+\left(\underset{V \sim \mathcal{D}}{\E}\left[\clus_{\alpha',\beta}\left(V\right)\right]-\underset{V \sim \mathcal{D}}{\E}\left[\clus_{\hat\alpha,\beta}\left(V\right)\right]\right)
+\left(\underset{V \sim \mathcal{D}}{\E}\left[\clus_{\hat\alpha,\beta}\left(V\right)\right]-\frac{1}{m} \sum_{i = 1}^m \clus_{\hat\alpha,\beta}\left(V^{(i)},\vec{Z}^{(i)}\right)\right)\\
&\leq\epsilon.
\end{align*}
\end{comment}

Now we will prove the claim that for all breakpoints $\alpha$, there exists a breakpoint $\hat\alpha$ outputted by Algorithm~\ref{alg:fast}
such that $|\alpha-\hat\alpha|<\frac{\epsilon}{5n^2k\log n\log\left(\frac{\ahi}{\alo}\right)}$. Denote $\epsilon'=\frac{\epsilon}{5n^2k\log n\log\left(\frac{\ahi}{\alo}\right)}$.
We give an inductive proof.
Recall that the algorithm may only find the values of breakpoints up to additive error $\epsilon'$, since the true breakpoints may be
irrational and/or transcendental.
Let $\hat T_t$ denote the execution tree of the algorithm after round $t$,
and let $T_t$ denote the \emph{true} execution tree on the sample.
That is, $T_t$ is the execution tree as defined earlier this section,
$\hat T_t$ is the execution tree with the algorithm's $\epsilon'$ imprecision on the values of alpha.
Note that if a node in $T_t$ represents an $\alpha$-interval of size smaller than $\epsilon'$, it is possible that $\hat T_t$
does not contain the node. Furthermore, $\hat T_t$ might contain spurious nodes with alpha-intervals of size smaller than $\epsilon'$.

Our inductive hypothesis has two parts.
The first part is that for each breakpoint $\alpha$ in $T_t$, there exists
a breakpoint $h(\alpha)$ in $\hat T_t$ such that $|\alpha-h(\alpha)|<\epsilon'$.
For the second part of our inductive hypothesis, we define
$B_t=\bigcup_{\alpha\text{ breakpoint}}\left([\alpha,h(\alpha)]\cup[h(\alpha),\alpha]\right)$,
the set of ``bad'' intervals. Note that for each $\alpha$, one of $[\alpha,h(\alpha)]$ and $[h(\alpha),\alpha]$ is empty.
Then define $G_t=[\alo,\ahi]\setminus B_t$, the set of ``good'' intervals.
The second part of our inductive hypothesis is that the set of centers for $\alpha$ in $T_t$ is the same as in $\hat T_t$, as long as $\alpha\in G_t$.
That is, if we look at the leaf in $T_t$ and the leaf in $\hat T_t$ whose alpha-intervals contain $\alpha$,
the set of centers for both leaves are identical.
Now we will prove the inductive hypothesis is true for round $t+1$, assuming it holds for round $t$.
Given $T_t$ and $\hat T_t$, consider a breakpoint $\alpha$ from $T_{t+1}$ introduced in round $t+1$.

Case 1: $\alpha\in G_t$.
Then the algorithm will recognize there exists a breakpoint, and use binary search to output a value $h(\alpha)$ such that $|\alpha-h(\alpha)|<\epsilon'$.
The interval $[\alpha,h(\alpha)]\cup[h(\alpha),\alpha]$ is added to $B_{t+1}$, but the good intervals to the left and right of this interval still have the correct centers.

Case 2: $\alpha\in B_t$. Then there exists an interval $[\alpha',h(\alpha')]\cup [h(\alpha'),\alpha]$ containing $\alpha$.
By assumption, this interval is size $<\epsilon'$, therefore, we set $h(\alpha')=h(\alpha)$, so there is a breakpoint within $\epsilon'$ of $\alpha$.

Therefore, for each breakpoint $\alpha$ in $T_{t+1}$, there exists a breakpoint $\hat \alpha$ in $\hat T_{t+1}$ such that $|\alpha-\hat\alpha|<\epsilon'$.
Furthermore, for all $\alpha\in G_{t+1}$, the set of centers for $\alpha$ in $T_{t+1}$ is the same as in $\hat T_{t+1}$.
This concludes the inductive proof.

Now we analyze the runtime of Algorithm~\ref{alg:fast}.
Let $(C,A)$ be any node in the algorithm, with centers $C$ and alpha
interval $A = [\alo, \ahi]$. Sorting the points in $\V$
according to their distance to $C$ has complexity $O(n \log n)$. Finding the
points sampled by $d^\alpha$-sampling with $\alpha$ set to $\alo$ and
$\ahi$ costs $O(n)$ time. Finally, computing the alpha interval $A_i$
for each child node of $(C,A)$ costs $O(n \log\frac{nH \log\left(\frac{\ahi}{\alo}\right)}{\epsilon})$ time, since we need
to perform $\log\frac{nkH \log\left(\frac{\ahi}{\alo}\right)}{\epsilon}$ iterations of binary search on $\alpha
\mapsto \frac{D_i(\alpha)}{D_n(\alpha)}$
and each evaluation of the function costs $O(n)$ time.
We charge this $O(n \log\frac{nH \log\left(\frac{\ahi}{\alo}\right)}{\epsilon})$ time to the corresponding child node.
If there are $N$ nodes in the execution tree, summing this cost over all nodes
gives a total running time of $O(N \cdot n \log\frac{nH \log\left(\frac{\ahi}{\alo}\right)}{\epsilon}))$.
If we let $\#I$ denote the total number of $\alpha$-intervals for $\V$,
then each layer of the execution tree has at most $\#I$ nodes, and the depth is
$k$, giving a total running time of $O(\#I \cdot k n \log\frac{nH \log\left(\frac{\ahi}{\alo}\right)}{\epsilon})$.

For the case when $\alo > 0$, Theorem~\ref{thm:expect} guarantees that we have
$\E[\#I]\leq 8nk\log (n) \log\left(\frac{\ahi}{\alo}\right)$. Therefore, the
expected runtime of Algorithm \ref{alg:fast} is $O\left(n^2 k^2
\log\left(\frac{\ahi}{\alo}\right) \left(\log\frac{nH
\log\left(\frac{\ahi}{\alo}\right)}{\epsilon}\right)\right)$.
The case where $\alo=0$ but there is a bound on the maximum ratio between any nonzero pair of distances is similar, except we use the appropriate statements from Theorem~\ref{thm:expect} and Lemma~\ref{lem:net}.
%$O\left(n^2 k^2\log n (\log \ahi + \log\log \dratio)\left(\log\frac{nH\log D}{\epsilon}\right)\right)$.
%This completes the proof.
\end{proof}

% Theorem~\ref{thm:runtime} gives meaningful runtime guarantees for $\alo>0$. Theorem~\ref{thm:runtime-dratio}
% in Appendix~\ref{app:theory} shows a similar guarantee for the case when $\alo=0$. The main difference in the
% argument is that we rely on the second guarantees of Theorem~\ref{thm:expect} and Lemma~\ref{lem:net}
% to compute the final runtime. The guarantees in Theorem~\ref{thm:runtime-dratio} require that the largest ratio
% between any pair of non-zero distances in the instance is at most $\dratio$.

%In Appendix~\ref{app:theory}, we show that $d^\alpha$-sampling is Lipschitz as a function of $\alpha$ in Lemma \ref{lem:net}. Therefore,
Since we showed that $d^\alpha$-sampling is Lipschitz as a function of $\alpha$ in Lemma \ref{lem:net},
it is also possible to find the best $\alpha$ parameter with sub-optimality at most $\epsilon$ by finding
the best point from a discretization of $[0,\ahi]$ with step-size $s = \epsilon/(Hn^2k\log \dratio)$.
The running time of this algorithm is $O(n^3 k^2 H \log \dratio / \epsilon)$, which is
significantly slower than the efficient algorithm presented in this section.
Intuitively, Algorithm~\ref{alg:fast} is able to binary search to find each breakpoint in time $O(\log\frac{nkH\log \dratio}{\epsilon})$,
whereas a discretization-based algorithm must check all values of alpha uniformly,
so the runtime of the discretization-based algorithm increases by a multiplicative factor of $O\left(\frac{nH\log \dratio}{\epsilon}\cdot \left(\log\frac{nH\log \dratio}{\epsilon}\right)^{-1}\right)$.

%\paragraph{Generalized families}

%Theorems~\ref{thm:expect} and~\ref{thm:runtime} are not just true specifically for $\lloyds$; they can be made much more general.
%In fact, these theorems are true for any randomized initialization procedure which satisfies basic structure.
%In this section, we will explore the minimum properties needed to still obtain these theorems.

\paragraph{Generalized families}

Theorems~\ref{thm:expect} and~\ref{thm:runtime} are not just true specifically for $\lloyds$; they can be made much more general.
In this section, we show these theorems are true for any family of randomized initialization procedures which satisfies a few simple properties.
First we formally define a parameterized family of randomized initialization procedures.

\begin{definition} \label{def:init}
Given a clustering instance $\mathcal{C}=(V,d)$ and a function $p:V\times 2^V\mapsto [0,1]$,
a \emph{$p$-randomized initialization method} is an iterative method such that $C$ is initialized as empty and
one new center is added in each round, and in round $t$, each point $v$ is chosen as a center with probability $p(v,C)$.
Note that for all $C\subseteq V$ such that $|C|\leq k$, we must have $\sum_{v\in V}p(v,C)=1$.
An \emph{$\alpha$-parameterized family} is a set of $p$-randomized initialization methods
$p_\alpha:V\times 2^V\mapsto [0,1]$.
\end{definition}

We will show how to prove general upper bounds on $\alpha$-parameterized families of $p$-randomized initialization
methods as long as the functions $p_\alpha$ satisfy a few simple properties.

Just as with $\lloyds$, we specify precisely how we execute the randomized steps
of the algorithm. We assume the algorithm draws a vector
  $\vec{Z}=\{z_1,\dots,z_k\}$ from $[0,1]^k$ uniformly at random. Then in round
$t$, the algorithm partitions $[0,1]$ into $n$ intervals, where there is an
interval $I_{v_i}$ for each $v_i$ with size equal to the probability
$p_\alpha(v_i, C)$ of choosing $v_i$ in round $t$ where $C$ is the set of
centers at the start of round $t$. Then the algorithm chooses the point $v_i$ as
a center, where $z_t\in I_{v_i}$.

\begin{comment}
\begin{algorithm}
\caption{$\alpha$-parameterized $p$-randomized initialization method}\label{alg:clus}
\begin{algorithmic} %[1]
\STATE {\bfseries Input:} Instance $\mathcal{V} = (V,d,k)$, parameter $\alpha$, function $p_\alpha:V\times 2^V\mapsto [0,1]$.
\STATE Initialize $C = \emptyset$ and draw a vector $\vec{Z}=\{z_1,\dots,z_k\}$ from $[0,1]^k$ uniformly at random.
\STATE For each $t = 1, \dots, k$:
\begin{enumerate}[leftmargin=10pt,itemsep=0pt]
\item Partition $[0,1]$ into $n$ intervals, where there is an interval $I_{v_i}$ for each $v_i$
with size equal to the probability of choosing $v_i$ in round $t$ ( $p_\alpha(v_i,C)$ divided by a normalizing constant).
\item Denote $c_t$ as the point such that $z_t\in I_{c_t}$, and add $c_t$ to $C$.
\end{enumerate}
\STATE {\bfseries Output:} Centers $C$.
\end{algorithmic}
\end{algorithm}
\end{comment}

Now we will show how to upper bound the number of discontinuities of the output
as a function of $\alpha$, which will lead to sample-efficient and
computationally-efficient meta-algorithms. We start with a few definitions.
Without loss of generality, we assume that in each round we rename the points so
that they satisfy $p_\alo(v_1,C)\geq\cdots\geq p_\alo(v_n,C)$. For a given $C$
and $v_i$, we define the partial sums $S_{i,C}(\alpha)=\sum_{j=1}^i
p_\alpha(v_i,C)$. Let the set of permissible values of $\alpha$ be the interval
$[\alo,\ahi]$. Let $\seed_\alpha(\V,\vec{Z},p)$ denote the sequence of centers
returned with randomness $\vec{Z}\in [0,1]^k$ and parameter $\alpha$. Let
$D_{p}=\max_{i,C,v,\alpha}\left(\frac{\partial S_{i,C}(\alpha)}{\partial
\alpha}\right)$.

\begin{theorem} \label{thm:init}
Given an $\alpha$-parameterized family such that \emph{(1)} for all $1\leq i\leq k$ and $C\subseteq V$ such that $|C|\leq k$,
each $S_{i,C}(\alpha)$ is monotone increasing and continuous as a function of $\alpha$,
and \emph{(2)} for all $1\leq i\leq j\leq n$ and $\alpha\in (\alo,\ahi)$, $S_{i,C}(\alpha)\leq S_{j,C}(\alpha)$,
then the expected number of discontinuities of $\seed_\alpha(\V,\vec{Z},p)$ as a function of $\alpha$ is
$O\left(nkD_p(\ahi-\alo)\right)$.
\end{theorem}

Note this theorem is a generalization of Theorem~\ref{thm:expect}, since $\lloyds$ is an $\alpha$-parameterized property
with the properties (due to Lemma~\ref{lem:monotoneBins}) and $D_p=4\log n$.
We give the proof of Theorem~\ref{thm:init} in Appendix \ref{app:theory}, since it is similar to the proof of Theorem~\ref{thm:expect}.
Intuitively, a key part of the argument is bounding the expected number of discontinuities in an arbitrary interval
$[\alpha_\ell,\alpha_{\ell+1}]$ for a current set of centers $C$. If the algorithm chooses $v_j$ as a center at $\alpha_{\ell}$
and $v_i$ at $\alpha_\ell$, then the only possible centers are $v_i,\dots,v_j$.
We show this key fact follows for any partial sum of probability functions as long as they are
monotone, continuous, and non-crossing.
Furthermore, $j-i$ is bounded by the maximum derivative of the partial sums over $[\alpha_\ell,\alpha_{\ell+1}]$,
so the final upper bound scales with $D_p$.
We can also achieve a generalized version of Theorem~\ref{thm:runtime}.

\begin{theorem} \label{thm:gen_runtime}
Given parameters $0\leq\alo<\ahi$, $\epsilon>0$, a sample $\mathcal{S}$ of size
$$m = O\left(\left(\frac{H}{\epsilon}\right)^2\log \left(\frac{\ahi n D_p}{\delta}\right)\right)$$  from $\left(\mathcal{D} \times [0,1]^k\right)^m$,
and an $\alpha$-parameterized family satisfying properties \emph{(1)} and \emph{(2)} from
Theorem~\ref{thm:init}, run Algorithm~\ref{alg:fast} on
each sample and collect all breakpoints (i.e., boundaries of the intervals
$A_i$). With probability at least $1-\delta$, the breakpoint $\bar\alpha$ with
lowest empirical cost satisfies
$|\clus_{\bar\alpha,\beta}(\mathcal{S})-\min_{0\leq\alpha\leq\ahi}\clus_{\alpha,\beta}(\mathcal{S})|<\epsilon$.
The total running time to find the best breakpoint is
$O\left(m n^2 k^2 \ahi D_p \log \left( \frac{nH }{\epsilon}\right)\log n\right)$.
\end{theorem}

\section{Experiments} \label{sec:experiments}

 In this section, we empirically evaluate the effect of the $\alpha$ and $\beta$
 parameters on clustering cost for real-world and synthetic clustering domains.
 We find that the optimal $\alpha$ and $\beta$ parameters vary significantly from
 domain to domain. We also find that the number of possible initial centers
 chosen by $d^\alpha$-sampling scales linearly with $n$ and $k$,
 suggesting our Theorem \ref{thm:expect} is tight up to logarithmic factors.
 Finally, we show the empirical distribution of $\alpha$-interval boundaries.

\vspace{1em} \noindent \textbf{Experiment Setup.} Our experiments evaluate the
$\lloyds$ family of algorithms on several distributions over clustering
instances. Our clustering instance distributions are derived from classification
datasets. For each classification dataset, we sample a clustering instance by
choosing a random subset of $k$ labels, sampling $N$ examples belonging to each
of the $k$ chosen labels. The clustering instance then consists of the $kN$
points, and the target clustering is given by the ground-truth labels. This
sampling distribution covers many related clustering tasks (i.e., clustering
different subsets of the same labels). We evaluate clustering performance in
terms of the Hamming distance to the optimal clustering, or the fraction of
points assigned to different clusters by the algorithm and the target
clustering. Formally, the Hamming distance between the outputted clustering
$\{C_1,\dots,C_k\}$ and the optimal clustering $\{C^*_1,\dots,C^*_k\}$ is
measured by $\min_\sigma \frac{1}{n}\sum_{i=1}^k C_i\setminus C^*_{\sigma(i)}$,
where the minimum is taken over all permutations $\sigma$ of the cluster
indices. In all cases, we limit the algorithm to performing $3$ iterations of
$\beta$-Lloyds. We use the following datasets:

\vspace{1em}
\noindent
\textit{MNIST:} We use the raw pixel representations of the subset of
MNIST~\citep{loosli2007training}. For MNIST we set $k = 5$, $N = 100$, so each
instance consists of $n = 500$ points.

\vspace{1em}
\noindent
\textit{CIFAR-10:} The CIFAR-10 dataset \citep{CIFAR10} is an image dataset with
10 classes. Following \citet{krizhevsky2012} we include 50 randomly rotated and
cropped copies of each example. We extract the features from the Google
Inception network \citep{Szegedy_2015_CVPR} using layer in4d. For CIFAR10 we set
$k = 5$, $N = 100$, so each instance consists of $n = 500$ points.

\vspace{1em}
\noindent
\textit{CNAE-9:} The CNAE-9~\citep{Dua:2017} dataset consists of 1080 documents
describing Brazilian companies categorized into 9 categories. Each example is
represented by an 875-dimensional vector, where each entry is the frequency of a
specific word. For CNAE-9 we set $k=4$ and $N = 100$, so each clustering
instance has $n = 400$ points.

\vspace{1em}
\noindent
\textit{Gaussian Grid:} We also use a synthetic 2-dimensional clustering
instance where points are sampled from a mixture of $9$ standard Gaussians
arranged in a $3\times 3$ grid with a grid stride of 5. For the Gaussian Grid
dataset, we set $k = 4$ and $N = 120$, so each clustering instance has $n = 500$
points.

\vspace{1em}\noindent\textbf{Parameter Study.} Our first experiment explores the
effect of the $\alpha$ and $\beta$ parameters. For each dataset, we sampled
$m=50,000$ sample clustering instances and run  $\lloyds{}$ for all combinations
of $50$ values of $\alpha$ evenly spaced in $[0,20]$ and $25$ values of $\beta$
evenly spaced in $[1,10]$. Figure~\ref{fig:parameterStudy} shows the average
Hamming error on each dataset as a function of $\alpha$ and $\beta$. The optimal
parameters vary significantly across the datasets. On MNIST and the Gaussian
grid, it is best to set $\alpha$ to be large, while on the remaining datasets
the optimal value is low. Neither the $k$-means$++$ algorithm nor farthest-first
traversal have good performance across all datasets. On the Gaussian grid
example, $k$-means++ has Hamming error $6.8\%$ while the best parameter setting
only has Hamming error $1.3\%$.

\begin{figure*}
\centering
  \subfigure[MNIST]{\includegraphics[height=3.9cm]{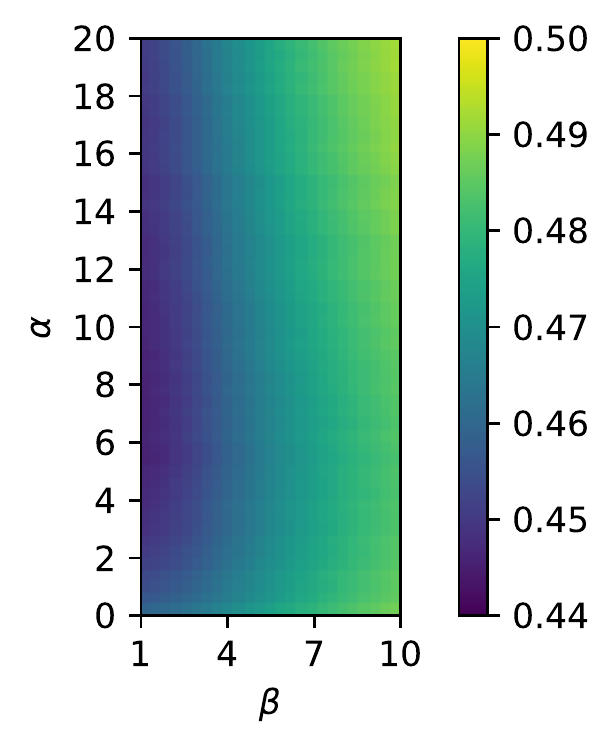}}
  \subfigure[CIFAR-10]{\includegraphics[height=3.9cm]{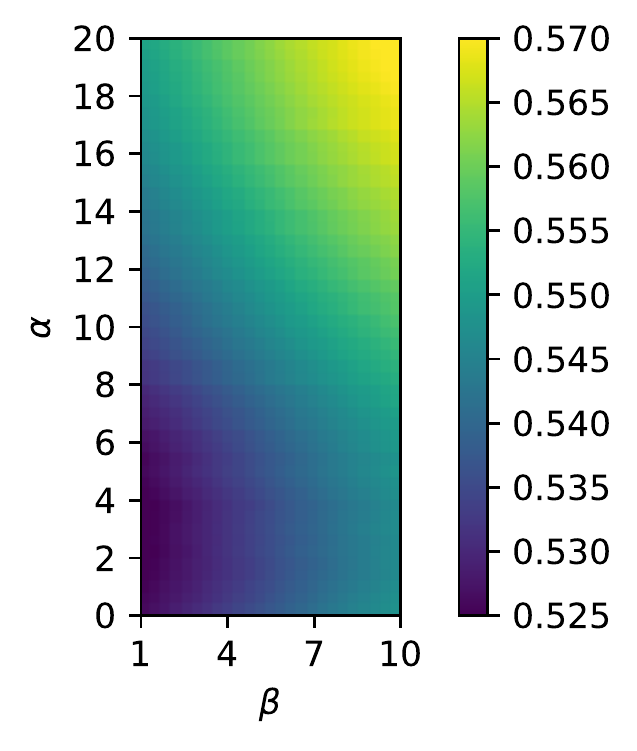}}
  \subfigure[CNAE-9]{\includegraphics[height=3.9cm]{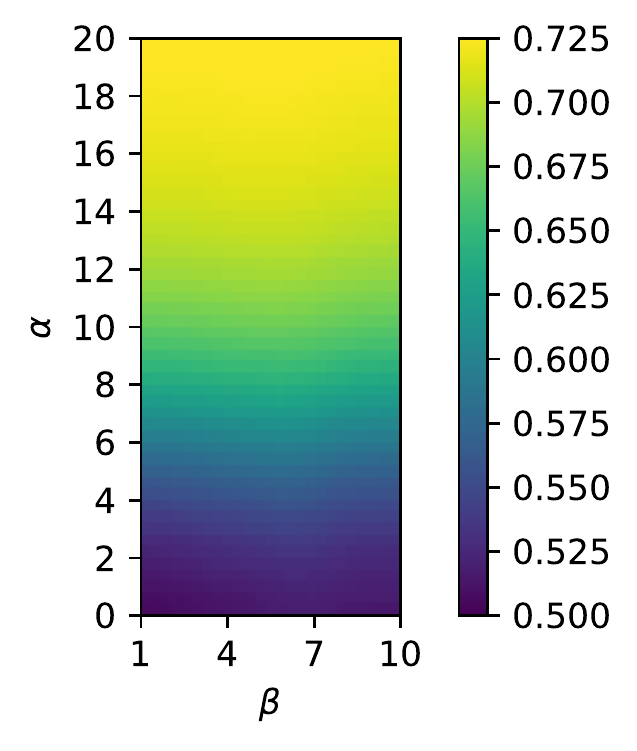}}
  \subfigure[Gaussian Grid]{\includegraphics[height=3.9cm]{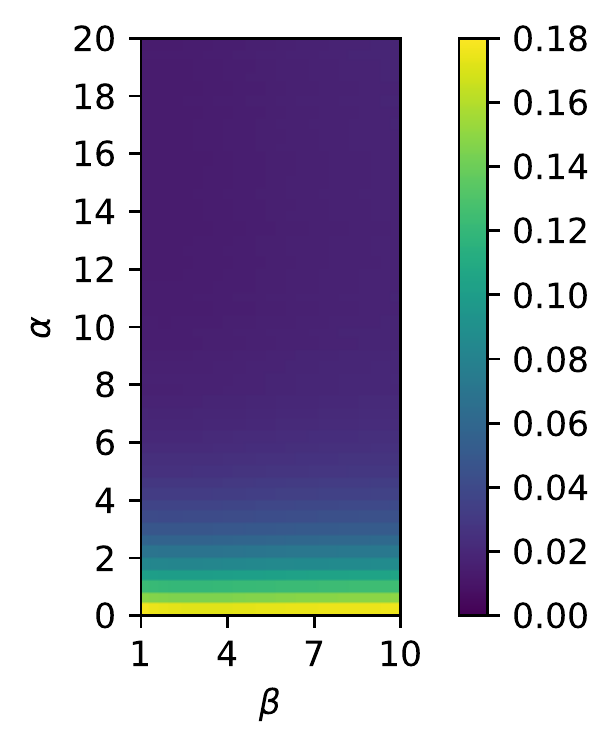}}
%  \vspace{-1em}
  \caption{Average Hamming error of $\lloyds{}$ as a function of $\alpha$ and $\beta$.}
  \label{fig:parameterStudy}
\end{figure*}

Next, we use Algorithm \ref{alg:fast} to tune $\alpha$ without discretization.
In these experiments we set $\beta=2$ and modify step 6 of
Algorithm~\ref{alg:clus} to compute the mean of each cluster, rather than the
point in the dataset minimizing the sum of squared distances to points in that
cluster (as is usually done in Lloyd's method). This modification improves
running time and also the Hamming cost of the resulting clusterings. For each
dataset, we sample $m = 50,000$ sample clustering instances and divide them
evenly into testing and training sets (for MNIST we set $m = 250,000$ instead).
We plot the average Hamming cost as a function of $\alpha$ on both the training
and testing sets. The optimal value of $\alpha$ varies between the different
datasets, showing that tuning the parameters leads to improved performance.
Interestingly, for MNIST the value of $\alpha$ with lowest Hamming error is
$\alpha = 4.1$ which does not correspond to any standard algorithm. Moreover, in
each case the difference between the training and testing error plots is small,
supporting our generalization claims.

\begin{figure*}
  \centering
  \subfigure[MNIST]{\includegraphics[height=3cm]{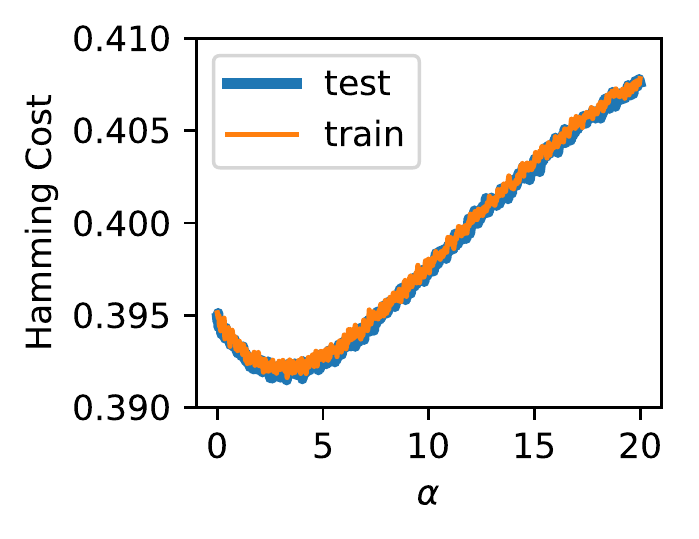}}
  \subfigure[CIFAR-10]{\includegraphics[height=3cm]{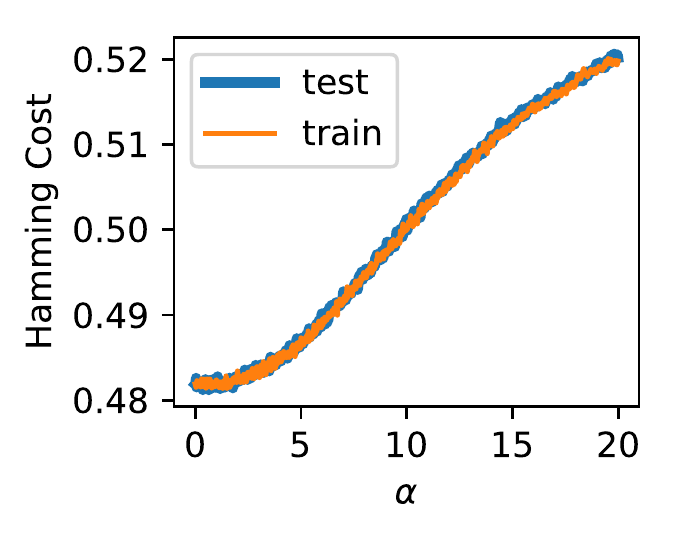}}
  \subfigure[CNAE-9]{\includegraphics[height=3cm]{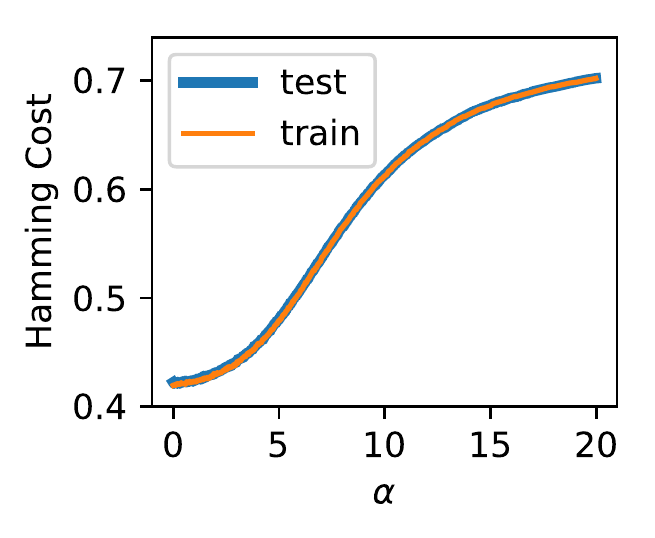}}
  \subfigure[Gaussian Grid]{\includegraphics[height=3cm]{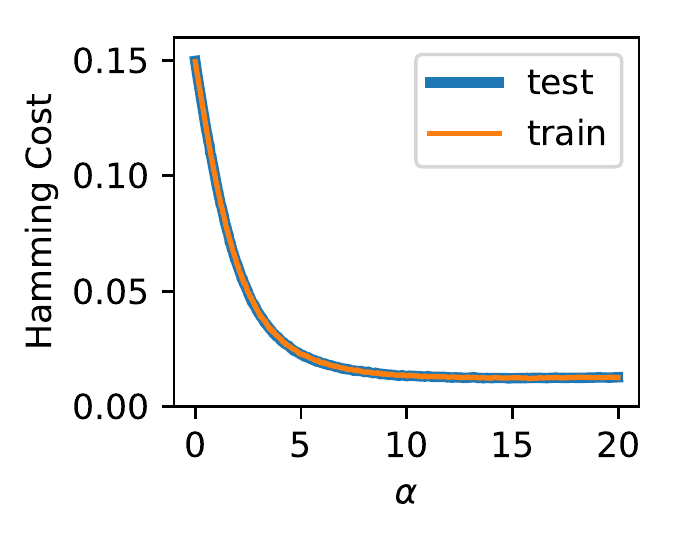}}
	\caption{Average Hamming error of $\lloyds{}$ as a function of $\alpha$ for $\beta = 2$.}
  \label{fig:alphaStudy}
\end{figure*}

\vspace{1em} \noindent \textbf{Number of $\alpha$-Intervals.} Next we report the
number of $\alpha$-intervals in the above experiments. On average, MNIST had
$826.1$ intervals per instance, CIFAR-10 had $994.5$ intervals per instance,
CNAE-9 had $855.9$ intervals per instance, and the Gaussian grid had $953.4$
intervals per instance.

In Figure~\ref{fig:alphaIntervals} we evaluate how the number of $\alpha$
intervals grows with the clustering instance size $n$. For $n \in \{50, 100,
150, \dots, 1000\}$, we modify the above distributions by setting $N = n/k$ and
plot the average number of $\alpha$-intervals on $m=5000$ samples. For all four
datasets, the average number of intervals grows nearly linearly with $n$.

\begin{figure}
\centering
\includegraphics[width=0.4\columnwidth]{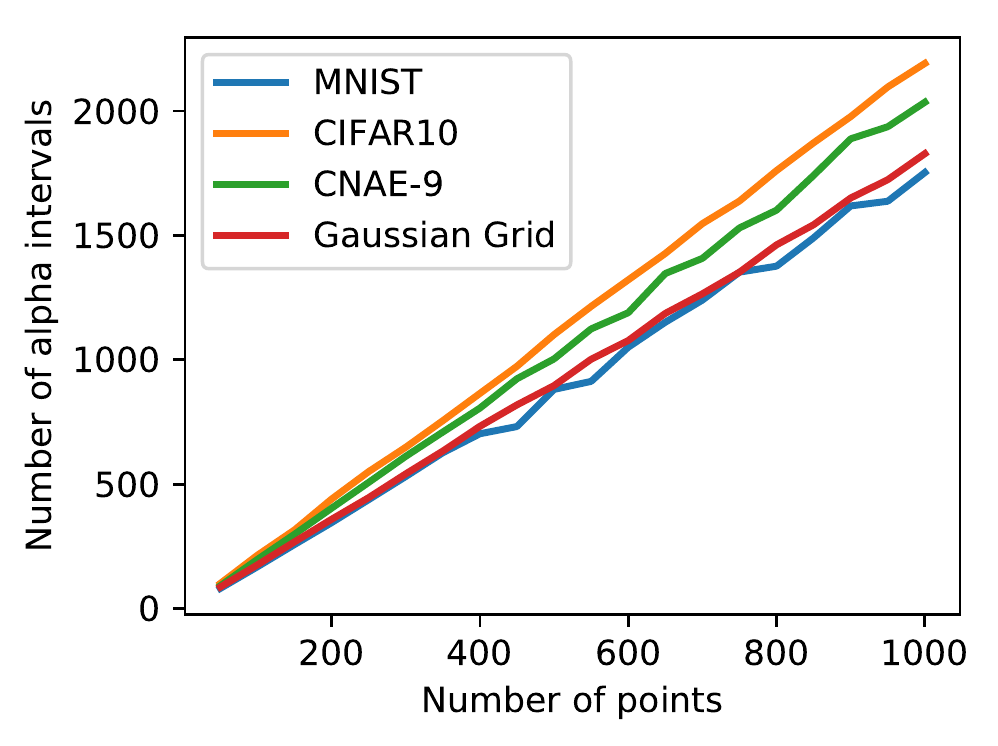}
\caption{The number of $\alpha$-intervals as a function of clustering instance size.}
\label{fig:alphaIntervals}
\end{figure}

\paragraph{Distribution of $\alpha$-decision Points.} Finally, for each dataset,
Figure~\ref{fig:alphaDistribution} shows a histogram of the distribution of the
$\alpha$-interval boundaries.

\begin{figure*}
  \centering
  \subfigure[MNIST]{\includegraphics[height=3cm]{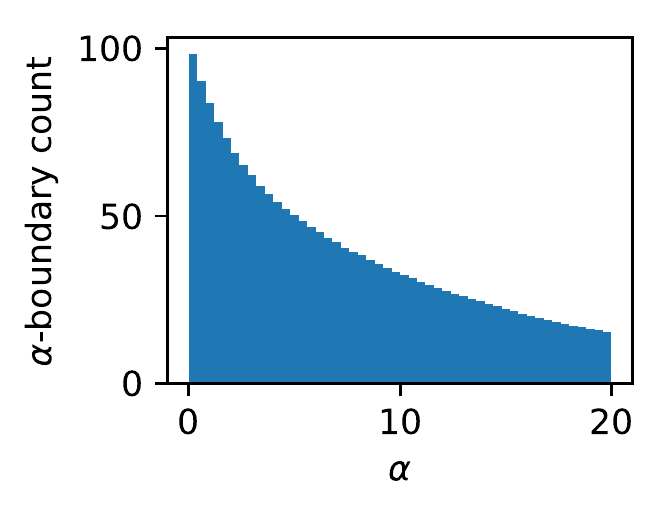}}
  \subfigure[CIFAR-10]{\includegraphics[height=3cm]{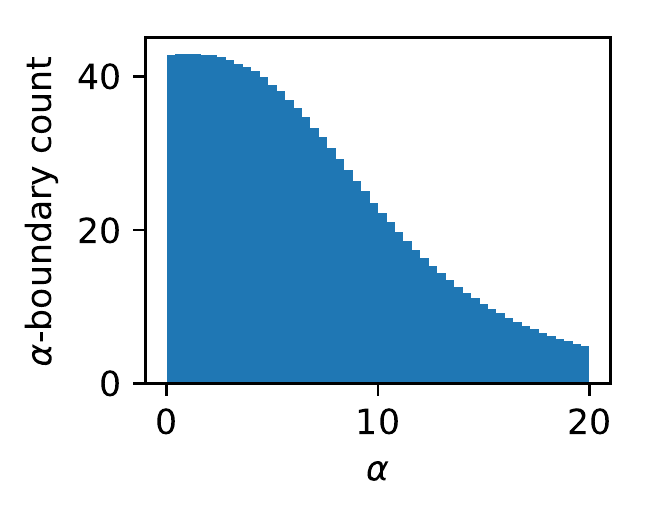}}
  \subfigure[CNAE-9]{\includegraphics[height=3cm]{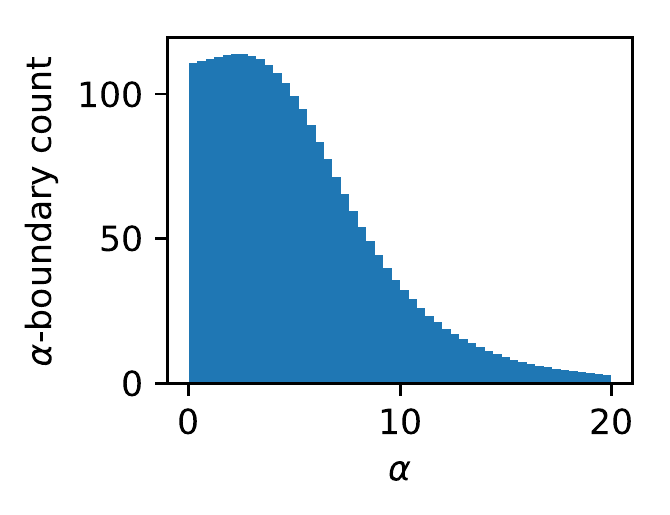}}
  \subfigure[Gaussian Grid]{\includegraphics[height=3cm]{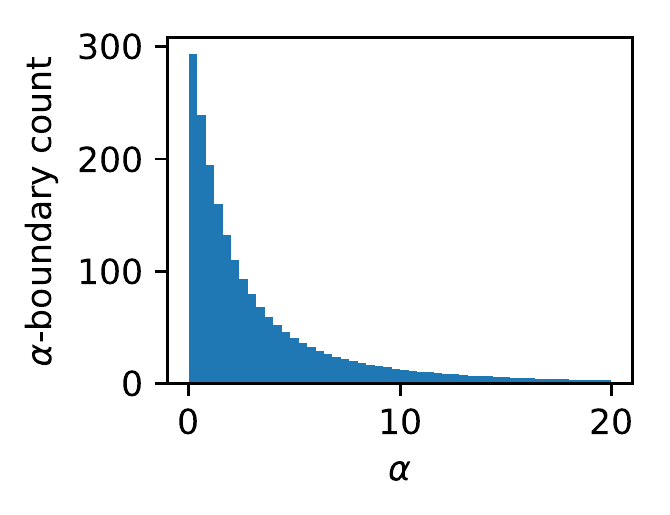}}
  \caption{Distribution of $\alpha$-decision points.}
  \label{fig:alphaDistribution}
\end{figure*}

\section{Conclusion} \label{sec:conclusion}

We define an infinite family of algorithms generalizing Lloyd's method, with one parameter controlling the initialization procedure,
and another parameter controlling the local search procedure.
This family of algorithms includes the celebrated $k$-means$++$ algorithm,
as well as the classic farthest-first traversal algorithm.
We provide a sample efficient and computationally efficient algorithm to learn a near-optimal parameter over an unknown distribution
of clustering instances, by developing techniques to bound the expected number of discontinuities in the cost as a function of the parameter.
We give a thorough empirical analysis, showing that the value of the optimal parameters transfer to related clustering instances.
We show the optimal parameters vary among different application domains, and the optimal parameters often significantly improve the error
compared to existing algorithms such as $k$-means++ and farthest-first traversal.

\section{Acknowledgments} \label{sec:acknowledgements} This work was supported
in part by NSF grants CCF-1535967, IIS-1618714, an Amazon Research Award, a
Microsoft Research Faculty Fellowship, a National Defense Science \& Engineering
Graduate (NDSEG) fellowship, and by the generosity of Eric and Wendy Schmidt by
recommendation of the Schmidt Futures program.

%\newpage

\bibliographystyle{plain}
\bibliography{clustering}

\begin{thebibliography}{10}

\bibitem{ackerman2010towards}
Margareta Ackerman, Shai Ben-David, and David Loker.
\newblock Towards property-based classification of clustering paradigms.
\newblock In {\em Proceedings of the Annual Conference on Neural Information
  Processing Systems (NIPS)}, pages 10--18, 2010.

\bibitem{ahmadian2016better}
Sara Ahmadian, Ashkan Norouzi-Fard, Ola Svensson, and Justin Ward.
\newblock Better guarantees for k-means and euclidean k-median by primal-dual
  algorithms.
\newblock In {\em Proceedings of the Annual Symposium on Foundations of
  Computer Science (FOCS)}, 2017.

\bibitem{arai2007hierarchical}
Kohei Arai and Ali~Ridho Barakbah.
\newblock Hierarchical k-means: an algorithm for centroids initialization for
  k-means.
\newblock {\em Reports of the Faculty of Science and Engineering},
  36(1):25--31, 2007.

\bibitem{arthur2011smoothed}
David Arthur, Bodo Manthey, and Heiko R{\"o}glin.
\newblock Smoothed analysis of the k-means method.
\newblock {\em Journal of the ACM (JACM)}, 58(5):19, 2011.

\bibitem{arthur2006slow}
David Arthur and Sergei Vassilvitskii.
\newblock How slow is the k-means method?
\newblock In {\em Proceedings of the twenty-second annual symposium on
  Computational geometry}, pages 144--153. ACM, 2006.

\bibitem{arthur2007k}
David Arthur and Sergei Vassilvitskii.
\newblock k-means++: The advantages of careful seeding.
\newblock In {\em Proceedings of the Annual Symposium on Discrete Algorithms
  (SODA)}, pages 1027--1035, 2007.

\bibitem{arya2004local}
Vijay Arya, Naveen Garg, Rohit Khandekar, Adam Meyerson, Kamesh Munagala, and
  Vinayaka Pandit.
\newblock Local search heuristics for k-median and facility location problems.
\newblock {\em SIAM Journal on Computing}, 33(3):544--562, 2004.

\bibitem{ashtiani2015representation}
Hassan Ashtiani and Shai Ben-David.
\newblock Representation learning for clustering: a statistical framework.
\newblock In {\em Proceedings of the Annual Conference on Uncertainty in
  Artificial Intelligence}, pages 82--91, 2015.

\bibitem{balcan2017learning}
Maria-Florina Balcan, Vaishnavh Nagarajan, Ellen Vitercik, and Colin White.
\newblock Learning-theoretic foundations of algorithm configuration for
  combinatorial partitioning problems.
\newblock In {\em Proceedings of the Annual Conference on Learning Theory
  (COLT)}, pages 213--274, 2017.

\bibitem{bartlett2002rademacher}
Peter~L Bartlett and Shahar Mendelson.
\newblock Rademacher and gaussian complexities: Risk bounds and structural
  results.
\newblock {\em Journal of Machine Learning Research}, 3(Nov):463--482, 2002.

\bibitem{byrka2015improved}
Jaros{\l}aw Byrka, Thomas Pensyl, Bartosz Rybicki, Aravind Srinivasan, and Khoa
  Trinh.
\newblock An improved approximation for k-median, and positive correlation in
  budgeted optimization.
\newblock In {\em Proceedings of the Annual Symposium on Discrete Algorithms
  (SODA)}, pages 737--756, 2015.

\bibitem{charikar1999constant}
Moses Charikar, Sudipto Guha, {\'E}va Tardos, and David~B Shmoys.
\newblock A constant-factor approximation algorithm for the k-median problem.
\newblock In {\em Proceedings of the Annual Symposium on Theory of Computing
  (STOC)}, pages 1--10, 1999.

\bibitem{cohen2016geometric}
Michael~B Cohen, Yin~Tat Lee, Gary Miller, Jakub Pachocki, and Aaron Sidford.
\newblock Geometric median in nearly linear time.
\newblock In {\em Proceedings of the Annual Symposium on Theory of Computing
  (STOC)}, pages 9--21. ACM, 2016.

\bibitem{dasgupta2005performance}
Sanjoy Dasgupta and Philip~M Long.
\newblock Performance guarantees for hierarchical clustering.
\newblock {\em Journal of Computer and System Sciences}, 70(4):555--569, 2005.

\bibitem{dempster1977maximum}
Arthur~P Dempster, Nan~M Laird, and Donald~B Rubin.
\newblock Maximum likelihood from incomplete data via the em algorithm.
\newblock {\em Journal of the royal statistical society}, pages 1--38, 1977.

\bibitem{Dua:2017}
Dua Dheeru and Efi Karra~Taniskidou.
\newblock {UCI} machine learning repository, 2017.

\bibitem{friedman2001elements}
Jerome Friedman, Trevor Hastie, and Robert Tibshirani.
\newblock {\em The elements of statistical learning}, volume~1.
\newblock Springer series in statistics New York, NY, USA:, 2001.

\bibitem{ganin2014unsupervised}
Yaroslav Ganin and Victor Lempitsky.
\newblock Unsupervised domain adaptation by backpropagation.
\newblock In {\em Proceedings of the International Conference on Machine
  Learning (ICML)}, 2015.

\bibitem{gonzalez1985clustering}
Teofilo~F Gonzalez.
\newblock Clustering to minimize the maximum intercluster distance.
\newblock {\em Theoretical Computer Science}, 38:293--306, 1985.

\bibitem{har2005fast}
Sariel Har-Peled and Bardia Sadri.
\newblock How fast is the k-means method?
\newblock {\em Algorithmica}, 41(3):185--202, 2005.

\bibitem{higgs1997experimental}
Richard~E Higgs, Kerry~G Bemis, Ian~A Watson, and James~H Wikel.
\newblock Experimental designs for selecting molecules from large chemical
  databases.
\newblock {\em Journal of chemical information and computer sciences},
  37(5):861--870, 1997.

\bibitem{jiang2012transfer}
Wenhao Jiang and Fu-lai Chung.
\newblock Transfer spectral clustering.
\newblock In {\em Proceedings of the Annual Conference on Knowledge Discovery
  and Data Mining (KDD)}, pages 789--803, 2012.

\bibitem{kanungo2002efficient}
Tapas Kanungo, David~M Mount, Nathan~S Netanyahu, Christine~D Piatko, Ruth
  Silverman, and Angela~Y Wu.
\newblock An efficient k-means clustering algorithm: Analysis and
  implementation.
\newblock {\em transactions on pattern analysis and machine intelligence},
  24(7):881--892, 2002.

\bibitem{kaufman2009finding}
Leonard Kaufman and Peter~J Rousseeuw.
\newblock {\em Finding groups in data: an introduction to cluster analysis},
  volume 344.
\newblock John Wiley \& Sons, 2009.

\bibitem{kleinberg2003impossibility}
Jon~M Kleinberg.
\newblock An impossibility theorem for clustering.
\newblock In {\em Proceedings of the Annual Conference on Neural Information
  Processing Systems (NIPS)}, pages 463--470, 2003.

\bibitem{kobren2017online}
Ari Kobren, Nicholas Monath, Akshay Krishnamurthy, and Andrew McCallum.
\newblock An online hierarchical algorithm for extreme clustering.
\newblock In {\em Proceedings of the Annual Conference on Knowledge Discovery
  and Data Mining (KDD)}, 2017.

\bibitem{koltchinskii2001rademacher}
Vladimir Koltchinskii.
\newblock Rademacher penalties and structural risk minimization.
\newblock {\em IEEE Transactions on Information Theory}, 47(5):1902--1914,
  2001.

\bibitem{CIFAR10}
Alex Krizhevsky.
\newblock Learning multiple layers of features from tiny images.
\newblock Technical report, University of Toronto, 2009.

\bibitem{krizhevsky2012}
Alex Krizhevsky, Ilya Sutskever, and Geoffrey~E Hinton.
\newblock Imagenet classification with deep convolutional neural networks.
\newblock In {\em Proceedings of the Annual Conference on Neural Information
  Processing Systems (NIPS)}, pages 1097--1105, 2012.

\bibitem{lloyd1982least}
Stuart Lloyd.
\newblock Least squares quantization in pcm.
\newblock {\em transactions on information theory}, 28(2):129--137, 1982.

\bibitem{loosli2007training}
Ga{\"e}lle Loosli, St{\'e}phane Canu, and L{\'e}on Bottou.
\newblock Training invariant support vector machines using selective sampling.
\newblock {\em Large scale kernel machines}, pages 301--320, 2007.

\bibitem{macqueen1967some}
James MacQueen et~al.
\newblock Some methods for classification and analysis of multivariate
  observations.
\newblock In {\em symposium on mathematical statistics and probability},
  volume~1, pages 281--297. Oakland, CA, USA, 1967.

\bibitem{massart2000some}
Pascal Massart.
\newblock Some applications of concentration inequalities to statistics.
\newblock In {\em Annales-Faculte des Sciences Toulouse Mathematiques},
  volume~9, pages 245--303. Universit{\'e} Paul Sabatier, 2000.

\bibitem{max1960quantizing}
Joel Max.
\newblock Quantizing for minimum distortion.
\newblock {\em IRE Transactions on Information Theory}, 6(1):7--12, 1960.

\bibitem{ostrovsky2012effectiveness}
Rafail Ostrovsky, Yuval Rabani, Leonard~J Schulman, and Chaitanya Swamy.
\newblock The effectiveness of lloyd-type methods for the k-means problem.
\newblock {\em Journal of the ACM (JACM)}, 59(6):28, 2012.

\bibitem{pelleg1999accelerating}
Dan Pelleg and Andrew Moore.
\newblock Accelerating exact k-means algorithms with geometric reasoning.
\newblock In {\em Proceedings of the Annual Conference on Knowledge Discovery
  and Data Mining (KDD)}, pages 277--281, 1999.

\bibitem{pena1999empirical}
Jos{\'e}~M Pena, Jose~Antonio Lozano, and Pedro Larranaga.
\newblock An empirical comparison of four initialization methods for the
  k-means algorithm.
\newblock {\em Pattern recognition letters}, 20(10):1027--1040, 1999.

\bibitem{pruitt2011ncbi}
Kim~D Pruitt, Tatiana Tatusova, Garth~R Brown, and Donna~R Maglott.
\newblock Ncbi reference sequences (refseq): current status, new features and
  genome annotation policy.
\newblock {\em Nucleic acids research}, 40(D1):D130--D135, 2011.

\bibitem{raina2007self}
Rajat Raina, Alexis Battle, Honglak Lee, Benjamin Packer, and Andrew~Y Ng.
\newblock Self-taught learning: transfer learning from unlabeled data.
\newblock In {\em Proceedings of the International Conference on Machine
  Learning (ICML)}, pages 759--766, 2007.

\bibitem{sener2016learning}
Ozan Sener, Hyun~Oh Song, Ashutosh Saxena, and Silvio Savarese.
\newblock Learning transferrable representations for unsupervised domain
  adaptation.
\newblock In {\em Proceedings of the Annual Conference on Neural Information
  Processing Systems (NIPS)}, pages 2110--2118, 2016.

\bibitem{Szegedy_2015_CVPR}
Christian Szegedy, Wei Liu, Yangqing Jia, Pierre Sermanet, Scott Reed, Dragomir
  Anguelov, Dumitru Erhan, Vincent Vanhoucke, and Andrew Rabinovich.
\newblock Going deeper with convolutions.
\newblock {\em The IEEE Conference on Computer Vision and Pattern Recognition
  (CVPR)}, June 2015.

\bibitem{tossavainen}
Timo Tossavainen.
\newblock On the zeros of finite sums of exponential functions.
\newblock {\em Australian Mathematical Society Gazette}, 33(1):47--50, 2006.

\bibitem{tzeng2014deep}
Eric Tzeng, Judy Hoffman, Ning Zhang, Kate Saenko, and Trevor Darrell.
\newblock Deep domain confusion: Maximizing for domain invariance.
\newblock {\em arXiv preprint arXiv:1412.3474}, 2014.

\bibitem{vattani2011k}
Andrea Vattani.
\newblock K-means requires exponentially many iterations even in the plane.
\newblock {\em Discrete \& Computational Geometry}, 45(4):596--616, 2011.

\bibitem{yang2009heterogeneous}
Qiang Yang, Yuqiang Chen, Gui-Rong Xue, Wenyuan Dai, and Yong Yu.
\newblock Heterogeneous transfer learning for image clustering via the social
  web.
\newblock In {\em Proceedings of the Conference on Natural Language
  Processing}, pages 1--9, 2009.

\end{thebibliography}

\newpage

\appendix

%\section{Related Work Continued} \label{app:related_work}
%\input{related_work-appendix}

\section{Table of Notation} \label{app:notation}
\begin{table*}%[t]
\caption{Notation table}
\begin{center}
\begin{tabular}{|c||c|}
	\hline
	Symbol & Description \\
	\hline
	$\V=(V,d,k)$ & Clustering instance of points $V$ with distance metric $d$, number of clusters $k$ \\
	\hline
	$n$ & $n=|V|$, the size of the point set \\
	\hline
	$k$ & $k$=number of clusters in $\V$\\
	\hline
	$\mathcal{C}^*=\{C_1^*,\dots,C_k^*\}$ & Optimal clusters according to an objective such as $k$-means \\
	\hline
	$\lloyds$ & Clustering algorithm with parameters $\alpha\in [0,\infty)$ and $\beta\in[1,\infty)$ \\
	\hline
	$\vec{Z}$ & $\vec{Z}=\{z_1,\dots,z_k\}\in [0,1]^k$ is the random seed for the $\lloyds$ algorithm\\
	\hline
	$\clus_{\alpha,\beta}\left(\V,\vec{Z}\right)$ & Cost of the clustering outputted by $\lloyds$ \\
	& with randomness $\vec{Z}\in [0,1]^k$ \\
	\hline
	$\clus_{\alpha,\beta}(\V)$ & $\clus_{\alpha,\beta}(\V)=\E_{\vec{Z}\sim [0,1]^k}\left[\clus_{\alpha,\beta}\left(\V,\vec{Z}\right)\right]$	 \\
	\hline
	$\seed_\alpha(\V, \vec{Z})$ & the output centers of phase 1 of Algorithm~\ref{alg:clus} run on $\V$ with vector $\vec{Z}$\\
	\hline
	$\lloyd_\beta(\V,C,T)$ & cost of the outputted clustering from phase 2 of Algorithm \ref{alg:clus} \\
	& on instance $\V$ with initial centers $C$, and a maximum of $T$ iterations.\\
	\hline
	$H$ & $H$= maximum loss for any clustering instance $\V$ \\
	\hline
	$s$ & $s$= the minimum ratio $\frac{d_1}{d_2}$ between two distances $d_1>d_2$ in the point set\\
	\hline
	$d_j$ & $d_j = d(c,j)$ only when the center $c$ is clear from context\\
	\hline
	$\dratio$ & $\dratio=\max_{u,v,x,y\in V}\frac{d(u,v)}{d(x,y)}$ (the ratio of the largest to smallest distances)\\
	\hline
	$D_i(\alpha)$ & $D_i(\alpha)=\sum_{j=1}^i d_j^\alpha$ \\
	\hline
	$\#I$ & The number of breakpoints of $\lloyds$ \\
	& (the number of discontinuities in $\clus_{\alpha,\beta}\left(\V,\vec{Z}\right)$)\\
	\hline
	$\#I_{t,\ell}$ & The number of breakpoints in round $t$ and interval $[\alpha_\ell,\alpha_{\ell+1}]$ (the number \\
	& of times the choice of $t$'th center changes as we vary $\alpha$ along $[\alpha_\ell,\alpha_{\ell+1}]$)\\
	\hline
	$E_{t,j}$ & the event that $\frac{D_j(\alpha_{\ell})}{D_n(\alpha_\ell)}<z_t<\frac{D_i(\alpha_{\ell+1})}{D_n(\alpha_{\ell+1})}$\\
	& (we achieve the maximum number of breakpoints given $v_i$ is the $t$'th center)\\
	\hline
	$E_{t,j}'$ & the event that $\frac{D_i(\alpha_{\ell+1})}{D_n(\alpha_{\ell+1})}<z_t<\frac{D_{j+1}(\alpha_{\ell})}{D_n(\alpha_\ell)}$\\
	& (we do not achieve the max number of breakpoints given $v_i$ is the $t$'th center)\\
	\hline
\end{tabular}
\end{center}
\label{tab:notation}
\end{table*}

\section{Details from Section \ref{sec:alpha}} \label{app:theory}

In this section, we give details and proofs from Section \ref{sec:alpha}.

% \medskip
% \noindent \textbf{Lemma~\ref{lem:farthest-first} (restated).}
% \emph{
% Given a clustering instance $\V$ and $\delta>0$,
% if $\alpha>\frac{\log\left(\frac{nk}{\delta}\right)}{\log s}$,
% then $d^\alpha$-sampling will give the same output as farthest-first traversal with probability $>1-\delta$.
% Here, $s$ denotes the minimum ratio $d_1/d_2$ between two distances $d_1>d_2$ in the point set.
% }
\lemFarthestFirst*

\begin{proof}  %[Proof of Lemma \ref{lem:farthest-first}]
Given such a clustering instance $\V=(V,d,k)$ and $\alpha$,
first we note that farthest-first traversal (i.e. $d^\infty$-sampling) and $d^\alpha$-sampling both start by picking a center
uniformly at random from $V$. Assume both algorithms have chosen initial center $v_1$, and let $C=\{v_1\}$ denote the
set of current centers.
In rounds 2 to $n$, farthest-first traversal deterministically chooses the center $u$ which maximizes $d_{\min}(u,C)$
(breaking ties uniformly at random).
We will show that with high probability, in every round, $d^\alpha$-sampling will also choose the center maximizing $d_{\min}(u,C)$
or break ties at random.
In round $t$, let $d_t=\text{max}_{u\in V}d_{\min}(u,C)$ (assuming $C$ are the first $t-1$ centers chosen by farthest-first traversal).
Let $d_t'$ denote the largest distance smaller than $d_t$, so by assumption, $d_t>s\cdot d_t'$.
Assume there are $x$ points whose minimum distance to $C$ is $d_t$.
Then the probability that $d^\alpha$-sampling will fail to choose one of these points is at most
\begin{align*}
\frac{(n-x)d_t'^\alpha}{(n-x)d_t'^\alpha+x(s d_t')^\alpha}&\leq \frac{n-x}{n-x+x\cdot s^\alpha}\\
&\leq \frac{n}{s^\alpha}
\end{align*}

Over all $k$ rounds of the algorithm, the probability that $d^\alpha$-sampling will deviate from farthest-first traversal
(assuming they start with the same first choice of a center and break ties at random in the same way) is at most
$\frac{nk}{s^\alpha}$, and if we set this probability smaller than $\delta$ and solve for $\alpha$,
we obtain

$$\alpha>\frac{\log\left(\frac{nk}{\delta}\right)}{\log s}.$$
\end{proof}

Now we define a common stability assumption called separability \cite{kobren2017online,pruitt2011ncbi},
which states that there exists a value $r$ such that all points in the same cluster have distance less than $r$ and all points in different clusters have distance greater than $r$.

\begin{definition}
A clustering instance satisfies $(1+c)$-separation if $(1+c)\max_{i\mid u,v\in C_i}d(u,v)<\min_{j\neq j'\mid u\in C_j,v\in C_{j'}} d(u,v)$.
\end{definition}

Now we show that under $(1+c)$-separation, $d^\alpha$-sampling will give the same output as farthest-first traversal with high probability if $\alpha>\log n$,
even for $c=.1$.

\begin{lemma} \label{lem:stability}
Given a clustering instance $\V$ satisfying $(1+c)$-separation, and $0<\delta$,
then if $\alpha>\frac{1}{c}\cdot\left(\log n + \log \frac{1}{\delta}\right)$, with probability $>1-\delta$,
$d^\alpha$-sampling with Lloyd's algorithm will output the optimal clustering.
\end{lemma}

\begin{proof}
Given $\V$ satisfying $(1+c)$-separation,
we know there exists a value $r$ such that for all $i$, for all $u,v\in C_i$, $d(u,v)<r$,
and for all $u\in C_j$, $v\in C_{j'\neq j}$, $d(u,v)>(1+c)r$.
WLOG, let $r=1$.
Consider round $t$ of $d^\alpha$-sampling
and assume that each center $v$ in the current set of centers $C$ is from a unique cluster in the optimal solution.
Now we will bound the probability that the center chosen in round $t$ is not from a unique cluster in the optimal solution.
Given a point $u$ from a cluster already represented in $C$, there must exist $v\in C$ such that
$d(u,v)<1$, so $d_{min}(u,C)<1$. Given a point $u$ from a new cluster, it must be the case that $d_{min}(u,C)>(1+c)$.
The total number of points in represented clusters is $<n$.
Then the probability we pick a point from an already represented cluster is
$\leq\frac{tn}{tn+c^\alpha}$. If we set $\frac{tn}{tn+c^\alpha}\leq \frac{\delta}{k}$ and solve for $\alpha$,
we obtain $\alpha>\frac{\log n+\log\frac{1}{\delta}}{\log c}\leq \frac{1}{c}\cdot\left(\log n + \log \frac{1}{\delta}\right)$.
Since this is true for an arbitrary round $t$, and there are $k$ rounds in total, we may union bound over all rounds
to show the probability $d^\alpha$-sampling outputting one center per optimal clustering is $>1-\delta$.
Then, using $(1+c)$-separation, the Voronoi tiling of these centers must be the optimal clustering, so Lloyd's algorithm
converges to the optimal solution in one step.
\end{proof}

Next, we give the formal proof of Theorem \ref{thm:alpha-just}.

\medskip
\noindent \textbf{Theorem~\ref{thm:alpha-just} (restated).}
\emph{
For $\alpha^*\in [.01,\infty)\cup\{\infty\}$ and $\beta^*\in [1,\infty)\cup\{\infty\}$, there exists a clustering instance $\V$
whose target clustering is the optimal $\ell_{\beta^*}$ clustering,
such that $\clus_{\alpha^*,\beta^*}(\V)<\clus_{\alpha,\beta}(\V)$ for all $(\alpha,\beta)\neq (\alpha^*,\beta^*)$.
}

\begin{proof}
Consider $\alpha^*,\beta^*\in [0,\infty)\cup\{\infty\}$.
The clustering instance consists of 6 clusters, $C_1,\dots,C_6$.
The target clustering will be the optimal $\ell_{\beta^*}$ objective.
The basic idea of the proof is as follows.

First, we show that for all $\beta$ and $\alpha\neq\alpha^*$, $\clus_{\alpha^*,\beta}(\V)<\clus_{\alpha,\beta}(\V)$.
We use clusters $C_1,\dots,C_4$ to accomplish this. We set up the distances so that
$d^{\alpha^*}$ sampling is more likely to sample one point per cluster than any other value of $\alpha$.
If the sampling does not sample one point per cluster, then it will fall into a high-error local minima trap that
$\beta$-Lloyd's method cannot escape, for any value of $\beta$.
Therefore, $d^{\alpha^*}$ sampling is more effective than any other value of $\alpha$.

Next, we use clusters $C_5$ and $C_6$ to show that if we start with one center in $C_5$ and one in $C_6$, then
$\beta^*$-Lloyd's method will strictly outperform any other value of $\beta$.
We accomplish this by adding three choices of centers for $C_5$.
Running $\beta^*$-Lloyd's method will return the correct center, but any other value of $\beta$ will return suboptimal centers which incur error
on $C_5$ and $C_6$. Also, we show that Lloyd's method returns the same centers on $C_1,\dots,C_4$, independent of $\beta$.

For the first part of the proof, we define three cliques (see Figure \ref{fig:alpha-just}).
The first two cliques are $C_1$ and $C_2$, and the third clique is $C_3\cup C_4$.
$C_1$ contains $w_1$ points at distance $x>1$, and $C_2$ contains $w_2$ points at distance $\frac{1}{x}$.
We set $w_2=x^{2\alpha^*} w_1$. The last clique contains $w$ points at distance 1.
Since the cliques are very far apart, the first three sampled centers will each be in a different clique, with high probability (for $\alpha>.01)$.
The probability of sampling a 4th center $x_4$ in the third clique, for $\alpha=\alpha^*+\delta$ is equal to

$$\frac{w}{w+(x^{\alpha^*+\delta}+x^{\alpha^*-\delta})w_2}.$$

Since $x^{\alpha^*+\delta}+x^{\alpha^*-\delta}$ is minimized when $\delta=0$, this probability is maximized when $\alpha=\alpha^*$.
Now we show that the error will be much larger when $x_4$ is not in the third clique.
We add center $c_1$ which is distance $x-\epsilon$ to all points in $C_1$. We also add centers $b_1$ and $b_1'$ which are distance $x-2\epsilon$
to $B_1$ and $B_1'$ such that $B_1$ and $B_1'$ form a partition of $C_1$.
Similarly, we add centers $c_2$, $b_2$, and $b_2'$ at distance $\frac{1}{x}-\epsilon$ and $\frac{1}{x}-2\epsilon$ to $C_2$, $B_2$, and $B_2'$,
respectively, such that $B_2$ and $B_2'$ form a partition of $C_2$.
Finally, we add $c_3$ and $c_4$ which are distance .5 to $C_3$ and $C_4$, respectively,
and we add $b_3$ which is distance $1-\epsilon$ to $C_3\cup C_4$.
Then the optimal centers for any $\beta$ must be $c_1,c_2,c_3,c_4$,
and this will be the solution of $\beta$-Lloyd's method, as long as the sampling procedure returned one point in the first two cliques,
and two points in the third clique.
If the sampling procedure returns two points in the first clique or second clique, then $\beta$-Lloyd's method will return
$b_1,b_1',c_2,b_3$ or $c_1,b_2,b_2',b_3$, respectively. This will incur error $w/2+w_1/2$ or $w/2+w_2/2$, since we set the target clustering
to be the optimal $\ell_{\beta^*}$ objective which is equal to $\{C_1,C_2,C_3,C_4\}$.
Note that we have set up the distances so that Lloyd's method is independent of $\beta$.
Therefore, the expected error is equal to

$$\frac{w}{w+(x^{\alpha^*+\delta}+x^{\alpha^*-\delta})w_2}\cdot (w/2+\min(w_1,w_2)/2).$$

This finishes off the first part of the proof.
Next, we construct $C_5$ and $C_6$ so that $\beta^*$-Lloyd's method will return the best solution, assuming the sampling returned one point in
$C_5$ and $C_6$. Later we will show how adding these clusters does not affect the previous part of the proof.
We again define three cliques.
The first clique is size $\frac{2}{3}w_5$ and distance $.1$, the second clique is size $\frac{1}{3}w_5$ and distance $.1$, and the third clique is size
$w_6$ and distance $.1$. The first two cliques are $C_5$, and the second clique is $C_6$.
The distance between the first two cliques is $.2$, and the distance between the first two and the third clique is $1000$.
Now imagine the first two cliques are parallel to each other, and there is a perpendicular bisector which contains possible centers for $C_5$.
I.e., we will consider possible centers $c$ for $C_5$ where the $\ell_\beta$ cost of $c$ is
$\frac{2}{3}w_5\cdot z^\beta+\frac{1}{3}w_5(.2-z)^\beta$ for some $0\leq z\leq .2$, . For $\beta\in (0,\infty)$, the $\beta$ which minimizes the expression must be in $[0,.2]$. We set $c_5$ corresponding to the $z$ which minimizes the expression for $\beta^*$, call it $z^*$.
Therefore, $\beta^*$-Lloyd's method will output $c_5$. We also set centers $b_5$ and $b_5'$ corresponding to $z^*-\epsilon$ and $z^*+\epsilon$.
Therefore, any value of $\beta$ slightly above or below $\beta^*$ will return a different center.
We add a center $C_6$ at distance $.1-\epsilon$ to the third clique. This is the only point we add, so it will always be chosen by
$\beta$-Lloyd's method for all $\beta$. Finally, we add two points $p_1$ and $p_2$ in between the second and third
cliques. We set the distances as follows. $d(c_5,p_1)=d(c_6,p_2)=500-\epsilon$, $d(c_5,p_2)=d(c_6,p_1)=500$,
$d(b_5,p_1)=500+\epsilon$, and $d(b_5',p_2)=500-2\epsilon$.
Since the weight of these two points are very small compared to the cliques, these points will have no effect on the prior sampling and Lloyd's method analyses.
The optimal clustering for the $\ell_{\beta^*}$ objective is to add $p_1$ to $C_5$ and $p_2$ to $C_6$.
However, running $\beta$-Lloyd's method for $\beta$ smaller or larger than $\beta^*$ will return center $b_5$ or $b_5'$ and incur error 1
by mislabeling $p_1$ or $p_2$.

Since all cliques from both parts of the proof are 1000 apart, with high probability, the first 5 sampled points will
be in cliques $C_1,C_2,C_3\cup C_4,C_5,$ and $C_6$. Since the cliques from the second part of the proof are distance .2,
while in the first part they were $>\frac{1}{x}$ apart, we can set the variables $x,w_1,w_2,w,w_5,w_6$ so that with high probability,
the sixth sampled point will not be in $C_5$ or $C_6$.
Therefore, the constructions in the second part do not affect the first part.
This concludes the proof.
\end{proof}

%%%%%%%%%%%%%%% combinatorial upper bound

Now we upper bound the number of discontinuities of  $\seed_\alpha(\V,\vec{Z})$.
Recall that $\seed_\alpha(\V,\vec{Z})$ denotes the outputted centers from phase 1 of Algorithm \ref{alg:clus} on instance $\V$ with randomness $\vec{Z}$.
For the first phase, our main structural result is to show that for a given clustering instance and value of $\beta$,
with high probability over the randomness in Algorithm \ref{alg:clus},
the number of discontinuities of the cost function $\clus_{\alpha,\beta}\left(\V,\vec{Z}\right)$ as we vary $\alpha\in [0,\ahi]$
is $O(nk(\log n)\ahi)$.
Our analysis crucially harnesses the randomness in the algorithm to achieve this bound.
In contrast, a combinatorial approach would only achieve a bound of $n^{O(k)}$,
which is the total number of sets of $k$ centers.
For completeness, we start with a combinatorial proof of $O(n^{k+3})$ discontinuities.
Although Theorem \ref{thm:combinatorial} is exponential as opposed to Theorem \ref{thm:expect}, it holds with probability 1 and
has no dependence on $\ahi$.
First, we need to state a consequence of Rolle's Theorem.

\begin{theorem}[ex. \cite{tossavainen}]\label{thm:roots}
Let $f$ be a polynomial-exponential sum of the form $f(x) = \sum_{i = 1}^N a_i b_i^x$, where $b_i > 0$, $a_i \in \R$, and at least one $a_i$ is non-zero. The number of roots of $f$ is upper bounded by $N$.
\end{theorem}

Now we are ready to prove the combinatorial upper bound.

\begin{theorem} \label{thm:combinatorial}
Given a clustering instance $\V$ and vector $\vec{Z}\in [0,1]^k$, the number of discontinuities
of $\seed_\alpha(\V,\vec{Z})$ as a function of $\alpha$ over $[0,\infty)\cup\{\infty\}$ is $O\left(\min\left(n^{k+3},n^2 2^n\right)\right)$.
\end{theorem}

\begin{proof}
Given a clustering instance $\V$ and a vector $\vec{Z}$,
consider round $t$ of the $d^\alpha$ seeding algorithm.
Recall that the algorithm decides on the next center based on $z_t\in [0,1]$, $\alpha$,
and the distance from each point to the set of current centers.
The idea of this proof will be to count the number of $\alpha$ intervals such that
within one interval, all possible decisions the algorithm makes are fixed.
In round $t$, there are ${n \choose t-1}$ choices for the set of current centers $C$.
%Recall that the decision for the next center depends on $z_t\in [0,1]$.
We denote the points $V=\{v_1,\dots,v_n\}$ and WLOG assume the algorithm orders the intervals
$I_{v_1},I_{v_2},\dots,I_{v_n}$.
Given a point $v_i\in V$, it
will be chosen as the next center if and only if $z_t$ lands in its interval, formally,
\begin{equation*}
\frac{\sum_{j=1}^{i-1}d_{\min}(v_j,C)^\alpha}{\sum_{j=1}^{n}d_{\min}(v_j,C)^\alpha}<z_t<
\frac{\sum_{j=1}^{i}d_{\min}(v_j,C)^\alpha}{\sum_{j=1}^{n}d_{\min}(v_j,C)^\alpha}
\end{equation*}
%\begin{align*} %arxiv version
%\frac{\sum_{j=1}^{i-1}d_{\min}(v_j,C)^\alpha}{\sum_{j=1}^{n}d_{\min}(v_j,C)^\alpha}&<z_t<
%\frac{\sum_{j=1}^{i}d_{\min}(v_j,C)^\alpha}{\sum_{j=1}^{n}d_{\min}(v_j,C)^\alpha}\\
%\implies \sum_{j=1}^{i-1}d_{\min}(v_j,C)^\alpha&<z_t\cdot\sum_{j=1}^{n}d_{\min}(v_j,C)^\alpha<\sum_{j=1}^{i}d_{\min}(v_j,C)^\alpha.
%\end{align*}

By Theorem \ref{thm:roots}, these two equations have at most $n+1$ roots each.
Therefore, in the $n+2$ intervals of $\alpha$ between each root, the decision whether or not to choose $v_i$ as a center in round $t$ is fixed.
Note that the center set $C$, the point $v_i$, and the number $z_t$ fixed the coefficients of the equation.
Since in round $t$, there are ${n\choose t}$ choices of centers, then there are $\sum_{t=1}^k {n\choose k}\cdot n\cdot 2$ total equations which determine the
outcome of the algorithm. Each equation has at most $n+1$ roots, so it follows there are $1+\sum_{t=1}^k {n\choose k}\cdot n\cdot 2(n+1)\in O(n^3\cdot n^k)$ total intervals
of $\alpha$ along $[0,\infty)\cup\{\infty\}$ such that within each interval, the entire outcome of the algorithm, $\seed_\alpha(\V,\vec{Z})$, is fixed.
Note that we used $k\cdot n^k$ to bound the total number of choices for the set of current centers.
We can also bound this quantity by $2^n$, since each point is either a center or not a center.
This results in a final bound of $O\left(\min\left(n^{k+3},n^2 2^n\right)\right)$.
\end{proof}

%%%%%%%%%%%%%%%%%%%%%%%%

\medskip\noindent\textbf{Lemma~\ref{lem:monotoneBins} (restated).}
\emph{
	Assume that $v_1$, \dots, $v_n$ are sorted in decreasing distance from a set $C$ of centers.
	Then for each $i=1,\dots,n$, the function
	$\alpha \mapsto\frac{D_i(\alpha)}{D_n(\alpha)}$ is monotone increasing and continuous
	along $[0,\infty)$.
	Furthermore, for all $1\leq i\leq j\leq n$ and $\alpha\in[0,\infty)$, we have $\frac{D_i(\alpha)}{D_n(\alpha)}\leq\frac{D_j(\alpha)}{D_n(\alpha)}$.
}
\begin{proof}
  Recall that $D_i(\alpha) = \sum_{j=1}^i \dmin(v_{(i)}, C)^\alpha$, where
  $v_1$, \dots, $v_n$ are the points sorted in decreasing order of
  distance to the set of centers $C$.

  First we show that for each $i$, the function $\alpha \mapsto D_i(\alpha)/D_n(\alpha)$
	is monotone increasing. Given $\alpha_1<\alpha_2$, we
  must show that for each $i$,
  \[
  \frac{D_i(\alpha_1)}{D_n(\alpha_1)} \leq \frac{D_i(\alpha_2)}{D_n(\alpha_2)}.
  \]
  This is equivalent to showing
  \[
  D_i(\alpha_1)D_n(\alpha_2)\leq D_i(\alpha_2)D_n(\alpha_1).
  \]
  Using the shorthand notation $d_j = d(v_j,C)$, we have
  \begin{align*}
  D_i(\alpha_1)D_n(\alpha_2)
  =& \left(\sum_{j=1}^i d_j^{\alpha_1}\right)\left(\sum_{j=1}^n d_j^{\alpha_2}\right) \\
  =& \left(\sum_{j=1}^i d_j^{\alpha_1}\right)\left(\sum_{j=1}^i d_j^{\alpha_2}\right)
  +\left(\sum_{j=1}^i d_j^{\alpha_1}\right)\left(\sum_{j=i+1}^n d_j^{\alpha_2}\right) \\
  =& \left(\sum_{j=1}^i d_j^{\alpha_1}\right)\left(\sum_{j=1}^i d_j^{\alpha_2}\right)
  +\sum_{j=1}^i\sum_{k=i+1}^n d_j^{\alpha_1} d_k^{\alpha_2} \\
  \leq& \left(\sum_{j=1}^i d_j^{\alpha_1}\right)\left(\sum_{j=1}^i d_j^{\alpha_2}\right)
  +\sum_{j=1}^i\sum_{k=i+1}^n d_j^{\alpha_1} d_k^{\alpha_2} \left(\frac{d_j}{d_k}\right)^{\alpha_2-\alpha_1} \\
  =& \left(\sum_{j=1}^i d_j^{\alpha_1}\right)\left(\sum_{j=1}^i d_j^{\alpha_2}\right)
  +\sum_{j=1}^i\sum_{k=i+1}^n d_j^{\alpha_2} d_k^{\alpha_1} \\
  =& \left(\sum_{j=1}^i d_j^{\alpha_2}\right)\left(\sum_{j=1}^n d_j^{\alpha_1}\right)\\
  =& D_i(\alpha_2)D_n(\alpha_1),
  \end{align*}
  as required.

	Next we show that $\alpha \mapsto D_i(\alpha)/D_n(\alpha)$ is continuous along $[0,\infty)$.
	$D_i(\alpha)$ and $D_n(\alpha)$ are both sums of simple exponential functions, so they are
	continuous. $D_n(\alpha)$ is always at least $n$ along $[0,\infty)$, therefore,
	$D_i(\alpha)/D_n(\alpha)$ is continuous.

	Finally, we show that for all $1\leq i\leq j\leq n$, we have $\frac{D_i(\alpha)}{D_n(\alpha)}\leq\frac{D_j(\alpha)}{D_n(\alpha)}$.
	Given $1\leq i\leq j\leq n$, then
	\begin{equation*}
	\frac{D_j(\alpha)}{D_n(\alpha)}-\frac{D_i(\alpha)}{D_n(\alpha)}
	=\frac{d_{i+1}^\alpha+\cdots+d_j^\alpha}{D_n(\alpha)}\geq 0
	\end{equation*}

	This completes the proof.
\end{proof}

Now to work up to the proof of Theorem \ref{thm:expect}, we bound the derivative of $\frac{D_i(\alpha)}{D_n(\alpha)}$.

\begin{lemma} \label{lem:deriv}
  Given distances $d_1\geq\cdots\geq d_n$, for any index $i$ and $\alpha > 0$,
  we have
  \[
  \left|\frac{\partial}{\partial\alpha}
  \left(\frac{D_i(\alpha)}{D_n(\alpha)}\right)\right|\leq
  \min\left\{\frac{2}{\alpha}\log n , \log \frac{d_1}{d_n} \right\}.
  \]
\end{lemma}

\begin{proof}
  For any $i \in [n]$, the derivative of $D_i(\alpha)$ is given by $D'_i(\alpha)
  = \sum_{j=1}^i d_j^\alpha \log d_j$. With this,
\begin{align}
\frac{\partial}{\partial\alpha} \left(\frac{D_i(\alpha)}{D_n(\alpha)}\right)
&= \frac{D'_i(\alpha)D_n(\alpha)-D'_n(\alpha)D_i(\alpha)}{(D_n(\alpha))^2} \nonumber \\
&= \frac{\left(\sum_{x=1}^i d_x^\alpha\log d_x\right)\left(\sum_{y=1}^n d_y^\alpha\right)
-\left(\sum_{x=1}^i d_x^\alpha\right)\left(\sum_{y=1}^n d_y^\alpha\log d_y \right)}{\sum_{x=1}^n\sum_{y=1}^n d_x^\alpha d_y^\alpha} \nonumber \\
&= \frac{\sum_{x=1}^i \sum_{y=1}^n \left( d_x^\alpha d_y^\alpha \log d_x \right)
-\sum_{x=1}^{i} \sum_{y=1}^n \left( d_x^\alpha d_y^\alpha \log d_y \right)}
{\sum_{x=1}^n\sum_{y=1}^n d_x^\alpha d_y^\alpha} \nonumber \\
&= \frac{\sum_{x=1}^i \sum_{y=i+1}^n \left( d_x^\alpha d_y^\alpha \log d_x \right)
-\sum_{x=1}^{i} \sum_{y=i+1}^n \left( d_x^\alpha d_y^\alpha \log d_y \right)}
{\sum_{x=1}^n\sum_{y=1}^n d_x^\alpha d_y^\alpha} \nonumber \\
&= \frac{\sum_{x=1}^i \sum_{y=i+1}^n \left( d_x^\alpha d_y^\alpha \log\left(\frac{d_x}{d_y}\right)\right)}
{\sum_{x=1}^n\sum_{y=1}^n d_x^\alpha d_y^\alpha} \label{eq:deriv}
\end{align}
At this point, because $d_1\geq\cdots\geq d_n\geq 0$, we achieve our second bound as follows.
\begin{equation}
\frac{\sum_{x=1}^i \sum_{y=i+1}^n \left( d_x^\alpha d_y^\alpha \log\left(\frac{d_x}{d_y}\right)\right)}
{\sum_{x=1}^n\sum_{y=1}^n d_x^\alpha d_y^\alpha}
\leq  \frac{\sum_{x=1}^i \sum_{y=i+1}^n d_x^\alpha d_y^\alpha}
{\sum_{x=1}^n\sum_{y=1}^n d_x^\alpha d_y^\alpha}
\cdot \log\frac{d_1}{d_n}
\leq \log \frac{d_1}{d_n}.
\label{eq:deriv_part_1}
\end{equation}
Recall our goal is to bound the derivative by a minimum of two quantities. Equation \ref{eq:deriv_part_1} gives us the second bound.
To achieve the first bound, we bound Equation \ref{eq:deriv} a different way, as follows.
\begin{align}
\frac{\sum_{x=1}^i \sum_{y=i+1}^n \left( d_x^\alpha d_y^\alpha \log\left(\frac{d_x}{d_y}\right)\right)}
{\sum_{x=1}^n\sum_{y=1}^n d_x^\alpha d_y^\alpha}
%
% &\leq\frac{\frac{1}{\alpha}\cdot\sum_{x=1}^i \sum_{y=i+1}^n \left( d_x^\alpha d_y^\alpha \log\left(\frac{d_x^\alpha}{d_y^\alpha}\right)\right)}
% {\sum_{x=1}^n\sum_{y=1}^n d_x^\alpha d_y^\alpha} \nonumber \\
% %
&\leq\frac{1}{\alpha}\cdot\frac{\sum_{x=1}^i d_x^\alpha \left( \sum_{y=i+1}^n d_y^\alpha \log\left(\frac{d_x^\alpha}{d_y^\alpha}\right)\right)}
{\sum_{x=1}^n \sum_{y=1}^n d_x^\alpha  d_y^\alpha }
\label{eq:deriv_part_2_cases}
\end{align}
To upper bound \eqref{eq:deriv_part_2_cases}, we will show that $\sum_{y=i+1}^n
d_y^\alpha\log\left(\frac{d_1^\alpha}{d_y^\alpha}\right)\leq 2\log n\sum_{y=1}^n
d_y^\alpha$. We bound each term of the sum in one of two cases:
\begin{description}
  \item[Case 1:] If $y$ is such that $d_y^\alpha \geq \frac{d_1^\alpha}{n^{2}}$, then
  $d_y^\alpha\log\left(\frac{d_1^\alpha}{d_y^\alpha}\right)\leq 2d_y^\alpha\log
  n$.
  \item[Case 2:] If $y$ is usch that $d_y^\alpha<\frac{d_1^\alpha}{n^2}$, then
  $d_y^\alpha\log\left(\frac{d_1^\alpha}{d_y^\alpha}\right)\leq
  d_1^\alpha\left(\frac{d_y^\alpha}{d_1^\alpha}\log\frac{d_1^\alpha}{d_y^\alpha}\right)\leq\frac{1}{n}\cdot
  d_1^\alpha$, where the last inequality follows because $\frac{1}{x}\log x\leq
  \frac{1}{n}$ for all $x>n^2$.
\end{description}

Let $C_1 = \{ y \geq i + 1\,|\, d_y^\alpha \geq \frac{d_1^\alpha}{n^{2}} \}$
and $C_2 = \{ y \geq i+1 \,|\, d_y^\alpha<\frac{d_1^\alpha}{n^2} \}$ be the
values of $y$ corresponding to Cases 1 and 2, respectively. Then we have
\[
  \sum_{y=i+1}^n d_y^\alpha \cdot \log \frac{d_x^\alpha}{d_y^\alpha}
  =
  \sum_{y \in C_1} d_y^\alpha \cdot \log \frac{d_x^\alpha}{d_y^\alpha}
  +
  \sum_{y \in C_2} d_y^\alpha \cdot \log \frac{d_x^\alpha}{d_y^\alpha}
  \leq
  2 \log(n) \sum_{y=2}^n d_y^\alpha
  +
  d_1^\alpha
  \leq
  2 \log(n) \sum_{y=1}^n d_y^\alpha.
\]
Substituting this into \eqref{eq:deriv_part_2_cases}, we have that
\[
  \frac{\partial}{\partial \alpha} \frac{D_i(\alpha)}{D_n(\alpha)}
  \leq
  \frac{2}{\alpha} \log(n) \cdot
  \frac{\sum_{x=1}^i \sum_{y=i+1}^n d_x^\alpha d_y^\alpha}
  {\sum_{x=1}^n \sum_{y=1}^n d_x^\alpha d_y^\alpha }
  \leq  \frac{2}{\alpha} \log n,
\]
as required.

% So, summing over all $y\in [i+1,n]$, it follows that
% $\sum_{y=i+1}^n d_y^\alpha \log\left(\frac{d_x^\alpha}{d_y^\alpha}\right) \leq (2\log n) \sum_{y=i+1}^n d_y^\alpha$.
% Therefore, Equation \ref{eq:deriv_part_2_cases} can be bounded by $\frac{2}{\alpha}\log n$.
% By combining this bound with Equation \ref{eq:deriv_part_1}, we achieve the desired result.
%
%So, for all $y$ in case 1, $d_y^\alpha\log\left(\frac{d_1^\alpha}{d_y^\alpha}\right)\leq 2d_y^\alpha\log n$,
%and $\sum_{y:\text{case }2}d_y^\alpha\log\left(\frac{d_1^\alpha}{d_y^\alpha}\right)\leq d_1^\alpha$.
%It follows that for all $i\geq 1$, $\sum_{y=i+1}^n d_y\log\left(\frac{d_1^\alpha}{d_y^\alpha}\right)\leq 2\log n\sum_{y=1}^n d_y$.
%Therefore, $\frac{\partial}{\partial\alpha} \left(\frac{D_i(\alpha)}{D_n(\alpha)}\right)\leq 4\log n$.
\end{proof}

% \medskip
% \noindent \textbf{Theorem~\ref{thm:expect} (restated).}
% \emph{
% Given a clustering instance $\V$, the expected number of discontinuities
% of $\seed_\alpha(\V,\vec{Z})$ as a function of $\alpha$ over $[0,\ahi]$ is $O(nk(\log n)(\log \ahi + \log\log D))$.
% Here, the expectation is over the uniformly random draw of $\vec{Z}\in [0,1]^k$ and $D$ denotes the diameter of the clustering instance.
% }
\thmExpect*

\begin{proof}
Given $\V$, we will show that $\E[\#I]\leq nk\log n\log\frac{\ahi}{\alo}$ over $[\alo,\ahi]$
and $\E[\#I]\leq 4nk\log n(1 + \log \alpha_h + \log\log D)$ over $[0,\ahi]$,
where $\#I$ denotes the total number of discontinuities of $\seed_\alpha(\V,\vec{Z})$
and the expectation
is over the draw of $\vec{Z}\in[0,1]^k$.
Consider round $t$ of a run of Algorithm \ref{alg:clus}.
Suppose at the beginning of round $t$, there are $L$ possible states of the algorithm, e.g., $L$ sets of $\alpha$
such that within a set, the choice of the first $t-1$ centers is fixed.
By Lemma \ref{lem:monotoneBins}, we can write these sets as
$[\alpha_0,\alpha_1],\dots,[\alpha_{L-1},\alpha_L]$, where $0=\alpha_0<\cdots<\alpha_L=\ahi$.
Given one interval, $[\alpha_{\ell},\alpha_{\ell+1}]$,
we claim the expected number of new breakpoints $\#I_{t,\ell}$ by choosing a center in round $t$ is bounded by
$$\min\left(2n\log \dratio (\alpha_{\ell+1}-\alpha_\ell), n-t-1, 4n\log n(\log\alpha_{\ell+1}-\log\alpha_\ell)\right).$$
Note that $\#I_{t,\ell}+1$ is the number of possible choices for the next center in round $t$ using $\alpha$ in $[\alpha_{\ell},\alpha_{\ell+1}]$.

The claim gives three different upper bounds on the expected number of new breakpoints,
where the expectation is only over $z_t$ (the uniformly random draw from $[0,1]$ used by Algorithm \ref{alg:clus} in round $t$),
and the bounds hold for any given configuration of $d_1\geq\cdots\geq d_n$.
To prove the first statement in Theorem~\ref{thm:expect}, we only need the last of the three bounds, and to prove the second statement,
we need all three bounds.
First we show how to prove these statements assuming the claim, and later we will prove the claim.
We prove the first statement as follows.
Let $\#I$ denote the total number of discontinuities of $\seed_\alpha(\V,\vec{Z})$ along $[\alo, \ahi]$.

\begin{align*}
E_{Z\in [0,1]^k}[\#I]&\leq E_{Z\in [0,1]^k} \left[\sum_{t=1}^k\sum_{\ell=1}^{L-1} (\#I_{t,\ell})\right]\\
&\leq \sum_{t=1}^k \sum_{\ell=0}^{L-1} E_{Z\in [0,1]^k}[\#I_{t,\ell}]\\
&\leq \sum_{t=1}^k \sum_{\ell=0}^{L-1} E_{Z]in [0,1]^k}\left(4n\log n\left(\log \alpha_{\ell+1}-\log\alpha_\ell)\right) \right)\\
&\leq \sum_{t=1}^k \left(4n\log \left(\log \ahi-\log\alo)\right) \right)\\
&\leq k \left(4n\log n \log \frac{\ahi}{\alo} \right)\\
&\leq nk\log n\log \frac{\ahi}{\alo}
\end{align*}

Now we prove the second statement of Theorem~\ref{thm:expect}.
Let $\ell^*$ denote the largest value such that $\alpha_{\ell^*}<\frac{1}{\log \dratio}$. Such an $\ell^*$ must exist because $\alpha_0=0$.
Then we have $\alpha_{\ell^*}<\frac{1}{\log \dratio}\leq\alpha_{\ell^*+1}$.
We use three upper bounds for three different cases of alpha intervals:
the first $\ell^*$ intervals, interval $[\alpha_{\ell^*},\alpha_{\ell^*+1}]$, and intervals $\ell^*+2$ to $L$.
Let $\#I$ denote the total number of discontinuities of $\seed_\alpha(\V,\vec{Z})$ along $[0,\ahi]$.

\begin{align*}
E_{Z\in [0,1]^k}[\#I]&\leq E_{Z\in [0,1]^k} \left[\sum_{t=1}^k\sum_{\ell=1}^{L-1} (\#I_{t,\ell})\right]\\
&\leq \sum_{t=1}^k \sum_{\ell=0}^{L-1} E_{Z\in [0,1]^k}[\#I_{t,\ell}]\\
&\leq \sum_{t=1}^k\left( \sum_{\ell=0}^{\ell^*-1} E_{Z\in [0,1]^k}[\#I_{t,\ell}] + E_{Z\in [0,1]^k}[\#I_{t,\ell^*}]
+ \sum_{\ell=\ell^*+1}^{L-1} E_{Z\in [0,1]^k}[\#I_{t,\ell}] \right)\\
&\leq \sum_{t=1}^k\left( \sum_{\ell=0}^{\ell^*-1} \left(2n\log D (\alpha_{\ell+1}-\alpha_\ell)\right) + (n-t-1)
+ \sum_{\ell=\ell^*+1}^{L-1} \left(4n\log n(\log \alpha_{\ell+1}-\log\alpha_\ell)\right) \right)\\
&\leq \sum_{t=1}^k\left( 2n\log D\cdot \alpha_{\ell^*} + n + 4n\log n\left(\log \alpha_h-\log\alpha_{\ell^*}\right) \right)\\
&\leq \sum_{t=1}^k\left( 2n\log D\cdot \frac{1}{\log D} + n + 4n\log n\left(\log \alpha_h-\log\left(\frac{1}{\log D}\right)\right) \right)\\
&\leq k \left( 2n+n + 4n\log n ( \log \alpha_h + \log\log D)\right) \\
&\leq 4nk\log n(1 + \log \alpha_h + \log\log D).
\end{align*}

%we claim $2n\min\left(\log D, \frac{2}{\alpha_\ell}\log n\right)(\alpha_{\ell+1}-\alpha_\ell)$.
Now we will prove the claim.
Recall that for $\alpha$-sampling, each point $x$ receives an interval along $[0,1]$ of size $\frac{d_x^\alpha}{D_n(\alpha)}$, so
the number of breakpoints in round $t$
along $[\alpha_{\ell},\alpha_{\ell+1}]$ corresponds to the
number of times $z_t$ switches intervals as we increase $\alpha$ from $\alpha_{\ell}$ to
$\alpha_{\ell+1}$.
By Lemma \ref{lem:monotoneBins}, the endpoints of these intervals
are monotone increasing, continuous, and non-crossing, so the number of breakpoints is exactly
$x-y$, where $x$ and $y$ are the minimum indices s.t.\ $\frac{D_x(\alpha_{\ell})}{D_n(\alpha_{\ell})}>z_t$
and $\frac{D_y(\alpha_{\ell+1})}{D_n(\alpha_{\ell+1})}>z_t$, respectively
(see Figure \ref{fig:case2}).
We want to compute the expected value of $x-y$ for $z_t$ uniform in $[0,1]$ (here, $x$ and $y$ are functions of $z_t$).

We take the approach of analyzing each interval individually.
One method for bounding $E_{z_t\in[0,1]}[x-y]$ is to compute the maximum possible number of
breakpoints for each interval $I_{v_j}$, for all $1\leq j\leq n$.
Specifically, if we let $i$ denote the minimum index such that $\frac{D_j(\alpha_{\ell})}{D_n(\alpha_{\ell})}<\frac{D_i(\alpha_{\ell+1})}{D_n(\alpha_{\ell+1})}$, then
\begin{align*}
E[\#I_{t,\ell}]&\leq \sum_{j=1}^n P\left(\frac{D_j(\alpha_{\ell})}{D_n(\alpha_{\ell})}<z_t<\frac{D_{j+1}(\alpha_{\ell})}{D_n(\alpha_{\ell})}\right)\cdot (j-i+1)\\
&\leq \sum_{j=1}^n \frac{d_j^{\alpha_\ell}}{D_n(\alpha_\ell)}\cdot (j-i+1). \label{eq:j-i}
\end{align*}
In this expression, we are using the worst case number of breakpoints within each bucket, $j-i+1$.

We cannot quite use this expression to obtain our bound; for example, when $\alpha_{\ell+1}-\alpha_\ell$ is extremely small, $j-i+1=1$,
so this expression will give us $E[\#I_{t,\ell}]\leq 1$ over $[\alpha_\ell,\alpha_{\ell+1}]$, but we need to show the expected number of breakpoints
is proportional to $\epsilon$ to prove the claim.
To tighten up this analysis, we will show that for each bucket, the probability (over $z_t$) of achieving the maximum number of breakpoints is low.

Assuming that $z_t$ lands in a bucket $I_{v_j}$, we further break into cases as follows.
Let $i$ denote the minimum index such that $\frac{D_i(\alpha_{\ell+1})}{D_n(\alpha_{\ell+1})}>\frac{D_j(\alpha_{\ell})}{D_n(\alpha_{\ell})}$.
Note that $i$ is a function of $j,\alpha_\ell$, and $\alpha_{\ell+1}$, but it is independent of $z_t$.
If $z_t$ is less than $\frac{D_i(\alpha_{\ell+1})}{D_n(\alpha_{\ell+1})}$,
then we have the maximum number of breakpoints possible, since the algorithm chooses center $v_{i-1}$ when $\alpha=\alpha_{\ell+1}$ and it chooses
center $v_j$ when $\alpha=\alpha_\ell$. The number of breakpoints is therefore $j-i+1$, by Lemma \ref{lem:monotoneBins}.
We denote this event by $E_{t,j}$, i.e., $E_{t,j}$ is the event that in round $t$, $z_t$ lands in $I_{v_j}$ and is less than
$\frac{D_i(\alpha_{\ell+1})}{D_n(\alpha_{\ell+1})}$.
If $z_t$ is instead greater than $\frac{D_i(\alpha_{\ell+1})}{D_n(\alpha_{\ell+1})}$,
then the algorithm chooses center $v_i$ when $\alpha=\alpha_{\ell+1}$,
so the number of breakpoints is $\leq j-i$.
We denote this event by $E_{t,j}'$ (see Figure \ref{fig:case2}).
Note that $E_{t,j}$ and $E_{t,j}'$ are disjoint and $E_{t,j}\cup E_{t,j}'$ is the event that $z_t\in I_{v_j}$.

Within an interval $I_{v_j}$,
the expected number of breakpoints is
$$P(E_{t,j})(j-i+1)+P(E_{t,j}')(j-i)=P(E_{t,j}\cup E_{t,j})(j-i)+P(E_{t,j}').$$
We will show that $j-i$ and $P(E_{t,j})$ can both be bounded using Lemma~\ref{lem:deriv}, which finishes off the claim.

First we upper bound $P(E_{t,j})$.
Recall this is the probability that $z_t$ is in between $\frac{D_j(\alpha_{\ell})}{D_n(\alpha_\ell)}$ and $\frac{D_i(\alpha_{\ell+1})}{D_n(\alpha_{\ell+1})}$,
which is
$$\frac{D_i(\alpha_{\ell+1})}{D_n(\alpha_{\ell+1})}-\frac{D_j(\alpha_{\ell})}{D_n(\alpha_\ell)}\leq\frac{D_j(\alpha_{\ell+1})}{D_n(\alpha_{\ell+1})}-\frac{D_j(\alpha_{\ell})}{D_n(\alpha_\ell)}=\frac{D_j(\alpha)}{D_n(\alpha)}\bigg\rvert_{\alpha_\ell}^{\alpha_{\ell+1}}.$$
Recall that Lemma~\ref{lem:deriv} states that
$\left|\frac{\partial}{\partial\alpha} \left(\frac{D_i(\alpha)}{D_n(\alpha)}\right)\right|\leq \min\left(\frac{2}{\alpha}\log n , \log \left(\frac{d_1}{d_n}\right)\right)$.
If we use the second part of the $\min$ expression (and use $\frac{d_1}{d_n}\leq \dratio$)
we get $\frac{D_j(\alpha)}{D_n(\alpha)}\mid_{\alpha_\ell}^{\alpha_{\ell+1}}\leq \log D(\alpha_{\ell+1}-\alpha_\ell)$.
If we use the first part of the $\min$ expression, we get
$$\frac{D_j(\alpha)}{D_n(\alpha)}\bigg\rvert_{\alpha_\ell}^{\alpha_{\ell+1}}\leq 2\log n\int_{\alpha_\ell}^{\alpha_{\ell+1}} \frac{1}{\alpha}\cdot d\alpha
\leq 2\log n\left(\log\alpha\right)\mid_{\alpha_\ell}^{\alpha_{\ell+1}}=2\log n(\log \alpha_{\ell+1}-\log \alpha_\ell).$$

Now we upper bound $j-i$.
Recall that $j-i$ represents the number of intervals between
$\frac{D_i(\alpha_\ell)}{D_n(\alpha_{\ell})}$ and $\frac{D_j(\alpha_\ell)}{D_n(\alpha_\ell)}$ (see Figure \ref{fig:case2}).
Note that the smallest interval in this range is $\frac{d_j^{\alpha_\ell}}{D_n(\alpha_\ell)}$, and
$$\frac{D_j(\alpha_\ell)}{D_n(\alpha_\ell)}-\frac{D_i(\alpha_\ell)}{D_n(\alpha_{\ell})}\leq \frac{D_i(\alpha_{\ell+1})}{D_n(\alpha_{\ell+1})}-\frac{D_i(\alpha_\ell)}{D_n(\alpha_{\ell})}.$$
Therefore, the expected number of breakpoints is at most
$\frac{D_n(\alpha_\ell)}{d_j^{\alpha_\ell}}\cdot\left(\frac{D_i(\alpha_{\ell+1})}{D_n(\alpha_{\ell+1})}-\frac{D_i(\alpha_\ell)}{D_n(\alpha_{\ell})}\right)$, and we can bound the second half of this fraction by again using Lemma \ref{lem:deriv}.
To finish off the proof, we have

\begin{align*}
E[\#I_{t,\ell}]&\leq\sum_j \left(P(E_{t,j}')\cdot (j-i)+P(E_{t,j})\cdot (j-i+1)\right)\\
&\leq \sum_j \left( P(E_{t,j}')\cdot (j-i)+ P(E_{t,j})\cdot(j-i)+P(E_{t,j})\right)\\
&\leq \sum_j \left( P(E_{t,j}'\cup E_{t,j})\cdot (j-i)+ P(E_{t,j})\right)\\
&\leq \sum_j \left( P(z_t\in I_{v_j})\cdot (j-i)\right)+ \sum_j P(E_{t,j})\\
&\leq \sum_j\left(\frac{d_j^{\alpha_\ell}}{D_n(\alpha_\ell)}\right)
\left( \frac{D_n(\alpha_\ell)}{d_j^{\alpha_\ell}}\cdot \frac{D_j(\alpha)}{D_n(\alpha)}\bigg\rvert_{\alpha_\ell}^{\alpha_{\ell+1}}\right)
+\sum_j \left( \frac{D_j(\alpha)}{D_n(\alpha)}\bigg\rvert_{\alpha_\ell}^{\alpha_{\ell+1}}\right)\\
&\leq 2n \left( \frac{D_j(\alpha)}{D_n(\alpha)}\bigg\rvert_{\alpha_\ell}^{\alpha_{\ell+1}} \right)\\
&\leq 2n \min\left( 2\log n(\log \alpha_{\ell+1}-\log \alpha_\ell) , \log D(\alpha_{\ell+1}-\alpha_\ell) \right)
\end{align*}

This accounts for two of the three upper bounds in our claim.
To complete the proof, we note that $E[\#I_{t,\ell}]\leq n-t-1$ simply because there are only $n-t$ centers available to be chosen
in round $t$ of the algorithm (and therefore, $n-t-1$ breakpoints).
\end{proof}

In fact, we can also show that the worst-case number of discontinuities is exponential.

%\medskip
%\noindent \textbf{Lemma~\ref{lem:pdim_lower}(restated).}
%\emph{
\begin{lemma}
Given $n$, there exists a clustering instance $\V$ of size $n$ and a vector $\vec{Z}$ such that the number of discontinuities of
$\seed_\alpha(\V,\vec{Z})$ as a function of $\alpha$ over $[0,2]$ is $2^{\Omega(n)}$.
\end{lemma}
%}

\begin{proof}
We construct
$\V=(V,d,k)$ and $\vec{Z}=\{z_1,\dots,z_k\}$ such that $\seed_\alpha(\V,\vec{Z})$ has $2^{k/3}$ different
intervals in $\alpha$ which give different outputs.

Here is the outline of the construction.
At the start, we set $z_1$ so that one point, $v$, will always be the first center chosen.
Then we add points $a_1,b_1,\dots,a_{k},b_{k}$ such that in round $i$, either $a_i$ or $b_i$ will be chosen as centers.
We carefully set the distances so that for each combinations of centers, there is an $\alpha$ interval which achieves this combination of centers.
Therefore, the total number of $\alpha$ intervals such that the output of the sampling step is fixed, is  $2^{k-1}$.
Our construction also uses points $a_1',b_1',\dots a_k',b_k'$ and $v_1,\dots,v_k$ which are never chosen as centers, but will be crucial in the analysis.

Next, we describe the distances between the points in our clustering instance.
Almost all distances will be set to 100, except for a few distances: for all $i$,
$d(a_i,a_i')=d(b_i,b_i')=\epsilon$, $d(b_i,v_i)=100-o_i$,
$d(a_i,b_i)=2o_i$, $d(v,a_1)=d(v,b_1)=99$, and $d(a_{i-1},a_i)=d(a_{i-1},b_i)=d(b_{i-1},a_i)=d(b_{i-1},b_i)=100-o_i$,
 for $0\leq \epsilon,o_1,\dots,o_{k}\leq 1$ to be specified later.
At the end, we will perturb all other distances by a slight amount $(<\epsilon)$ away from 100, to break ties.

Now we set up notation to be used in the remainder of the proof. We set $z_i=\frac{1}{2}$ for all $i$.
For $1\leq i\leq k$, given $\vec{x}\in\{0,1\}^{i-1}$, let $E_{\vec{x}}$ denote the equation
in round $i$ which determines whether $a_i$ or $b_i$ is chosen as the next center,
in the case where for all $1\leq j<i$, $a_j\in C$ if $\vec{x}_j=0$, or else $b_j\in C$
(and let $E'$ denote the single equation in round 2).
Specifically, $E_{\vec{x}}$ is the following expression

$$ \frac{100^\alpha\left(\frac{n-(i-1)}{2}\right)}{100^\alpha(n-2(i-1))+(100-o_i)^\alpha+\sum_{j=1}^{i-1}(100-\vec{x}_jo_j)^\alpha}.$$

Let $\alpha_{\vec{x}}$ denote the solution to equation $E_{\vec{x}}=\frac{1}{2}$ in $[1,3]$, if it exists.
In the rest of the proof, we must show there exist parameters $\epsilon,o_0,\dots,o_{k}$ which admit
an ordering to the values $\alpha_{\vec{x}}$ which ensures that each $\alpha_{\vec{x}}$ falls in the correct
range to split up each interval, thus achieving $2^{k-1}$ intervals.
The ordering of the $\alpha_{\vec{x}}$'s can be specified by two conditions:
\emph{(1)} $\alpha_{[\vec{x}~0]}<\alpha_{[\vec{x}]}<\alpha_{[\vec{x}~1]}$
and \emph{(2)} $\alpha_{[\vec{x}~0~\vec{y}]}<\alpha_{[\vec{x}~1~\vec{z}]}$
for all $\vec{x},\vec{y},\vec{z}\in\bigcup_{i< k}\{0,1\}^i$ and $|{\bf y}|=|{\bf z}|$.
To prove the $\alpha_{\vec{x}}$'s follow this ordering, we use an inductive argument.
We must show the following claim: there exist $0<o_1,\dots,o_k<1$ such that if we solve $E_{\vec{x}}=\frac{1}{2}$ for
$\alpha_{\vec{x}}$ for all $\vec{x}\in\cup_{i<k}\{0,1\}^i$, then the $\alpha$'s satisfy
$\alpha_{[\vec{x}~0]}<\alpha_{[\vec{x}]}<\alpha_{[\vec{x}~1]}$ and for all $i<k$, $\alpha_{[\vec{x}~1]}<\alpha_{[\vec{y}~0]}$
for $\vec{x},\vec{y}\in\{0,1\}^i$ and $x_1\dots x_i<y_1\dots y_i$.
% prove base case

Given $\vec{x}\in\{0,1\}^i$, for $1\leq i\leq k-1$, let
$p(\vec{x}),n(\vec{x})\in\{0,1\}^i$ denote the vectors which sit on either
side of $\alpha_{\vec{x}}$ in the desired ordering, i.e., $\alpha_{\vec{x}}$ is the only $\alpha_{\vec{y}}$
in the range $(\alpha_{p(\vec{x})},\alpha_{n(\vec{x})})$ such that $|\vec{y}|=i$.
If $\vec{x}=[1\dots 1]$, then set $\alpha_{n(\vec{x})}=3$, and if $\vec{x}=[0\dots 0]$, then set $\alpha_{p(\vec{x})}=1$.

Given $1\leq i\leq k-2$, assume there exist $0<o_1,\dots,o_i<1$ such that the statement is true.
Now we will show the statement holds for $i+1$.
Given $\vec{x}\in\{0,1\}^i$, by assumption, we have that the solution to $E_{\vec{x}}=\frac{1}{2}$ is equal to $\alpha_{\vec{x}}$.
First we consider $E_{[\vec{x}~0]}=\frac{1}{2}$. Note there are two differences in the equations $E_{\vec{x}}$ and $E_{[\vec{x}~0]}$.
First, the number of $100^\alpha$ terms in the numerator decreases by 1, and the number of terms in the denominator decreases by 2.
WLOG, at the end we set a constant $c=\frac{n}{k}$ large enough so that this effect on the root of the equation is negligible for all $n$.
Next, the offset in the denominator changes from $(100-o_i)^\alpha$ to $(100-o_{i+1})^\alpha$.
Therefore, if $0<o_{i+1}<o_i$, then $\alpha_{[\vec{x}~0]}<\alpha_{\vec{x}}$. Furthermore, there exists an upper bound
$0<z_{i+1}<o_i$ such that for all $0<o_{i+1}<z_{i+1}$, we have $\alpha_{[\vec{x}~0]}\in (\alpha_{p(\vec{x})},\alpha_{\vec{x}})$.
Next we consider $E_{[\vec{x}~1]}=\frac{1}{2}$.
As before, the number of $100^\alpha$ terms decrease, which is negligible.
Note the only other change is that an $100^\alpha$ term is replaced with $(100-o_{i+1})^\alpha$.
Therefore, as long as $0<o_{i+1}<o_i$, then $\alpha_{\vec{x}}<\alpha_{[\vec{x}~1]}$, and similar to the previous case,
there exists an upper bound $0<z_{i+1}'<o_i$ such that for all $0<o_{i+1}<z_{i+1}'$, we have
$\alpha_{[\vec{x}~1]}\in (\alpha_{\vec{x}},\alpha_{n(\vec{x})})$.
We conclude that there exists $0<o_{i+1}<\min(z_i,z_i')<o_i$ such that
$\alpha_{p(\vec{x})}<\alpha_{[\vec{x}~0]}<\alpha_{\vec{x}}<\alpha_{[\vec{x}~1]}<\alpha_{n(\vec{x})}$, thus finishing the inductive proof.

%end of induction
Now we have shown that there are $2^{k'}$ nonoverlapping $\alpha$ intervals, such that within every interval,
$d^\alpha$-sampling chooses a unique set of centers, for $k'=k-1$.
To finish our structural claim, we will show that after $\beta$-Lloyd's method, the cost function
$\clus_{\alpha,\beta}(\V,\vec{Z})$ alternates $2^{k'}$ times above and below a value $r$ as $\alpha$ increases.
We add two points, $a$ and $b$, so that  $d(v,a)=d(v,b)=100$, $d(a_k,a)=d(b_k,b)=100-\epsilon$, and the distances from all other
points to $a$ and $b$ are length $100+\epsilon$. Then we add many points in the same location as $v$, $a_i$, and $b_i$, so that any
set $c$ returned by $d^\alpha$-sampling is a local minima for $\beta$-Lloyd's method, for all $\beta$.
Furthermore, these changes do not affect the previous analysis, as long as we appropriately balance the terms in the numerator and denominator of each equation $E_{\vec{x}}$ (and for small enough $\epsilon$).
Finally, we set $v$ and $a$ to have label 1 in the target clustering, and all points are labeled 2.
Therefore, as $d^{\alpha}$-sampling will alternate between $a_k\in C$ and $b_k\notin C$ as we increase $\alpha$,
$a$ and $v$ alternate being in the same or different clusters, so the function $\seed_\alpha(\V,\vec{Z})$ will alternate between different
outputs $2^{\Omega(n)}$ times as a function of $\alpha$.
\end{proof}

Now we give the details for Theorem \ref{thm:pdim_upper}.

\medskip
\noindent \textbf{Theorem~\ref{thm:pdim_upper} (restated).}
\emph{
Given $T\in\mathcal{N}$, a clustering instance $\V$, and a fixed set $C$ of initial centers,
the number of discontinuities of $\lloyd_\beta(\V,C,T)$ as a function of $\beta$ on instance $\V$ is $O(\min(n^{3T},n^{k+3}))$.
}

\begin{proof}
Given a clustering instance $\V$ and a vector $\vec{Z}$,
we bound the number of possible intervals created by the Lloyd's step, given a fixed set of initial centers.
Define $\texttt{lloyd}_\beta(\V,C)$ as the cost of the clustering outputted by the $\beta$-Lloyd iteration
algorithm on $\V$ using initial centers $C$.
Note that the Voronoi partitioning step is only dependent on $C$, in particular, it is independent of $\beta$.
Let $\{C_1,\dots,C_k\}$ denote the Voronoi partition of $V$ induced by $C$.
Given one of these clusters $C_i$, the next center is computed by $\min_{c\in C_i}\sum_{v\in C_i}d(c,v)^\beta$.
Given any $c_1,c_2\in C_i$, the decision for whether $c_1$ is a better center than $c_2$ is governed by
$\sum_{v\in C_i}d(c_1,v)^\beta<\sum_{v\in C_i}d(c_2,v)^\beta$.
Again by Theorem~\ref{thm:roots}, this equation has at most $2n+1$ roots.
Notice that this equation depends on the set $C$ of centers, the choice of a cluster $C_i$,
and the two points $c_1,c_2\in C_i$.
Then there are ${n\choose k}\cdot n\cdot {n\choose 2}$ total equations which fix the outcome of the Lloyd's method,
and there are ${n\choose k}\cdot n\cdot {n\choose 2}\cdot (2n+1)\leq n^{k+4}$ total intervals of $\beta$ such that the outcome of Lloyd's method
is fixed.

Next we give a different analysis which bounds the number of discontinuities by $n^{3T}$, where $T$ is the maximum number of Lloyd's iterations.
By the same analysis as the previous paragraph, if we only consider one round, then the total number of equations which govern the output
of a Lloyd's iteration is ${n\choose 2}$, since the set of centers $C$ is fixed. These equations have $2n+1$ roots, so the total number
of intervals in one round is $O(n^3)$. Therefore, over $T$ rounds, the number of intervals is $O(n^{3T})$.
\end{proof}

Now we give the details for the proofs of the generalized results,
Theorems~\ref{thm:init} and \ref{thm:gen_runtime}.

\medskip
\noindent \textbf{Theorem~\ref{thm:init} (restated).}
\emph{
Given an $\alpha$-parameterized family such that \emph{(1)} for all $1\leq i\leq k$ and $C\subseteq V$ such that $|C|\leq k$,
each $S_{i,C}(\alpha)$ is monotone increasing and continuous as a function of $\alpha$,
and \emph{(2)} for all $1\leq i\leq j\leq n$ and $\alpha\in (\alo,\ahi)$, $S_{i,C}(\alpha)\leq S_{j,C}(\alpha)$,
then the expected number of discontinuities of $\seed_\alpha(\V,\vec{Z},p)$ as a function of $\alpha$ is
$O\left(nkD_p(\ahi-\alo)\right)$.
}

\begin{proof}
Given $\V$ and $[0,\ahi]$, we
will show that $\E[\#I]\leq nk\log n\cdot\ahi$,
where $\#I$ denotes the total number of discontinuities of $\seed_\alpha(\V,\vec{Z})$
and the expectation
is over the randomness $Z\in[0,1]^k$ of the $d^\alpha$-sampling algorithm.
Consider round $t$ of a run of the algorithm.
Suppose at the beginning of round $t$, there are $L$ possible states of the algorithm, e.g., $L$ sets of $\alpha$
such that within a set, the choice of the first $t-1$ centers is fixed.
From the monotonicity property, we can write these sets as
$[\alpha_0,\alpha_1],\dots,[\alpha_{L-1},\alpha_L]$, where $0=\alpha_0<\cdots<\alpha_L=\ahi$.
Recall that
$D_{p}=\max_{i,C,v,\alpha}\left(\frac{\partial S_{i,C}(\alpha)}{\partial \alpha}\right)$.
Given one interval, $[\alpha_{\ell},\alpha_{\ell+1}]$,
we claim the expected number of new breakpoints $\#I_{t,\ell}$ by choosing a center in round $t$ is bounded by $nD_p(\alpha_{\ell+1}-\alpha_{\ell})$.
Note that $\#I_{t,\ell}+1$ is the number of possible choices for the next center in round $t$ using $\alpha$ in $[\alpha_{\ell},\alpha_{\ell+1}]$.

The claim gives an upper bound on the expected number of new breakpoints,
where the expectation is only over $z_t$ (the uniformly random draw from $[0,1]$ used by
the initialization algorithm in round $t$),
and the bound holds for any set of centers and points.
Assuming the claim, we can finish off the proof by using linearity of expectation as follows.

\setcounter{figure}{2}
\begin{figure}[ht]
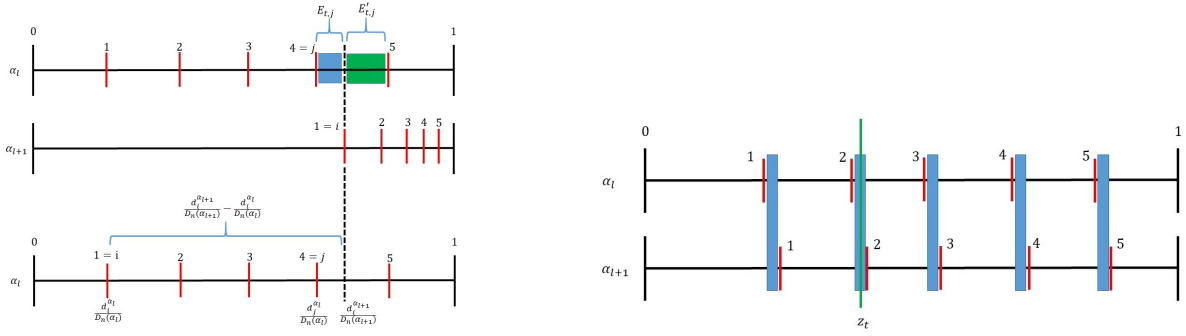

  \subfigure{\includegraphics[width=.5\linewidth]{case2.jpg}}
  \subfigure{\includegraphics[width=.5\linewidth]{epsilon-interval.jpg}}
 % \vspace{-1em}
  \caption{(Repeated) Definition of $E_{t,j}$ and $E_{t,j}'$, and details for bounding $j-i$ (left).
	Intuition for bounding $P(E_{t,j})$, where the blue regions represent $E_{t,j}$ (right).
	}
  \label{fig:case2_}
\end{figure}
\setcounter{figure}{7}

\begin{align*}
E_{Z\in [0,1]^k}[\#I]&\leq E_{Z\in [0,1]^k} \left[\sum_{t=1}^k\sum_{\ell=1}^L (\#I_{t,\ell})\right]\\
&\leq \sum_{t=1}^k \sum_{\ell=1}^L E_{Z\in [0,1]^k}[\#I_{t,\ell}]\\
&\leq \sum_{t=1}^k \sum_{\ell=1}^L nD_p(\alpha_{\ell+1}-\alpha_{\ell})\\
&\leq nk D_p\cdot\ahi\\
\end{align*}

Now we will prove the claim.
Recall that for $\alpha$-sampling, each point $v_i$ receives an interval along $[0,1]$ of size
proportional to $p_\alpha(v_i,C)$
$\frac{d_x^\alpha}{D_n(\alpha)}$, so
the number of breakpoints in round $t$
along $[\alpha_{\ell},\alpha_{\ell+1}]$ corresponds to the
number of times $z_t$ switches intervals as we increase $\alpha$ from $\alpha_{\ell}$ to
$\alpha_{\ell+1}$.
Note that $S_{i,C}(\alpha)$ corresponds to the division boundary between the interval of
$v_i$ and $v_{i+1}$.
By assumption, these divisions
are monotone increasing, so the number of breakpoints is exactly
$x-y$, where $x$ and $y$ are the minimum indices s.t.\ $\frac{S_{x,C}(\alpha_{\ell})}{S_{n,C}(\alpha_{\ell})}>z_t$
and $\frac{S_{y,C}(\alpha_{\ell+1})}{S_{n,C}(\alpha_{\ell+1})}>z_t$, respectively.
We want to compute the expected value of $x-y$ for $z_t$ uniform in $[0,1]$ (here, $x$ and $y$ are functions of $z_t$).

We take the approach of analyzing each interval individually.
One method for bounding $E_{z_t\in[0,1]}[x-y]$ is to compute the maximum possible number of
breakpoints for each interval $I_{v_j}$, for all $1\leq j\leq n$.
Specifically, if we let $i$ denote the minimum index such that
$\frac{S_{j,C}(\alpha_{\ell})}{S_{n,C}(\alpha_{\ell})}<\frac{S_{i,C}(\alpha_{\ell+1})}{S_{n,C}(\alpha_{\ell+1})}$, then
\begin{align*}
E[\#I_{t,\ell}]&\leq \sum_{j=1}^n P\left(\frac{S_{j,C}(\alpha_{\ell})}{S_{n,C}(\alpha_{\ell})}<z_t<\frac{S_{j+1,C}(\alpha_{\ell})}{S_{n,C}(\alpha_{\ell})}\right)\cdot (j-i+1)\\
&\leq \sum_{j=1}^n \frac{p_{\alpha_\ell}(v_i,C)}{S_{n,C}(\alpha_\ell)}\cdot (j-i+1).
\end{align*}
In this expression, we are using the worst case number of breakpoints within each bucket, $j-i+1$.

We cannot quite use this expression to obtain our bound; for example, when $\alpha_{\ell+1}-\alpha_\ell$ is extremely small, $j-i+1=1$,
so this expression will give us $E[\#I_{t,\ell}]\leq 1$ over $[\alpha_\ell,\alpha_{\ell+1}]$, which is not sufficient to prove the claim.
Therefore, we give a more refined analysis by further breaking into cases based on
whether $z_t$ is smaller or larger than $\frac{S_{i,C}(\alpha_{\ell+1})}{S_{n,C}(\alpha_{\ell+1})}$.
In case 1, when $z_t>\frac{S_{i,C}(\alpha_{\ell+1})}{S_{n,C}(\alpha_{\ell+1})}$, the number of breakpoints is $j-i$, and we will show that
$j-i\leq nD_p(\alpha_{\ell+1}-\alpha_\ell)$.
In case 2, when $z_t<\frac{S_{i,C}(\alpha_{\ell+1})}{S_{n,C}(\alpha_{\ell+1})}$, the number of breakpoints is $j-i+1$, but we will show the probability of this case is low.

For case 2, the probability that $z_t$ is in between $\frac{S_{j,C}(\alpha_{\ell})}{S_{n,C}(\alpha_\ell)}$ and $\frac{S_{i,C}(\alpha_{\ell+1})}{S_{n,C}(\alpha_{\ell+1})}$
is $\frac{S_{j,C}(\alpha_{\ell+1})}{S_{n,C}(\alpha_{\ell+1})}-\frac{S_{j,C}(\alpha_{\ell})}{S_{n,C}(\alpha_\ell)}
\leq\frac{S_{j,C}(\alpha_{\ell+1})}{S_{n,C}(\alpha_{\ell+1})}-\frac{S_{j,C}(\alpha_{\ell})}{S_{n,C}(\alpha_\ell)}$.
Therefore, we can bound this quantity by bounding the derivative $\left|\frac{\partial}{\partial\alpha} \left(\frac{S_{j,C}(\alpha)}{S_{n,C}(\alpha)}\right)\right|$,
which is at most $D_p$ by definition.

For case 1, recall that $j-i$ represents the number of intervals between
$\frac{S_{i,C}(\alpha_\ell)}{S_{n,C}(\alpha_{\ell})}$ and $\frac{S_{j,C}(\alpha_\ell)}{S_{n,C}(\alpha_\ell)}$.
Note that the smallest interval in this range is $\frac{p_{\alpha_\ell}(v_i,C)}{S_{n,C}(\alpha_\ell)}$, and
$\frac{S_{j,C}(\alpha_\ell)}{S_{n,C}(\alpha_\ell)}-\frac{S_{i,C}(\alpha_\ell)}{S_{n,C}(\alpha_{\ell})}\leq \frac{S_{i,C}(\alpha_{\ell+1})}{S_{n,C}(\alpha_{\ell+1})}-\frac{S_{i,C}(\alpha_\ell)}{S_{n,C}(\alpha_{\ell})}$.
Therefore, the expected number of breakpoints is at most
$\frac{S_{n,C}(\alpha_\ell)}{p_{\alpha_\ell}(v_j,C)}\cdot\left(\frac{S_{i,C}(\alpha_{\ell+1})}{S_{n,C}(\alpha_{\ell+1})}-\frac{S_{i,C}(\alpha_\ell)}{S_{n,C}(\alpha_{\ell})}\right)$, and we can bound the second half of this fraction by again using the derivative of $S_{i,C}(\alpha)$.

Putting case 1 and case 2 together, we have

%S_{n,C}
%S_{j,C}
%p_{\alpha_\ell}(v_i,C)

\begin{align*}
E[\#I_{t,\ell}]&\leq\sum_j P\left(\frac{S_{i,C}(\alpha_{\ell+1})}{S_{n,C}(\alpha_{\ell+1})}<z_t<\frac{S_{j+1,C}(\alpha_{\ell})}{S_{n,C}(\alpha_{\ell})}\right)\cdot (j-i)
+\sum_j P\left(\frac{S_{j,C}(\alpha_{\ell})}{S_{n,C}(\alpha_{\ell})}<z_t<\frac{S_{i,C}(\alpha_{\ell+1})}{S_{n,C}(\alpha_{\ell+1})}\right)\cdot (j-i+1)\\
&\leq\sum_j P\left(\frac{S_{j,C}(\alpha_{\ell})}{S_{n,C}(\alpha_{\ell})}<z_t<\frac{S_{j+1,C}(\alpha_{\ell})}{S_{n,C}(\alpha_{\ell})}\right)\cdot (j-i)
+\sum_j P\left(\frac{S_{j,C}(\alpha_{\ell})}{S_{n,C}(\alpha_{\ell})}<z_t<\frac{S_{i,C}(\alpha_{\ell+1})}{S_{n,C}(\alpha_{\ell+1})}\right)\\
&\leq\sum_j \frac{p_{\alpha_\ell}(v_j,C)}{S_{n,C}(\alpha_\ell)}\cdot \frac{S_{n,C}(\alpha_\ell)}{p_{\alpha_\ell}(v_j,C)}\cdot\left(\frac{S_{j,C}(\alpha_{\ell+1})}{S_{n,C}(\alpha_{\ell+1})}-\frac{S_{i,C}(\alpha_\ell)}{S_{n,C}(\alpha_{\ell})}\right)
+\sum_j P\left(\frac{S_{j,C}(\alpha_{\ell})}{S_{n,C}(\alpha_{\ell})}<z_t<\frac{S_{j,C}(\alpha_{\ell+1})}{S_{n,C}(\alpha_{\ell+1})}\right)\\
&\leq\sum_j \left(\frac{S_{i,C}(\alpha_{\ell+1})}{S_{n,C}(\alpha_{\ell+1})}-\frac{S_{i,C}(\alpha_\ell)}{S_{n,C}(\alpha_{\ell})}\right)
+\sum_j \left(\frac{S_{j,C}(\alpha_{\ell+1})}{S_{n,C}(\alpha_{\ell+1})}-\frac{S_{j,C}(\alpha_{\ell})}{S_{n,C}(\alpha_\ell)}\right)\\
&\leq\sum_j D_p (\alpha_{\ell+1}-\alpha_\ell) + \sum_j D_p(\alpha_{\ell+1}-\alpha_\ell)\\
&\leq 2nD_p(\alpha_{\ell+1}-\alpha_\ell)
\end{align*}

This concludes the proof.
\end{proof}

\medskip
\noindent \textbf{Theorem~\ref{thm:gen_runtime} (restated).}
\emph{
Given parameters $0\leq\alo<\ahi$, $\epsilon>0$, a sample $\mathcal{S}$ of size
$$m = O\left(\left(\frac{H}{\epsilon}\right)^2\log \left(\log\frac{\ahi n D_p}{\delta}\right)\right)$$  from $\left(\mathcal{D} \times [0,1]^k\right)^m$,
and an $\alpha$-parameterized family satisfying properties \emph{(1)} and \emph{(2)} from
Theorem~\ref{thm:init}, run Algorithm~\ref{alg:fast} on
each sample and collect all break-points (i.e., boundaries of the intervals
$A_i$). With probability at least $1-\delta$, the break-point $\bar\alpha$ with
lowest empirical cost satisfies
$|\clus_{\bar\alpha,\beta}(\mathcal{S})-\min_{0\leq\alpha\leq\ahi}\clus_{\alpha,\beta}(\mathcal{S})|<\epsilon$.
The total running time to find the best break point is
$O\left(m n^2 k^2 \ahi D_p \log \left( \frac{nH }{\epsilon}\right)\log n\right)$.
}

The proof is almost identical to the proof of Theorem~\ref{thm:runtime}.

\begin{proof}
First we argue that one of the breakpoints output by Algorithm~\ref{alg:fast}
on the sample is approximately optimal.
Formally, denote $\bar\alpha$ as the breakpoint returned by the algorithm with the lowest empirical cost over the sample,
and denote $\alpha^*$ as the value with the minimum true cost over the distribution.
We define $\hat\alpha$ as the empirically optimal value over the sample.
We also claim that for all breakpoints $\alpha$, there exists a breakpoint $\hat\alpha$ outputted by Algorithm~\ref{alg:fast}
such that $|\alpha-\hat\alpha|<\frac{\epsilon}{5n^2kL\log n}$.
We will prove this claim at the end of the proof.
Assuming the claim is correct, we denote $\alpha'$ as a breakpoint outputted by the algorithm such that
$|\hat\alpha-\alpha'|<\frac{\epsilon}{5n^2kL\log n}$.

For the rest of the proof, denote
$\underset{V \sim \mathcal{D}}{\E}\left[\clus_{\alpha,\beta}\left(V\right)\right]=\text{true}(\alpha)$
and $\frac{1}{m} \sum_{i = 1}^m \clus_{\alpha,\beta}\left(V^{(i)},\vec{Z}^{(i)}\right)=\text{sample}(\alpha)$
since beta, the distribution, and the sample are all fixed.

By construction, we have $\text{sample}(\hat\alpha)\leq \text{sample}(\alpha^*)$
and $\text{sample}(\bar\alpha)\leq \text{sample}(\alpha')$.
By Theorem~\ref{thm:rademacher}, with probability $>1-\delta$, for all $\alpha$ (in particular, for $\bar\alpha,~\hat\alpha,~\alpha^*,$ and $\alpha'$), we have
$\left|\text{sample}(\alpha)-\text{true}(\alpha)\right|<\epsilon/5$.
Finally, by Lemma~\ref{lem:net}, we have
$$|\hat\alpha-\alpha'|<\frac{\epsilon}{5n^2kL\log n}~\implies
~\left|\text{true}(\hat\alpha)-\text{true}(\alpha')\right|<\epsilon/5.$$

Using these five inequalities for $\alpha',~\hat\alpha,~\bar\alpha,$ and $\alpha^*$, we can show the desired outcome as follows.
\begin{align*}
\text{true}(\bar\alpha)-\text{true}(\alpha^*)
&\leq\left(\text{true}(\bar\alpha)-\text{sample}(\bar\alpha)\right)+\text{sample}(\bar\alpha)
-\left(\text{true}(\alpha^*)-\text{sample}(\alpha^*)\right)-\text{sample}(\alpha^*)\\
&\leq\epsilon/5+\text{sample}(\alpha')+\epsilon/5-\text{sample}(\hat\alpha)\\
&\leq \left(\text{sample}(\alpha')-\text{true}(\alpha')\right)+\left(\text{true}(\alpha')-\text{true}(\hat\alpha)\right)
+\left(\text{true}(\hat\alpha)-\text{sample}(\hat\alpha)\right)+\frac{2\epsilon}{5}\\
&\leq\epsilon.
\end{align*}

\begin{comment}
\begin{align*}
\underset{V \sim \mathcal{D}}{\E}\left[\clus_{\bar\alpha,\beta}\left(V\right)\right]-\underset{V \sim \mathcal{D}}{\E}\left[\clus_{\alpha^*,\beta}\left(V\right)\right]
&\leq\left(\underset{V \sim \mathcal{D}}{\E}\left[\clus_{\bar\alpha,\beta}\left(V\right)\right]-\frac{1}{m} \sum_{i = 1}^m \clus_{\alpha,\beta}\left(V^{(i)},\vec{Z}^{(i)}\right)\right)
+\frac{1}{m} \sum_{i = 1}^m \clus_{\alpha',\beta}\left(V^{(i)},\vec{Z}^{(i)}\right)
-\left(\underset{V \sim \mathcal{D}}{\E}\left[\clus_{\alpha^*,\beta}\left(V\right)\right]-\frac{1}{m} \sum_{i = 1}^m \clus_{\alpha^*,\beta}\left(V^{(i)},\vec{Z}^{(i)}\right)\right)
-\frac{1}{m} \sum_{i = 1}^m \clus_{\alpha^*,\beta}\left(V^{(i)},\vec{Z}^{(i)}\right)\\
%
&\leq\frac{2\epsilon}{5}
+\left(\frac{1}{m} \sum_{i = 1}^m \clus_{\alpha',\beta}\left(V^{(i)},\vec{Z}^{(i)}\right)-\underset{V \sim \mathcal{D}}{\E}\left[\clus_{\alpha',\beta}\left(V\right)\right]\right)
+\left(\underset{V \sim \mathcal{D}}{\E}\left[\clus_{\alpha',\beta}\left(V\right)\right]-\underset{V \sim \mathcal{D}}{\E}\left[\clus_{\hat\alpha,\beta}\left(V\right)\right]\right)
+\left(\underset{V \sim \mathcal{D}}{\E}\left[\clus_{\hat\alpha,\beta}\left(V\right)\right]-\frac{1}{m} \sum_{i = 1}^m \clus_{\hat\alpha,\beta}\left(V^{(i)},\vec{Z}^{(i)}\right)\right)\\
&\leq\epsilon.
\end{align*}
\end{comment}

Now we will prove the claim that for all breakpoints $\alpha$, there exists a breakpoint $\hat\alpha$ outputted by Algorithm~\ref{alg:fast}
such that $|\alpha-\hat\alpha|<\frac{\epsilon}{5n^2kL\log n}$. Denote $\epsilon'=\frac{\epsilon}{5n^2kL\log n}$.
We give an inductive proof.
Recall that the algorithm may only find the values of breakpoints up to additive error $\epsilon'$, since the true breakpoints may be
irrational and/or transcendental.
Let $\hat T_t$ denote the execution tree of the algorithm after round $t$,
and let $T_t$ denote the \emph{true} execution tree on the sample.
That is, $T_t$ is the execution tree as defined earlier this section,
$\hat T_t$ is the execution tree with the algorithm's $\epsilon'$ imprecision on the values of alpha.
Note that if a node in $T_t$ represents an alpha-interval of size smaller than $\epsilon'$, it is possible that $\hat T_t$
does not contain the node. Furthermore, $\hat T_t$ might contain spurious nodes with alpha-intervals of size smaller than $\epsilon'$.

Our inductive hypothesis has two parts.
The first part is that for each breakpoint $\alpha$ in $T_t$, there exists
a breakpoint $h(\alpha)$ in $\hat T_t$ such that $|\alpha-h(\alpha)|<\epsilon'$.
For the second part of our inductive hypothesis, we define
$B_t=\bigcup_{\alpha\text{ breakpoint}}\left([\alpha,h(\alpha)]\cup[h(\alpha),\alpha]\right)$,
the set of ``bad'' intervals. Note that for each $\alpha$, one of $[\alpha,h(\alpha)]$ and $[h(\alpha),\alpha]$ is empty.
Then define $G_t=[\alo,\ahi]\setminus B_t$, the set of ``good'' intervals.
The second part of our inductive hypothesis is that the set of centers for $\alpha$ in $T_t$ is the same as in $\hat T_t$, as long as $\alpha\in G_t$.
That is, if we look at the leaf in $T_t$ and the leaf in $\hat T_t$ whose alpha-intervals contain $\alpha$,
the set of centers for both leaves are identical.
Now we will prove the inductive hypothesis is true for round $t+1$, assuming it holds for round $t$.
Given $T_t$ and $\hat T_t$, consider a breakpoint $\alpha$ from $T_{t+1}$ introduced in round $t+1$.

Case 1: $\alpha\in G_t$.
Then the algorithm will recognize there exists a breakpoint, and use binary search to output a value $h(\alpha)$ such that $|\alpha-h(\alpha)|<\epsilon'$.
The interval $[\alpha,h(\alpha)]\cup[h(\alpha),\alpha]$ is added to $B_{t+1}$, but the good intervals to the left and right of this interval still have the correct centers.

Case 2: $\alpha\in B_t$. Then there exists an interval $[\alpha',h(\alpha')]\cup [h(\alpha'),\alpha]$ containing $\alpha$.
By assumption, this interval is size $<\epsilon'$, therefore, we set $h(\alpha')=h(\alpha)$, so there is a breakpoint within $\epsilon'$ of $\alpha$.

Therefore, for each breakpoint $\alpha$ in $T_{t+1}$, there exists a breakpoint $\hat \alpha$ in $\hat T_{t+1}$ such that $|\alpha-\hat\alpha|<\epsilon'$.
Furthermore, for all $\alpha\in G_{t+1}$, the set of centers for $\alpha$ in $T_{t+1}$ is the same as in $\hat T_{t+1}$.
This concludes the inductive proof.

Now we analyze the runtime of Algorithm~\ref{alg:fast}.
Let $(C,A)$ be any node in the algorithm, with centers $C$ and alpha
interval $A = [\alo, \ahi]$. Sorting the points in $\V$
according to their distance to $C$ has complexity $O(n \log n)$. Finding the
points sampled by $d^\alpha$-sampling with $\alpha$ set to $\alo$ and
$\ahi$ costs $O(n)$ time. Finally, computing the alpha interval $A_i$
for each child node of $(C,A)$ costs $O(n \log\frac{nH}{\epsilon})$ time, since we need
to perform $\log\frac{nkH\log n}{\epsilon}$ iterations of binary search on $\alpha
\mapsto \frac{D_i(\alpha)}{D_n(\alpha)}$
and each evaluation of the function costs $O(n)$ time.
We charge this $O(n \log\frac{nH}{\epsilon})$ time to the corresponding child node.
If there are $N$ nodes in the execution tree, summing this cost over all nodes
gives a total running time of $O(N \cdot n \log\frac{nH}{\epsilon}))$.
If we let $\#I$ denote the total number of $\alpha$-intervals for $\V$,
then each layer of the execution tree has at most $\#I$ nodes, and the depth is
$k$, giving a total running time of $O(\#I \cdot k n \log\frac{nH}{\epsilon})$.

From Theorem \ref{thm:expect}, we have $\E[\#I]\leq 8nk\log n\cdot \ahi$.
Therefore, the expected runtime of Algorithm \ref{alg:fast} is
 $O\left(n^2 k^2 \ahi (\log n)\left(\log\frac{nH}{\epsilon}\right)\right)$.
This completes the proof.
\end{proof}

%\input{algo-appendix}

%\section{Additional Experiments} \label{app:experiments}
%\input{experiments-appendix}

\end{document}